\documentclass[sts]{imsart}

\RequirePackage{amsthm,amsmath,amsfonts,amssymb}
\RequirePackage[authoryear]{natbib}
\RequirePackage[colorlinks,citecolor=blue,urlcolor=blue]{hyperref}
\RequirePackage{graphicx}

\usepackage[utf8]{inputenc}
\usepackage{amsmath}
\usepackage{amsfonts,cancel}
\usepackage[margin=1in]{geometry}
\usepackage{framed}
\usepackage{tikz}
\usepackage{listings}
\usepackage{graphicx}
\graphicspath{{figs/}}
\usepackage{bbm}
\usepackage{amsthm}
\usepackage[normalem]{ulem}

\startlocaldefs

\theoremstyle{plain}
\newtheorem{theorem}{Theorem}[section]

\newtheorem{proposition}[theorem]{Proposition}
\theoremstyle{definition}
\newtheorem{definition}[theorem]{Definition}

\newtheorem{remark}[theorem]{Remark}
\numberwithin{equation}{section}
\usepackage{anyfontsize}
\usepackage{amssymb}
\usepackage{mathrsfs}
\usepackage{amsthm}
\def\E{\mathbb E}
\def\N{\mathbb N}
\def\C{\mathbb C}
\def\I {\mathbbm i}
\def\ii {\mathbbm i}
\def\P{\mathbb P}
\def\R{\mathbb R}

\usepackage{mathtools}

\newcommand{\whittW}{W}
\newcommand{\besselJ}{J}
\newcommand{\besselI}{I}
\newcommand{\besselK}{\mathcal{K}}
\newcommand{\struveL}{L}

\newcommand{\whittM}{M}
\newcommand{\struveA}{A}

\newcommand{\us}{\mathcal{S}}
\newcommand{\ang}{\varphi}

\usepackage{bbold}
\usepackage[utf8]{inputenc}
\usepackage{dirtytalk}
\usepackage{float}

\def\T{\mathbb T}

\DeclareUnicodeCharacter{2014}{\dash}
\usepackage{hyperref}
\usepackage{url}
\usepackage{natbib}

\allowdisplaybreaks

\newcommand{\Cov}{\textrm{Cov}}

\newcommand{\mb}[1]{\boldsymbol{#1}}
\newcommand{\Y}{\boldsymbol{Y}}
\newcommand{\ti}{\boldsymbol{s}}
\newcommand{\tilc}{s}
\newcommand{\cf}{\boldsymbol{C}}
\newcommand{\hv}{\boldsymbol{h}}
\newcommand{\Nu}{\boldsymbol{\nu}}
\newcommand{\measurenb}{\mu}

\newcommand{\measure}{\mb{\measurenb}}
\newcommand{\rmeasure}{\mb{B}_{\measure}}
\newcommand{\sd}{\measure}
\newcommand{\sde}{\measurenb}
\newcommand{\Imat}{\mb{I}}
\newcommand{\amat}{\mb{a}}
\newcommand{\Sig}{\mb{\Sigma}}
\newcommand{\sig}{\sigma}

\newcommand{\comment}[1]{\fbox{{\color{red}Commented out material.}}}

\DeclareMathAlphabet{\mathbb}{U}{msb}{m}{n}

\newcommand{\add}[1]{{#1}}

\endlocaldefs

\begin{document}

\begin{frontmatter}

\title{Multivariate Mat\'ern Models -- A Spectral Approach}

\runtitle{Multivariate Mat\'ern Models}

\begin{aug}

\author[A]{\fnms{Drew}~\snm{Yarger}\thanksref{t1}\ead[label=e1]{anyarger@purdue.edu}},
\author[B]{\fnms{Stilian}~\snm{Stoev}\thanksref{t2}\ead[label=e2]{sstoev@umich.edu}},
\and
\author[C]{\fnms{Tailen}~\snm{Hsing}\thanksref{t2}\ead[label=e3]{thsing@umich.edu}}

\thankstext{t1}{Partially supported by NSF Grant DGE-1841052}
\thankstext{t2}{Supported by NSF Grant DMS-1916226}

\address[A]{Drew Yarger is Assistant Professor, Department of Statistics, Purdue University,
West Lafayette, IN 47907\printead[presep={\ }]{e1}.}

\address[B]{Stilian Stoev is Professor, Department of Statistics, University of Michigan,
Ann Arbor, MI 48109\printead[presep={\ }]{e2}.}

\address[C]{Tailen Hsing is Professor, Department of Statistics, University of Michigan,
Ann Arbor, MI 48109\printead[presep={\ }]{e3}.}

\end{aug}

\begin{abstract}

%
%
The classical Mat\'ern model has been a staple in spatial statistics. Novel data-rich applications in environmental and 
physical sciences, however, call for new, flexible vector-valued spatial and space-time models. Therefore, the extension 
of the classical Mat\'ern model has been a problem of active theoretical and methodological interest.  
In this paper, we offer a new perspective to extending the Mat\'ern covariance 
model to the vector-valued setting. 
We adopt a spectral, stochastic integral approach, which allows us to address challenging issues on the validity of 
the covariance structure and at the same time to obtain new, flexible, and interpretable models.  In particular, our 
multivariate extensions of the Mat\'ern model allow for \add{asymmetric} covariance structures.  
Moreover, the spectral approach provides an essentially complete flexibility in modeling the local structure of the process.  
We establish closed-form representations of the cross-covariances when available, compare them with existing models, simulate Gaussian instances of these new processes, and demonstrate estimation of the model's parameters through maximum likelihood.  
An application of the new class of multivariate Mat\'ern models to data indicate their success in capturing 
inherent covariance-asymmetry phenomena.
\end{abstract}


\end{frontmatter}

\section{Introduction and main ideas}\label{sec:mm_intro}

In the past two decades, there have been considerable interests in modeling vector-valued spatial processes, especially in the context of environmental data \citep[cf.][]{genton_cross_covariance_2015}. 
One popular model for such processes is the multivariate Mat\'ern model, which extends the Mat\'ern model to the multivariate case and was first introduced in \cite{gneiting_matern_2010}.
Extensions, improvements, and analysis of multivariate Mat\'ern models are numerous \citep[for more information and examples, see][]{porcu_matern_2023, alegria_bivariate_2021, apanasovich_valid_2012, emery_new_2022,genton_cross_covariance_2015, kleiber_coherence_2018}.
Section 3 of \cite{genton_cross_covariance_2015} and Section 5.2 of \cite{porcu_matern_2023} review multivariate Mat\'ern models.

Let $\mb{Y}(\mb{s}) = (Y_i(\mb{s}))_{i=1}^p$ denote a zero-mean stochastic process with finite variance taking values in $\R^p,\ p\ge 1$, and indexed by $\mb{s}\in\R^d$.  
Unless stated otherwise, all vectors we consider will be column-vectors, and they will be denoted with boldface letters to distinguish them from scalars.  

Most of the existing multivariate extensions of the Mat\'ern model start by prescribing a Mat\'ern-like cross-covariance function for $\mb{Y}=\{\mb{Y}(\mb{s}),\ \mb{s}\in\R^d\}$, or some modification of one.  
While natural and appealing, this leads to formidable challenges in verifying the validity (positive definiteness) of the resulting matrix-valued auto-covariance. 
One consequence of this approach is that many of the existing multivariate Mat\'ern models are
covariance-symmetric or reversible, i.e., $\E[ \mb{Y}(\mb{s}) \mb{Y}(\mb{t})^\top] = \E[ \mb{Y}(\mb{t}) \mb{Y}(\mb{s})^\top]$,  for all $\mb{s},\, \mb{t} \in \R^d$.
While this property is clearly automatic in the case when $\mb{Y}(\mb{t})$ is scalar-valued ($p=1$), it is an exception rather than the norm in the vector-valued case ($p\ge 2$).  
This observation shows that the existing covariance-symmetric multivariate Mat\'ern models are rather restrictive. 
Lastly, while the classical Mat\'ern is well understood, the local behavior of the multivariate Mat\'ern-type models is largely unexplored.  
We will give a more comprehensive overview of the existing literature on multivariate Mat\'ern models in Section \ref{sec:new_Matern}.

In this paper, we aim to address the above challenges by adopting a stochastic integral perspective to the construction of the multivariate Mat\'ern models.  This will allow us to automatically obtain valid covariance structures with interpretable parameterizations and obtain a rich family of vector-valued models including symmetric as well as asymmetric covariance structures.

Recall that $\mb{Y}$ is said to be second-order stationary if its covariance function
 $$
 \mb{C}_{\mb{Y}}(\mb{h}) :=\E [\mb{Y}(\mb{h})\mb{Y}(\mb{0})^\top]  
 $$
is defined and depends only on the lag $\mb{h}\in \R^d$. Note that for $p\ge 2$, the function $\mb{h} \mapsto \mb{C}_{\mb{Y}}(\mb{h})$ is $p\times p$-matrix valued, and 
 as in the classical scalar-valued setting ($p=1$), is positive (semi)definite in the sense that
 \begin{equation}
   \label{eq:psd-C}\sum_{i,j=1}^n a_i a_j \mb{C}_{\mb{Y}}(\mb{s}_i - \mb{s}_j) 
 \end{equation}
 is a positive semi-definite matrix for all $a_i\in \R$ and points $\mb{s}_i \in \mathbb{R}^d$, $i=1, \dots, n$. 
 A more detailed review of the spectral theory for vector-valued second-order stationary processes is
 given in Section \ref{sec:mmm}, below.  
 
To illustrate ideas, consider for a moment the scalar-valued case $p=1$.  If $Y = \{Y(\mb{s}),\ \mb{s}\in\R^d\}$ is zero-mean, second-order stationary and $L^2-$continuous, a classical result due to \cite{cramer1942harmonic} yields the following {\em spectral representation} of $Y$ as
 a {\em stochastic integral:}
 \begin{align}\label{eq:eta_gen}
  Y(\mb{s}) = \int_{\R^d} e^{\I \langle\mb{s}, \mb{x}\rangle} \xi(d\mb{x}),
\end{align}
which is valid almost surely for each $\mb{s}\in \R^d$, where the integration is with respect to a random, complex-valued measure $\xi$ with mean zero and orthogonal increments.  Namely, we have that $\E[ \xi(A)]  = 0$, 
$\xi(A \cup B) = \xi(A) + \xi(B)$, for $A\cap B=\emptyset$, and for a finite Borel measure $F_Y$ on $\mathbb R^d$, we have
$
\E[\xi(A) \overline{\xi(B)} ] = F_Y(A\cap B),$  for all Borel sets $A, B\subset \R^d$ (cf.\ Section \ref{sec:mmm} below).
Relation \eqref{eq:eta_gen} readily implies that the auto-covariance of $Y$ is
\begin{equation}\label{e:C_Y}
C_Y(\mb{h}) = \E[Y(\mb{t}+\mb{h}) Y(\mb{t})] = \int_{\R^d} e^{\I \langle \mb{h},\mb{x}\rangle } F_Y(d\mb{x})
\end{equation}
for $\mb{h}, \mb{t} \in \R^d$, so that $F_Y(d\mb{x}) = \E[ |\xi(d\mb{x})|^2]$ is precisely the spectral measure of the process $Y,$ per the
classical Bochner Theorem (see Theorem \ref{thm:Bochner-Neeb}).

Since the process $Y$ in \eqref{eq:eta_gen} is real-valued, the random measure $\xi$ is necessarily 
Hermitian, i.e., $\xi(-A) = \overline{\xi(A)}$, and the spectral measure $F_Y$ is real and symmetric, i.e., $\overline{F_Y(-A)} = F_Y(A)$ (Proposition \ref{p:Y-real}). 
If $F_Y$ has a density  $f_Y$ with respect to the Lebesgue measure, then $f_Y$ is referred to as the {\em spectral density} of the process $Y.$  

We illustrate these fundamental notions with the classical Mat\'ern covariance in $\R^d,\ d\ge 1$. 
Recall that a scalar-valued second-order stationary process $Y=\{Y(\mb{s}),\ \mb{s}\in\R^d\}$ is said to follow the Mat\'ern model 
if its auto-covariance and spectral measures are respectively:
\begin{align}\label{eq:matern_form}\begin{split}
C_Y(\mb{h})&= {\cal M}(\|\mb{h}\|; a, \nu, \sigma^2)\\&:= \frac{\sigma^2  2^{1-\nu} }{\Gamma(\nu)}  ( a\|\mb{h}\|) ^{\nu} \besselK_\nu(a\|\mb{h}\|),
    \end{split}\\
 f_{d,(a,\nu,\sigma^2)}(\mb{x}) d\mb{x} &=
 \frac{\sigma^2 c(d, \nu, a)^2}{\left(a^2 + \|\mb{x}\|^2\right)^{\nu +\frac{d}{2}}} d\mb{x},
 \label{eq:matern_spectral_density}
\end{align}
where $\|\cdot\|$ denotes the Euclidean norm in $\R^d$, $c(d, \nu, a)^2 = a^{2\nu}\Gamma(\nu + d/2)/(\pi^{d/2}\Gamma(\nu))$, $\Gamma(\cdot)$ is 
the Euler gamma function, and $\besselK_\nu(\cdot)$ is the modified Bessel function of the second kind with order $\nu$ \citep[page 48,][]{stein_interpolation_2013}.
The parameterization is such that $\sigma^2 = C_Y(\mb{0}) = {\rm Var}(Y(\mb{t}))$ is the marginal variance of the model,
$a>0$ is the inverse range parameter, and $\nu>0$ is the shape or smoothness parameter.  
Observe that $C_Y(\mb{h})$ is isotropic and valid, i.e., positive semi-definite
{\em for every} $d\ge 1$.  The latter property is closely related to an important result due to \cite{schoenberg:1938}, which we discuss 
in the Supplement where we also give a simple derivation of \eqref{eq:matern_spectral_density} from \eqref{eq:matern_form}.

In the rest of this section, we sketch the main ideas behind our approach for $d=1$.
A more systematic exposition for all $d$ is given in Section \ref{sec:spatial}.   We begin with the scalar-valued
case of $p=1$.  

Consider a complex Gaussian measure $B(dx)$ on $\R$ with orthogonal increments such that 
$B(dx) = \overline{B(-dx)}$ and $\E[ |B(dx)|^2] = \sigma^2\cdot dx$, where $dx$ stands for Lebesgue measure and $\sigma>0$. 
Letting 
\begin{align}
\xi(dx) := c(1, \nu, a) \left(a^2 + x^2\right)^{\left(-\nu - \frac{1}{2} \right)/2} B(dx)\label{eq:eta_bolin},
\end{align}
in view of \eqref{eq:matern_spectral_density}, we obtain $\E \left[|\xi(dx)|^2\right] = f_{1, (a, \nu, \sigma^2)}(x)dx$. Hence by 
\eqref{e:C_Y} the stochastic integral in \eqref{eq:eta_gen} defines a stationary Gaussian process with the Mat\'ern 
auto-covariance ${\cal M}(\cdot; a, \nu, \sigma^2)$. 
Alternatively, since 
\begin{align*}
&\left|(a \pm \I x)^{-\nu - \frac{1}{2}}\right|^2=(a \pm \I x)^{-\nu - \frac{1}{2}} \overline{(a \pm \I x)^{-\nu - \frac{1}{2}}}\\
&=(a +  \I x)^{-\nu - \frac{1}{2}}(a - \I x)^{-\nu - \frac{1}{2}} =[(a +  \I x)(a - \I x)]^{-\nu - \frac{1}{2}} \\
&=\left(a^2 + x^2\right)^{-\nu - \frac{1}{2}},
\end{align*}
where $\pm$ is either $+$ or $-$, one could take
\begin{align}
 \xi(dx) =  c(1, \nu, a)\left(a \pm \I x\right)^{-\nu - \frac{1}{2}}B(dx).\label{eq:eta_ours}
\end{align}
This choice also yields $\E[|\xi(dx)|^2] = f_{1, (a, \nu, \sigma^2)}(x)dx$ and hence the stochastic integral
\eqref{eq:eta_gen} defines a processes with {\em the same} Mat\'ern covariance for $d=1$. 
We will primarily focus on this second representation in \eqref{eq:eta_ours}, since it will naturally lead to general complex-valued spectral densities and potentially asymmetric models for $p\ge 2$ (in Sections \ref{Introduce_model} and \ref{sec:spatial}).  
The restriction $a > 0$ ensures that powers of the complex number $a \pm \I x$ are well-defined by avoiding the singularity at the origin of the complex plane. 

We now turn to multivariate Mat\'ern, that is, where $p > 1$. Motivated by the Cram\'er stochastic integral representation \eqref{eq:eta_gen} of the 
classical Mat\'ern process with \eqref{eq:eta_ours}, we consider 
\begin{align}\label{eq:Y-d=1-def}
 \mb{Y}(s)=\int_{\R} e^{\I s x} \mb{\xi}(dx), \mbox{and } \mb{\xi}(dx) := \mb{P}(x)\mb{B}(dx).
\end{align}
Now, however, the term $(a+\I x)^{-\nu-1/2}$ in \eqref{eq:eta_ours} is replaced by a $\mathbb{C}^{p\times p}$-valued function 
$\mb{P}(x)$, while $\mb{B}(dx)$ is a $\mathbb{C}^p$-vector-valued counterpart to the random measure $B(dx)$.
Specifically, we consider
$$\mb{P}(x) := \textrm{diag}(c_j\cdot (a_j + \I x)^{-\nu_j - \frac{1}{2}}, j=1, \dots, p) \in \mathbb{C}^{p\times p},$$
with normalizing constants $c_j := c(1, \nu_j, a_j)$ and parameters $a_j >0$ and $\nu_j>0,\ j=1,\cdots,p$. 

On the other hand, the {\em vector-valued} zero-mean random measure $\mb{B}(dx) \in \mathbb{C}^p$ 
in \eqref{eq:Y-d=1-def} is assumed to be Hermitian $\mb{B}(dx) = \overline{\mb{B}(-dx)}$, to have orthogonal increments (cf.\ Definition 
\ref{def:xi-measure}), where 
\begin{align*}
\sd(dx) :&= \mathbb{E}\left[\mb{B}(dx)\overline{\mb{B}(dx)}^\top\right] 
         \\
         &= \left(\Re(\mb{\Sigma}_H) + \textrm{sign}(x) \I \Im(\mb{\Sigma}_H)\right) dx
         \end{align*}
for some self-adjoint and positive-semidefinite matrix $\mb{\Sigma}_H = (\sigma_{jk})_{p\times p}$. That is, 
$\mb{\Sigma}_H = \mb{\Sigma}_H^*:= \overline{\mb{\Sigma}_H}^\top$, 
and the quadratic forms $\mb{x}^\top \Sig_H \mb{x} \ge 0$, for all $\mb{x}\in\C^p$. Here, $\Re(\cdot)$ and $\Im(\cdot)$ denote the 
real and imaginary part functions, and $\textrm{sign}(x) = \mathbb{I}(x > 0) - \mathbb{I}(x < 0)$ where $\mathbb{I}(\cdot)$ is the indicator 
function.  

By \eqref{eq:Y-d=1-def}, one can show that
\begin{align*}
\cf_{\mb{Y}} (h) &= \E[ \mb{Y}(t+h)\mb{Y}(t)^\top] 
 \\&= \int_\mathbb{R} e^{\I hx} \mb{P}(x) \sd(dx) \overline{\mb{P}(x)}
\\ &= \int_\mathbb{R} e^{\I hx} \mb{F}(dx), 
\end{align*}
where $\mb{F}(dx) = (F_{jk}(dx))_{p\times p}:=\mb{P}(x) \sd(dx) \overline{\mb{P}(x)}$ is now the matrix-valued {\em spectral 
measure} of $\mb{Y}$. (The definition and integration with respect to the measure 
$\mb{F}(dx)$ is discussed in Section \ref{sec:mmm}.)

The spectral measure $\mb{F}$ takes values in the space of all self-adjoint and positive-semidefinite
${p\times p}$ matrices.
In contrast with the scalar case ($p=1$), however, the spectral measure $\mb{F}$ may have complex components. 
In fact, $\mb{F}$ is real if and only if the process $\mb{Y}$ is covariance-symmetric (Proposition \ref{p:Y-real-reversible}).  

To gain some intuition and to connect to some of the existing multivariate Mat\'ern models, 
let $1\le j,k\le p$, and consider the cross-covariances: 
\begin{align}
C_{jk}(h) &= [\mb{C}_{\mb{Y}}(h)]_{jk} =  {\rm Cov}(Y_j(t+h),Y_k(t)) \nonumber\\ \begin{split}
&= 
 c_jc_k\int_{\mathbb{R}}e^{\I hx} (a_j + \I x)^{-\nu_j - \frac{1}{2}} \\
 &~~~~~~~~~~~~~~~\times(\Re(\sigma_{jk}) + \textrm{sign}(x)\I\Im(\sigma_{jk})) \\
 &~~~~~~~~~~~~~~~\times (a_k - \I x)^{-\nu_k - \frac{1}{2}}dx.
 \end{split}\label{eq:new_cc}
\end{align}
This shows that the univariate components of $\mb{Y}(s) = (Y_j(s))_{j=1}^p,\ s\in \R$, have Mat\'ern auto-covariances. 
Indeed, since $\Sig_H = \overline{\Sig_H}^\top$, we have $\Im(\sigma_{jj})=0$, and in fact $\Re(\sigma_{jj}) = \sigma_{jj}\ge 0$ by the
positivity of $\Sig_H$.  Thus, in view of \eqref{eq:new_cc}, the process $\{Y_j(s),\ s\in \R\}$ has precisely the Mat\'ern spectral 
density in \eqref{eq:matern_spectral_density} with parameters $a_j$, $\nu_j$, and $\sigma_{jj}$, provided $c_j := c(1, \nu_j, a_j)$.  
(Note that 
the diagonal terms $\sigma_{jj}$ now refer to marginal variance parameters.) 

Having non-zero off-diagonal entries in $\Sig_H$ is a simple and yet flexible way to model the cross-covariance structure 
between the components of $\mb{Y}$. We give next some intuition. In the special case when $\sigma_{jk}$ is real and 
$\nu_j=\nu_k=\nu$, $a_j=a_k=a$, and $c_j=c_k=c(1,\nu,a)$, for 
example, we obtain
\begin{align*}
 {\rm Cov}(Y_j(t), Y_k(s)) &= \sigma_{jk}\cdot c(1,\nu,a)^2 \\&
 ~~~~~~\times \int_{\R}  e^{\I (t-s) x} (a^2 +  x^2)^{-\nu - \frac{1}{2}} dx\\
& = \sigma_{jk} {\cal M}(|t-s|;a,\nu,1),
\end{align*}
so that $\sigma_{jk}$ may be interpreted as the (lag-independent) cross-covariance between the one-dimensional Mat\'ern processes $Y_j$ and $Y_k$.  In this case, the bi-variate process $\{(Y_j(t),Y_k(t))^\top,\ t\in\R\}$ has a real spectral measure, and it is covariance-symmetric, i.e., ${\rm Cov}(Y_j(t),Y_k(s)) = {\rm Cov}(Y_j(s),Y_k(t))$. 
Even if $\sigma_{jk}$ is real, however, when $\nu_j \not=\nu_k$ or $a_j \neq a_k$, we obtain that 
${\rm Cov}(Y_j(t),Y_k(s)) \not\equiv  {\rm Cov}(Y_j(s),Y_k(t))$, for all $s,t\in\R$.  
The closed-form expression of the resulting asymmetric cross-covariances is established in Theorem \ref{thm:re_ent} below in terms of a Whittaker special function.  
The possibility to have complex $\sigma_{jk}$'s adds even more flexibility to the model and the closed-form expression of the cross-covariances in the purely imaginary case ($\Re(\sigma_{jk}) = 0$), is established in
Theorem \ref{thm:im_ent} below.  
It involves the modified Struve and modified Bessel functions of the first kind. 
By combining the cases of real and imaginary $\sigma_{jk}$'s one obtains the general closed-form of the cross-covariances of this novel class of models (cf.\ Section \ref{sec:review}, below).

For general $d$, our construction results in anisotropic cross-covariances that are asymmetric. 
In much of spatial statistics, there is a focus on geometric anisotropy, where, in the scalar case, $C(\mb{h}) = C_{iso}(\sqrt{\mb{h}^\top \mb{Q} \mb{h}})$ for a positive definite matrix $\mb{Q}$ and isotropic covariance function $C_{iso}(\cdot)$ \citep{chiles2012geostatistics}. 
There is also an interest in zonal isotropy, where the covariance only depends on some elements of $\mb{h}$ \citep{chiles2012geostatistics}.
Here, we introduce cross-covariances that have more general and complex anisotropic structure. 


Our work uses the form of \eqref{eq:eta_ours} to construct multivariate Mat\'ern models, while 
\cite{bolin_multivariate_nodate}, \add{as well as \cite{terres2018bayesian} and \cite{guinness2014multivariate},} used instead \eqref{eq:eta_bolin},  
where they also essentially take $\sd(dx) = \Re(\Sig_H) dx$ in our notation. This results only in 
symmetric (reversible) models with $\mb{C}(h) =  \mb{C}(-h)$ for all $h \in \mathbb{R}$ (see also Section \ref{sec:alt_formal}, below, for more details). 
We advocate for using \eqref{eq:eta_ours} because, as illustrated, one obtains a more flexible family of 
models, and in many cases have closed-form expressions for their cross-covariances.

\medskip
{\em The remainder of the paper is structured as follows.}
In Section \ref{sec:spect_theory}, we review spectral theory and the notion of covariance-symmetry for multivariate processes. 
This section is technical, and some details may be omitted by the reader.
We also review the literature on multivariate Mat\'ern models in more detail. 
In Section \ref{Introduce_model}, we comprehensively address the case of $d=1$ as presented here in the introduction.
First, in Section~\ref{sec:re_ent} we find all closed-form expressions and explore the class of models that have {\it real directional measure} with $\Im(\sig_{jk}) =0$ in \eqref{eq:new_cc}.
We next consider a second class of asymmetric models
with \textit{imaginary} or \textit{complex directional measure}, i.e., $\Im(\sig_{jk}) \neq 0$, in Section~\ref{sec:im_ent}. 

We then introduce a model for general $d \in \{1,2, 3, \dots\}$ in Section~\ref{sec:spatial}. 
Although closed-form representations of the cross-covariances are not often available for $d > 1$, one can readily numerically approximate them using fast Fourier transforms. Furthermore, simulating processes based on these multivariate spatial cross-covariances is straightforward due to \cite{emery_improved_2016}. 

In Section \ref{sec:computation_simulation}, we demonstrate the utility of the model in simulation studies in terms of parameter estimation, spatial prediction, and a likelihood ratio test for testing the null hypothesis $\Im(\sigma_{jk}) = 0$. 
We find that estimating $\Im(\sigma_{jk})$ does not significantly detract from the estimation of $\Re(\sigma_{jk})$, and its estimation is not substantially more challenging than the estimation of $\Re(\sigma_{jk})$.
Furthermore, we find that ignoring the imaginary component $\Im(\sigma_{jk})$ when present results in less accurate predictions as well as detectably worse model fit as determined by the likelihood ratio test. 

In Section \ref{sec:data_analysis}, we estimate the proposed model separately in the context of two datasets: real-estate housing inventory and price data \citep{redfin}, and the Argo temperature data of the ocean at two different depths \citep{argo2020}, studied in \cite{bolin_multivariate_nodate}.
For the real estate data, the likelihood ratio test rejects the null hypothesis that $\sigma_{jk} \in \mathbb{R}$,
implying asymmetry in the cross-covariance, which can also be interpreted in terms of lags in supply and demand of the housing market.
We then outline a number of areas for future research in Section \ref{sec:mm_discussion}. 

All code used to present results in the paper is available at \url{https://github.com/dyarger/multivariate_matern}, where we also 
provide an R Shiny application that computes and plots the cross-covariances introduced here. 
The Appendix and Supplement contain further results on these multivariate Mat\'ern models.
For example, in the Supplement, we describe the tangent processes of the introduced multivariate Mat\'ern models. 
The tangent processes describe the (scaled) local behavior of the multivariate Mat\'ern processes, and, in the Gaussian case, they are operator fractional Brownian fields (OFBFs), an extension of fractional Brownian fields to the vector-valued setting \citep{didier_domain_2018, didier_integral_2011}. 
See \cite{shen_tangent_2022} and references therein for more background on tangent processes. 
We establish that our family of Mat\'ern-type models can realize essentially all possible OFBF tangent fields, thus providing maximal flexibility in the local behavior of the multivariate processes.

\section{Spectral theory and multivariate Mat\'ern models -- an overview}\label{sec:spect_theory}

In this section, for the sake of completeness, we first review some fundamental results on
the spectral representation of vector-valued processes.  We then discuss our contributions in the context of the 
recent and expanding literature on multivariate Mat\'ern and multivariate spatial processes.

\subsection{The multivariate Bochner and Cram\'er theorems} \label{sec:mmm}

We provide a treatment of the mathematical background behind the spectral properties of multivariate, second-order stationary spatial processes $\mb{Y} = \{\mb{Y}(\mb{s}),\ \mb{s}\in\R^d\}$ taking values in $\C^p$ over the field of complex numbers $\C$. 
The classical theorems of Bochner and Cram\'er are at the foundation of spectral analysis
of second-order processes and random fields.  These theorems also have far-reaching extensions to multivariate processes \citep{hannan:1970,yaglom_correlation_1987,gelfand_multivariate_2010} 
and, more generally, processes taking values in separable Hilbert spaces \citep[see, e.g.,][and the references therein]{neeb:1998,shen_tangent_2022}. 
Below, we will confine attention to the multivariate case. However, keep in mind that much of what we present has infinite-dimensional
extensions. 


We will equip $\C^p$ with the Euclidean inner product $\langle \mb{x},\mb{y}\rangle = \mb{x}^\top \overline{ \mb{y}} = \sum_{j=1}^p x_j \overline y_j$ and corresponding norm
$\|\mb{x}\|:= \sqrt{\langle \mb{x},\mb{x}\rangle}$. Let $L^2(\C^p)$ denote the class of $\C^p$-valued random elements $\mb{Y}$ defined on a common probability space
such that $\E[\|\mb{Y}\|^2]<\infty$.  An $\C^p$-valued stochastic process $\{\mb{Y}(\mb{s}),\ \mb{s}\in \R^d\}$ will be referred to as {\em second-order} if $\mb{Y}(\mb{s})\in L^2(\C^p)$ for all $\mb{s}$.

The {\em cross-covariance} of two zero-mean random vectors $\mb{X}$ and $\mb{Y}$ in $L^2(\C^p)$ is defined as the matrix
\begin{equation}\label{e:CXY}
\mb{C}(\mb{X},\mb{Y}):= \E\left[ \mb{X}\overline{\mb{Y}}^\top\right] = \Big( \E \left[ X_j\overline Y_k\right] \Big)_{p\times p}.
\end{equation}
Since $\left(\mb{X} \overline{\mb{Y}}^\top\right)^* = \mb{Y}\overline{\mb{X}}^\top,$ we have that $\mb{C}(\mb{X},\mb{Y})^* = \mb{C}(\mb{Y},\mb{X})$, where $\mb{A}^*$ denotes the conjugate transpose matrix of $\mb{A}$.
In particular the covariance matrix of $\mb{Y}$, $\mb{C}(\mb{Y},\mb{Y})$ is positive semidefinite. Recall that a complex matrix $\mb{A}$ of dimension $p\times p$ is said to be positive semidefinite or just positive if $\langle\mb{A}\mb{x},\mb{x}\rangle\ge 0$ for all $\mb{x}\in \C^p$, which necessarily implies that $\mb{A}$ is Hermitian, i.e., $\mb{A}=\mb{A}^*$. For convenience, we shall denote the class of complex $p\times p$ matrices as $\T$ and 
the subclass of positive matrices as $\T_+$.  


\begin{definition} A zero-mean second-order $\C^p$-valued stochastic process $\{\mb{Y}(\mb{s}),\ \mb{s}\in \R^d\}$ is said to be {\em covariance stationary} 
if its cross-covariance function $\mb{C}(\mb{Y}(\mb{t}), \mb{Y}(\mb{s})) $ depends only on the lag $\mb{t} - \mb{s}$, in which case define
\begin{equation}\label{e:C-acf}
\mb{C}(\mb{h}):=  \E\left[ \mb{Y}(\mb{h}) \overline{\mb{Y}(\mb{0})}^\top \right],
\ \ \mb{h}\in \R^d,
\end{equation}
which is said to be the stationary (matrix-valued) 
covariance function of $\{\mb{Y}(\mb{s}), \mb{s} \in \mathbb{R}^d\}$.
\end{definition}

As with scalar-valued processes, the stationary matrix-valued covariance functions
are {\em positive semi-definite} in the sense that:
\begin{equation} 
\label{e:C-psd-strong}
\sum_{j,k=1}^n \langle \mb{C}(\mb{s}_j-\mb{s}_k)\mb{x}_j, \mb{x}_k\rangle \ge 0,
\end{equation}
for all $\mb{s}_j\in\R^d,\ \mb{x}_j\in \C^p,\ j=1,\dots,n$ and $n\in \N$. This is immediate from the fact that the left-hand-side of \eqref{e:C-psd-strong} 
equals $\E | \sum_{j=1}^n \langle \mb{x}_j,\mb{Y}(\mb{s}_j)\rangle |^2 \ge 0$.  Interestingly, at least for the class
of continuous functions $\mb{C}(\cdot)$, the above property is equivalent to the seemingly weaker requirement that
\begin{equation} \label{e:C-psd-weak}
    \sum_{j,k=1}^n a_j \overline a_k \mb{C}(\mb{s}_j-\mb{s}_k)  \in \T_+ 
\end{equation}
for all $a_j \in \C,\ \mb{s}_j\in \R^d,\ j=1,\dots,n$ \citep[see, e.g., Corollary 4.4 in][]{shen:stoev:hsing:2020_extended}.

The next result extends the celebrated characterization of continuous positive-definite functions due to Bochner \citep[see, e.g.,][]{bochner1948vorlesungen}. 
Various versions of this result have appeared in the literature; see, for example, \cite{kallianpur1971spectral}, \cite{holmes1979mathematical}, 
\cite{neeb:1998}, \cite{durand2020spectral} and \cite{van2020note}. 
For a self-contained proof, see Theorem 4.2 in \cite{shen:stoev:hsing:2020_extended}. 
In the following, a $\mathbb{T}_{+}$-valued set function $F$ on $\mathcal{B}(\mathbb{R}^d)$ is said to be a $\mathbb{T}_{+}$-valued measure if it is $\sigma$-additive and $|F_{ij}(\mathbb{R}^d)| < \infty $ for all $i,j= 1, \dots, p$.

\begin{theorem}[Bochner]\label{thm:Bochner-Neeb} Let $\mb{C}:\R^d\to \T$ be continuous in each of its entries at $\mb{0}$. The matrix-valued function $\mb{C}$ is positive semidefinite in the sense of \eqref{e:C-psd-weak} if and only if 
there exists a finite $\T_+-$valued measure $\mb{F}$ on $(\R^d,{\cal B}(\R^d))$ such that
\begin{equation}\label{e:thm:Bochner-Neeb}
    \mb{C}(\mb{s}) = \int_{\R^d} e^{\ii \langle \mb{s},\mb{x}\rangle} \mb{F}(d\mb{x}),\ \ \mb{s}\in\R^d.
\end{equation}
In this case, the measure $\mb{F}$ is unique and $\mb{C}:\R^d\to \T$ is uniformly continuous in each of its entries.
Conversely, for every finite $\T_+-$valued measure $\mb{F}$, Relation \eqref{e:thm:Bochner-Neeb} defines a positive
semidefinite function in the sense of \eqref{e:C-psd-strong}.
\end{theorem}

When $\mb{C}$ is the matrix-valued auto-covariance function of a second-order process $\mb{Y}$, the measure $\mb{F}$ in 
\eqref{e:thm:Bochner-Neeb} is referred to as the {\em spectral measure} of $\mb{Y}$.  
In contrast to the scalar case, however, the spectral measure is 
now $\mathbb{C}^p$-valued and for each Borel set $A$,  $\mb{F}(A)$ is 
simply a $p\times p$ positive-semidefinite and self-adjoint matrix, so that in particular
$$
F_{jk}(A) = \overline{F_{kj}(A)},\ j,k\in \{1,2,\dots, p\}.
$$
In this case, the integration in \eqref{e:thm:Bochner-Neeb} can be viewed component-wise with 
respect to the signed  complex measures $F_{jk}(d\mb{x})$.

\begin{remark}  When there exists a measurable function $\mb{f}:\R^d \to \T$, such that 
$\mb{F}(A) = \int_A \mb{f}(\mb{x}) d\mb{x},$ $A\in {\cal B}(\R^d),$
then $\mb{f}$ is referred to as the {\em spectral density} of $\mb{C}$.  One sufficient condition for the existence of a spectral density is that all the entries of $\mb{C}$ are integrable,
in which case one obtains the formula
$
\mb{f}(\mb{x}) = (2\pi)^{-d} \int_{\R^d} e^{\ii \langle \mb{s},\mb{x}\rangle } \mb{C}(\mb{s})d\mb{s}.
$
Integration is again taken simply component-wise.
\end{remark}

%

The Cram\'er Theorem provides an important representation of $L^2$-continuous second-order processes
as stochastic integrals with respect to a random measure with orthogonal increments.  We present next
the corresponding Cram\'er-type result for $\C^p$-valued processes. We begin with defining the random measure.

\begin{definition}\label{def:xi-measure} Let $\mb{F}$ be a finite $\T$-valued measure on $(\R^d,{\cal B}(\R^d))$.  An $\C^p$-valued 
random set-function $\mb{\xi} :{\cal B}(\R^d)\to L^2(\C^p)$ is said to be a random measure with orthogonal increments and
structure or control measure $\mb{F}$ if:
\begin{itemize}
    \item [(i)] $\mb{\xi}(\emptyset) = \mb{0}$ and $\mb{\xi}$ is $\sigma$-additive, i.e., $\mb{\xi}(\cup_{n=1}^\infty A_n) = \sum_{n=1}^\infty \mb{\xi}(A_n),\ a.s.$, for all sequences of pairwise disjoint measurable sets $A_n\subset \R^d$. 
    \item [(ii)] $\mb{\xi}$ has orthogonal increments:
    $$
    \E\left[ \mb{\xi}(A) \overline{\mb{\xi}(B)}^\top \right]= \mb{F}(A\cap B),\mbox{for all }A, B\in {\cal B}(\R^d).
    $$
\end{itemize}
\end{definition}  

The existence of such random measures is in fact established as a by-product in the proof of Cram\'er's theorem. To gain 
some intuition, suppose that $A$ and $B$ are disjoint.  Then, property (ii) implies that $\E[\mb{\xi}(A) \overline{\mb{\xi}(B)^\top} ] = 0$,
so that ${\rm Cov}( \langle \mb{x},\mb{\xi}(A)\rangle, \langle \mb{y}, \mb{\xi}(B)\rangle) = 0$, for all $\mb{x},\mb{y}\in \C^p$ and all 
$A\cap B=\emptyset$.  That is, the measure $\mb{\xi}$ assigns orthogonal (uncorrelated) $\C^p$-valued random elements 
to non-overlapping sets in $\R^d$.

For simple functions $\mb{f}(\mb{x}) = \sum_{j=1}^n \mb{a}_j1_{A_j}(\mb{x})$, for some $\mb{a}_j\in\C^{p\times p}$ and pairwise 
disjoint $A_j$'s, the stochastic integral
$
{\cal I}(\mb{f}):= \int_{\R^d}\mb{f}(\mb{x})\mb{\xi}(d\mb{x}) := \sum_{j=1}^n \mb{a}_j \mb{\xi}(A_j),
$
is well defined and such that 
\begin{equation}\label{e:Ifg-isometry}
\E\left[ {\cal I}(\mb{f}) {\cal I}(\mb{g})^* \right] = \int_{\R^d} \mb{f}(\mb{x})\mb{F}(d\mb{x}) \mb{g}(\mb{x})^*,
\end{equation}
where $\mb{a}^*:= \overline{\mb{a}}^\top$ is the adjoint of $\mb{a}\in \C^{p\times p}$.
By the orthogonality of the increments of $\mb{\xi}$, the last integral involving the sandwiched $\T_+$-valued measure $\mb{F}$ 
is simply equal to the sum
$ \sum_{j} \mb{a}_j \mb{F}(A_j) \mb{b}_j^*,$ where without loss of generality the simple function
$\mb{g}(\mb{x}) = \sum_{j=1}^n \mb{b}_j 1_{A_j}(\mb{x})$ involves the same collection of disjoint sets $A_j$'s.

Thus, with a standard isometry argument, the definition of ${\cal I}$ can be extended to the class of
all Borel functions $\mb{f}$, such that $\int_{\R^d} \|\mb{f}(\mb{x})\|^2 \|\mb{F}\|_{F}(d\mb{x})<\infty$. Then, we can state the following result for which the proof is a natural extension of the original ideas of \cite{cramer1942harmonic} \citep[cf.\ 
Theorem 4.7 in][]{shen:stoev:hsing:2020_extended}.


\begin{theorem}[Cram\'er]\label{thm:Cramer} Let $\mb{Y}=\{\mb{Y}(\mb{s}),\ \mb{s}\in\R^d\}$ be an $L^2$-continuous covariance stationary $\C^p$-valued
process with spectral measure $\mb{F}$.  Then, there exists an almost surely unique random measure $\mb{\xi}$ with orthogonal
increments and structure measure $\mb{F}$ such that, for each $\mb{s}\in\R^d$,
\begin{equation}\label{e:thm:Cramer}
\mb{Y}(\mb{s}) = \int_{\R^d} e^{\ii \langle \mb{s},\mb{x}\rangle} \mb{\xi}(d\mb{x}),\ \ \mbox{ almost surely.} 
\end{equation}
\end{theorem}

We emphasize that the above stochastic integral gives an almost sure representation for each {\em fixed} $\mb{s}$, but this does not mean that the representation is valid path-wise (with probability one). 


\begin{remark}
We now discuss an approach to generate a process $\mb{Y}(\mb{s}) \in \mathbb{C}^p$ from a covariance with a given spectral density function $\mb{f}(\mb{x})$. One may decompose the spectral density as 
\begin{align*}
\int_A \mb{f}(\mb{x}) d\mb{x} &= \int_A \mb{g}(\mb{x}) \mb{\eta}(d\mb{x})\mb{g}(\mb{x})^* \\
&= \int_A \mb{g}(\mb{x}) \mb{\eta}(d\mb{x}) \overline{\mb{g}(\mb{x})}^\top, 
\end{align*}
for some $\T_+$-valued measure $\mb{\eta}$.  Indeed, since $\mb{f}(\mb{x}) \in \T_+$, 
one can take in particular $\mb{g}(\mb{x}):= \mb{f}(\mb{x})^{1/2}$ and
$\mb{\eta}(d\mb{x}) = \mb{I}_p d\mb{x}$.  Define
$$
\mb{\xi}(A):= \int_A \mb{g}(\mb{x}) \mb{B}(d\mb{x}),
$$
where $\mb{B}(d\mb{x})$ is a zero-mean complex Gaussian random measure with orthogonal increments such that 
\begin{align}\begin{split}
\mb{B}(-d\mb{x}) &= \overline{\mb{B}(d\mb{x})}\ \mbox{ and }\\ \E\left[ \mb{B}(A_1) \overline{\mb{B}(A_2)}^\top\right]
&= \mb{\eta}(A_1\cap A_2),\label{eq:isom}\end{split}
\end{align}
for Borel sets $A_1$ and $A_2$.

Define
$$
\mb{Y}(\mb{s}) := \int_{\R^d} e^{\ii \langle \mb{s},\mb{x}\rangle} \mb{\xi}(d\mb{x})\equiv \int_{\R^d} e^{\ii \langle \mb{s},\mb{x}\rangle} \mb{g}(\mb{x}) \mb{B}(d\mb{x}).
$$Then note that by property \eqref{eq:isom}
\begin{align*}
&\E\left[\mb{Y}(\mb{t})\overline{\mb{Y}}(\mb{s})^\top\right]\\
&~~~~~= \int_{\R^d} e^{\ii \langle \mb{t}-\mb{s},\mb{x}\rangle} \mb{g}(\mb{x})\E\left[ \mb{B}(d\mb{x})\overline{\mb{B}(d\mb{x})}^\top\right] \overline{\mb{g}(\mb{x})}^\top\\
&~~~~~ = \int_{\R^d} e^{\ii \langle \mb{t}-\mb{s},\mb{x}\rangle} \mb{g}(\mb{x}) \mb{\eta}(d\mb{x}) \mb{g}(\mb{x})^* \\
&~~~~~= \int_{\R^d} e^{\ii \langle \mb{t}-\mb{s},\mb{x}\rangle} \mb{f}(\mb{x}) d\mb{x} =: \mb{C}(\mb{t}-\mb{s}).
\end{align*}
This example illustrates the relationship between the Bochner and Cram\'er theorems and shows in 
particular why every continuous positive definite function can be realized as the auto-covariance
of (a possibly complex-valued) second order process $\mb{Y}$. 
\end{remark}

In applications, we normally deal with {\em real} processes, where $\mb{Y}(\mb{s})$ takes values in $\R^p$ and of course its auto-covariance is real.  This imposes constraints on both the spectral measure $\mb{F}$ in \eqref{e:thm:Bochner-Neeb} and the orthogonal random measure $\mb{\xi}$ in \eqref{e:thm:Cramer} 
as shown next.

\begin{proposition}\label{p:Y-real} Let $\mb{Y}=\{\mb{Y}(\mb{s}),\ \mb{s}\in \R^d\}$ be a zero-mean second-order $L^2$-continuous
and covariance stationary process with auto-covariance $\mb{C}$, spectral measure $\mb{F}$, and stochastic
representation as in \eqref{e:thm:Cramer}. The following statements hold:

\begin{itemize}
    
    \item [(i)] $\mb{C}(\mb{s})$ is real for all $\mb{s}\in\R^d$, if and only if 
     $\mb{F}(-A) = \overline{\mb{F}(A)}$, for all $A\in {\cal B}(\R^d)$
     \item[(ii)] The process $\mb{Y}$ is real if and only if $\mb{\xi}(-A) = \overline{\mb{\xi}(A)}$, a.s., for all $A\in {\cal B}(\R^d)$
    \item [(iii)] Claim (ii) implies (i), and the converse is not true. However, every real auto-covariance can 
    be realized by a real process $\mb{Y}$.
\end{itemize}
\end{proposition}

\begin{definition} \label{def:Hermitian} The measure $\mb{F}$ and the random measure $\mb{\xi}$ that satisfy properties
(i) and (ii) of Proposition \ref{p:Y-real}
will be referred to as Hermitian.
\end{definition}

Let now $\mb{Y}=\{\mb{Y}(\mb{s}),\ \mb{s}\in\R^d\}$ be a covariance-stationary $\R^p$-valued process. In the multivariate case 
($p\ge 2$) the auto-covariance of $\mb{Y}$ is typically not symmetric, i.e., in general  $\mb{C}(-\mb{s}) \not =\mb{C}(\mb{s})$.  Nevertheless, many existing
multivariate models explicitly or implicitly impose this symmetry property, which can severely constrain the structure of
the spectrum.

\begin{proposition}\label{p:Y-real-reversible} A second-order stationary $\R^p$-
valued process $\mb{Y}=\{\mb{Y}(\mb{s}),\ \mb{s}\in \R^d\}$ is 
covariance-symmetric, i.e.\ $\mb{C}(\mb{s}) = \mb{C}(-\mb{s})$ for 
all $\mb{s}\in\R^d$, if and only if its spectral measure $\mb{F}$ is 
{\em real} and {\em symmetric}, that is,
$$
\mb{F}(A) = \mb{F}(-A) = \overline{\mb{F}(A)},\ \ A\in {\cal B}(\R^d).
$$
\end{proposition}
\begin{proof} In view of \eqref{e:thm:Bochner-Neeb} and covariance-symmetry, 
$$
\mb{C}(\mb{s}) =  \mb{C}(-\mb{s}) = \int_{\R^d} e^{\ii \langle \mb{s},\mb{x}\rangle} \mb{F}(-d\mb{x}),
$$
which by the uniqueness of $\mb{F}$ in \eqref{e:thm:Bochner-Neeb} entails $\mb{F}(d\mb{x}) = \mb{F}(-d\mb{x})$. With $\mb{Y}$ real, Proposition \ref{p:Y-real}
also implies that $\mb{F}(d\mb{x}) = \overline{\mb{F}(-d\mb{x})}$.  This completes the proof.
\end{proof}
\begin{remark}
Propositions \ref{p:Y-real} and \ref{p:Y-real-reversible} imply that 
a real, covariance-stationary process $\mb{Y}$ will have a complex-valued spectral measure unless it is covariance-symmetric.
\end{remark}

\subsection{Comparison with existing multivariate Mat\'ern models} \label{sec:new_Matern}

While the spectral density has been an important feature of multivariate spatial models, models in the literature are primarily described in the spatial, as opposed to the frequency, domain. 
For example, \cite{gneiting_matern_2010} and subsequent work begins by proposing cross-covariances that are proportional to a Mat\'ern covariance and have their own $a$ and $\nu$ parameters; in mathematical notation, one takes \begin{align*}
C_{jk}(h) &= \frac{2^{1- \nu_{jk}}}{\Gamma(\nu_{jk})}\sigma_{jk} (a_{jk}|h|)^{\nu_{jk}} \besselK_{\nu_{jk}}(a_{jk}|h|)
\end{align*}for $a_{jk} >0$, $\nu_{jk} > 0$, and $\sigma_{jk}$ real-valued and positive. 

\cite{gneiting_matern_2010} then use the matrix-valued spectral densities to ensure validity of the model. 
\cite{bolin_multivariate_nodate} and \cite{guinness_nonparametric_2022} also construct models through the spectral domain. 
While \cite{bolin_multivariate_nodate}, as mentioned earlier, describes a similar approach to ours, \cite{guinness_nonparametric_2022} focuses on data observed on regular grids. 
Analysis of multivariate Mat\'ern models through the spectral domain include \cite{kleiber_coherence_2018} and \cite{guinness2022inverses}.
The title of this paper refers to the approach of using of a complex-valued variance parameterization in the spectral domain and the factoring of the spectral density into the context of this previous literature.
From here, we more generally compare our proposed models with previous literature, focusing on a few of the models' aspects. 
\begin{itemize}
\item {\bf Model validity:} For the cross-covariances based on our spectral approach, model validity is immediate in any spatial dimension $d$ and any number of components $p$. 
For the multivariate Mat\'ern proposed by \cite{gneiting_matern_2010}, finding parameter values for which the model is valid has been rather nontrivial. Although substantial progress has been made on that by ensuing work, for instance, \cite{apanasovich_valid_2012}, \cite{du2011spherically}, and \cite{emery_new_2022}, 
the parameter constraints still tend to be technical and difficult to interpret.


\item {\bf Model flexibility:} The new cross-covariance models introduced here have a large amount of flexibility.
In particular, the cross-covariances can have flexible asymmetric forms, the lag $\mb{h} \in \mathbb{R}^d$ with the most dependence between two processes may not be $
\mb{h}=\mb{0}$, and the dependence between processes may be positive for some lags yet negative for others. 
These properties were not originally available for the multivariate Mat\'ern of \cite{gneiting_matern_2010}, which is symmetric and entails positive (or negative) dependence for all spatial lags. 

Considerable research efforts have led to a variety of improvements in the flexibility of multivariate Mat\'ern models.
For example, \cite{li_approach_2011}, \cite{qadir_flexible_2020}, and \cite{mu2024gaussian}propose delay-type asymmetries to cross-covariance functions. \cite{vu_modeling_2022} introduce a deformation approach that also generalizes the delay-type asymmetries. 
From a different perspective, \cite{alegria_bivariate_2021} propose symmetric cross-covariances for bivariate processes with maximal correlation at a lag other than $h=0$. 
More general approaches for introducing flexibility in cross-covariance models include latent dimensions \citep{apanasovich_cross_covariance_2010} and conditioning  \citep{cressie_multivariate_2016}. 

\item {\bf Parameter space:} The proposed model has a substantially different parameter space compared to the multivariate Mat\'ern of \cite{gneiting_matern_2010}. 
The multivariate Mat\'ern models based on \cite{gneiting_matern_2010} require the additional ``smoothness'' parameter $\nu_{jk}$ and the inverse range parameter $a_{jk}$ that describe each cross-covariance. 
\cite{kleiber_coherence_2018} notes that these parameters do not have ``straightforward interpretations.'' 
Since the size of the parameter space is a computational concern for such models \citep{guinness_nonparametric_2022}, each additional parameter in the cross-covariances complicates model estimation. 

In our cross-covariances, the introduction of additional parameters $a_{jk}$ and $\nu_{jk}$ is not necessary.
The parameter constraints on the model introduced here are also more simple: one needs the smoothness and range parameters of each process to be positive (that is, that each univariate process is valid) as well as the matrix-valued spectral measure $\sd(\cdot)$ describing the variance and covariances to be positive and Hermitian. 
For existing multivariate Mat\'ern models, the validity constraints also challenge the interpretation of other parameters: for example, the ``correlation'' parameter between two processes may be limited to an interval smaller than $(-1,1)$, complicating its interpretation as such \citep{emery_new_2022}. 

\end{itemize}

As we described in the bullet points above, there has been substantial progress dealing with the individual issues discussed there.
However, for the most part, these issues are considered separately.
For example, the literature in ``model flexibility'' above builds upon existing covariances or cross-covariances, leading to restrictive conditions for model validity and an expanded parameter space when using the multivariate Mat\'ern of \cite{gneiting_matern_2010}. 
Conversely, the model introduced by \cite{bolin_multivariate_nodate} has simple validity conditions and a reduced parameter space, but it does not introduce additional structure in the cross-covariance functions. 
A key advantage of the spectral approach in this paper is that it allows us to address a multitude of issues simultaneously.

\section{Multivariate Mat\'ern models in one dimension}\label{Introduce_model}

This section provides the formulation and description of new multivariate Mat\'ern models when $d=1$. We return to the proposed spectral density in \eqref{eq:new_cc} based on the self-adjoint matrix $\Sig_H = [\sigma_{jk}]_{j,k=1}^p$ which gives a cross-covariance between process $j$ and $k$ at lag  $h = \tilc_1 - \tilc_2$ of 
\begin{align}\label{eq:d=1-C-definition}\begin{split}
\mathbb{E}&\left[Y_j(\tilc_1) Y_k(\tilc_2)\right]=c_jc_k\int_\mathbb{R} e^{\I hx}(a_j + \I x)^{-\nu_j - \frac{1}{2}}\\
&~~~~~~~~~~~\times \{\Re(\sigma_{jk}) + \I\Im(\sigma_{jk})\textrm{sign}(x)\} \\
&~~~~~~~~~~~\times(a_k - \I x)^{-\nu_k - \frac{1}{2}} dx. \end{split}
\end{align}
One can also consider a model with permuted signs of $a_j - \I x$ and $a_k + \I x$; these turn out to be reflected versions of cross-covariances with $a_j + \I x$ and $a_k -\I x$ and thus correspond to similar shapes; see Proposition \ref{prop:reflected} later. 
Let $g_{jk}(x) =c_jc_k(a_j + \I x)^{-\nu_j - \frac{1}{2}}(a_k - \I x)^{-\nu_k - \frac{1}{2}}$, and notice that
\begin{align}\begin{split}
\mathbb{E}&\left[Y_j(\tilc_1) Y_k(\tilc_2)\right]=\Re(\sig_{jk})\int_\mathbb{R}e^{\I hx} g_{jk}( x) dx \\
&~~~~~~+ \Im(\sig_{jk})\I\int_\mathbb{R}e^{\I hx} g_{jk}( x) \textrm{sign}(x)  dx\label{eq:initial_breakdown} \\
&=:\Re(\sig_{jk})C_{jk}^\Re(h) + \Im(\sig_{jk})C_{jk}^\Im(h)\end{split}, \end{align}where $C_{jk}^\Re(h)\in \mathbb{R}$ and $C_{jk}^\Im(h)\in \mathbb{R}$ are the respective portions of the cross-correlation functions corresponding to $\Re(\sig_{jk})$ and $\Im(\sig_{jk})$. The next two subsections deal with the two terms of \eqref{eq:initial_breakdown} individually. 

\subsection{Cross-covariances with real directional measure}\label{sec:re_ent}
Here, we assume that $\sig_{jk}$ is real, and from \eqref{eq:initial_breakdown} we obtain the valid cross-covariances of the form \begin{align}\begin{split}
C_{jk}(h)&:=\mathbb{E}[Y_j(h)Y_{k}(0)]\\
&=\Re(\sig_{jk})\int_\mathbb{R} e^{\I hx}g_{jk}(x) dx.\label{eq:re_part_int2}\end{split}\end{align}
The resulting expression for this integral involves the Whittaker function $\whittW_{\kappa, \mu}(z)$, which we briefly discuss following the results in Chapter 13 of \cite{NIST:DLMF}. 
In particular, define \begin{align*}
    \whittW_{\kappa, \mu}(z) = \textrm{exp}\left(-\frac{1}{2}z\right)z^{\frac{1}{2} + \mu} U\left(\frac{1}{2} + \mu - \kappa,1 + 2\mu, z\right),
\end{align*}where $U(a,b,z)$ is a confluent hypergeometric function. 
See, for example, the works \cite{NIST:DLMF} or \cite{abramowitz_handbook_1972} for full definitions of $W_{\kappa, \mu}(z)$ and $U(a,b,z)$. 
Furthermore, the limiting form 
\begin{align}
    \whittW_{\kappa, \mu}(z) \sim \textrm{exp}\left(-\frac{1}{2}z\right)z^\kappa,\ \ \mbox{ as } z\to\infty\label{eq:asymp_whitt_inf}
\end{align} holds \citep{NIST:DLMF}. 
Here, $f(x) \sim g(x)$ as $x\to x_0$, means that $f(x)/g(x)\to 1$ as $ x\to x_0$.
The function also satisfies $\whittW_{\kappa, \mu}(z) = \whittW_{\kappa, -\mu}(z)$, and the modified Bessel function of the second kind, $\besselK_\nu(z)$, is related to a special case of the function $\whittW_{\kappa, \mu}(z)$:\begin{align*}
    \whittW_{0,\nu}(2z)&=\sqrt{\frac{2z}{\pi}}\besselK_\nu(z),
\end{align*}which we will use to compare this cross-covariance to the Mat\'ern covariance.

Now, we provide a closed-form expression of the integral in \eqref{eq:re_part_int2}. 

\begin{theorem}\label{thm:re_ent}
Suppose $\Im(\sig_{jk}) = 0$ and for notational ease, define the values 
\begin{align*}
a_+ &= \frac{(a_j + a_k)}{2},&  a_- &= \frac{(a_j - a_k)}{2},\\  \nu_+ &= \frac{(\nu_j + \nu_k)}{2},& \mbox{ and }\ \nu_- &= \frac{(\nu_j - \nu_k)}{2}.
\end{align*}
Then, the closed-form formula of the $j,k$-th entry of the cross-covariance as in \eqref{eq:re_part_int2} is \begin{align}
\begin{split}
  &C_{jk}(h)=\Re(\sigma_{jk}) c_jc_k \frac{\pi}{a_+} \left(\frac{|h|}{2a_+}\right)^{\nu_+ -\frac{1}{2}} \\ 
  &~~~~~~~~~\times{\normalfont \textrm{exp}}(-ha_-)\begin{cases}
\frac{\whittW_{\nu_-, \nu_+}(2a_+|h|)}{\Gamma(\nu_j+\frac{1}{2})} & h > 0\\
\frac{\whittW_{-\nu_-, \nu_+}(2a_+|h|)}{\Gamma(\nu_k+\frac{1}{2})} & h < 0
\end{cases}.\label{eq:whittaker}
\end{split}
\end{align}
For $h = 0$, the cross-covariance value is \begin{align*}
      C_{jk}(0)
      &= \Re(\sigma_{jk}) \frac{a_j^{\nu_j}a_k^{\nu_k}}{a_+^{2\nu_+}}\frac{\Gamma(2\nu_+)}{\sqrt{\Gamma(2\nu_j)\Gamma(2\nu_k)}}.
\end{align*}
\end{theorem}
\begin{proof}

For $h\neq 0$, by applying  3.384 (9) of \cite{GR_table_2015}, we see that \begin{align*}
&\int_\mathbb{R}e^{\I hx} g_{jk}( x) dx \\
&~~~~~=c_jc_k\int_\mathbb{R} e^{\I hx}(a_j+ \I x)^{-\nu_j - \frac{1}{2}}(a_k - \I x)^{-\nu_k - \frac{1}{2}} dx \\ &~~~~~=c_jc_k2\pi(2a_+)^{ -\nu_+ - \frac{1}{2}} |h|^{\nu_+ -\frac{1}{2}}\textrm{exp}\left(-ha_-\right)\\ &~~~~~~~~~~~\times\begin{cases}
\whittW_{\nu_-, - \nu_+}(2a_+|h|)/\Gamma(\nu_j+\frac{1}{2}) & h > 0\\
\whittW_{-\nu_-, -\nu_+}(2a_+|h|)/\Gamma(\nu_k+\frac{1}{2}) & h < 0,
\end{cases}
\end{align*}under the assumption that $h \in \mathbb{R}$ and $\nu_j$, $\nu_k$, $a_j$, and $a_k \in \mathbb{R}_{>0}$. 
Simplifying and applying the formula $\whittW_{\kappa, \mu}(z) = \whittW_{\kappa, -\mu}(z)$ gives the final form. 

For $h = 0$, since the function $x \mapsto (a_j+ \I x)^{-\nu_j - \frac{1}{2}}(a_k - \I x)^{-\nu_k - \frac{1}{2}}$ is integrable, its Fourier transform $C_{jk}(h)$ is uniformly continuous in $h$, we use the expansion for $W_{\kappa, \mu}(z)$ near $z=0$ described in \cite{NIST:DLMF}: that is, \begin{align*}
    W_{\nu_-, \nu_+}(2a_+ |h|) &\overset{h \downarrow 0}{\sim} \frac{\Gamma(2\nu_+)}{\Gamma(\nu_k + \frac{1}{2})}(2a_+|h|)^{\frac{1}{2} - \nu_+}.
\end{align*}
A similar expression is obtained when taking $h \uparrow 0$.
The final form comes from substituting in the values of $c_j$ and $c_k$ and using properties of the gamma function. 
\end{proof}

This formula is relatively complicated and not immediately intuitive, so we next discuss the intricacies of this model in a series of remarks. 

\begin{remark}[Relation to probability density function of a gamma difference distribution]
The form \eqref{eq:whittaker} is proportional to the probability density function of the ``gamma difference'' distribution; see \cite{klar_note_2015} for more information. 
In particular, if $X_j \sim \textrm{Gamma}(\nu_j + 1/2, a_j)$ and $X_k \sim \textrm{Gamma}(\nu_k + 1/2, a_k)$ are independent, then $X_j - X_k$ has probability density function proportional to \eqref{eq:whittaker}; the random variable $X_j- X_k$ has mean $\nu_j/a_j - \nu_k/a_k$ and variance $(\nu_j + 1/2)/a_j^2 + (\nu_k + 1/2)/a_k^2$. 
Note that $g_{jk}(x)$ is proportional to the product of characteristic functions of $-X_j$ and $X_k$, which are $(a_j + \I x)^{-\nu_j - \frac{1}{2}}$ and $(a_k - \I x)^{-\nu_k - \frac{1}{2}}$, respectively.
Since this cross-covariance function is always proportional to probability density functions, it takes the sign of $\Re(\sigma_{jk})$ for all lags. 
Notice that, by using the normalization constant of the gamma difference distribution presented in \cite{klar_note_2015}, we can use the adjustment $\tilde{c}_j\tilde{c}_k$ in the place of $c_jc_k$ with $\tilde{c}_j^{-1} = \sqrt{2\pi} a_j^{-\nu_j - 1/2}$ to obtain $\int_\mathbb{R} C_{jk}(h) dh = \Re(\sigma_{jk})$, if such a property is desired.
\end{remark}

\begin{remark}[Special cases]
We consider a series of special cases of this cross-covariance. First, consider the case where $\nu := \nu_j = \nu_k$. Here, since $\whittW_{0, \nu}(2z) = \sqrt{2z/\pi} \besselK_\nu(z)$,  \eqref{eq:whittaker} reduces to \begin{align*}
    C_{jk}(h) = \Re(\sigma_{jk})\frac{2(2a_+)^{ -\nu}(a_ja_k|h|)^{\nu}}{\Gamma(\nu)} e^{-h a_-}\besselK_{\nu}(a_+|h|).
\end{align*}This is proportional to the probability density function of a Bessel function distribution or variance-gamma distribution; see Section 4.1 of \cite{kotz_laplace_2001} or \cite{fischer_variance_gamma_2023}. Furthermore, if $\nu = \nu_j = \nu_k$ and $a := a_j = a_k$, \eqref{eq:whittaker} reduces to \begin{align*}
C_{jk}(h) = \Re(\sigma_{jk})\frac{2^{1- \nu} }{\Gamma(\nu)}(a|h|)^{\nu}\besselK_{\nu}(a|h|),\end{align*}leading directly to a function proportional to the Mat\'ern covariance, for this case matching the cross-covariance in \cite{gneiting_matern_2010} and \cite{bolin_multivariate_nodate}. 
Characterization of the Mat\'ern covariance as proportional to the Bessel function distribution has been established, for example, in Section 2.3 of \cite{paciorek_2003}. 

Next, consider the case where $\nu_k = 1/2$. 
Due to the representation $\whittW_{\frac{1}{4} - \frac{\nu_j}{2}, \frac{1}{4} + \frac{\nu_j}{2}}(2a_+|h|) = e^{-a_+|h|} (2a_+|h|)^{\frac{1}{4} - \frac{\nu_j}{2}} $ for $h < 0$ when using (13.18.2) of \cite{NIST:DLMF}, we have \begin{align*}
    C_{jk}(h) &=\Re(\sigma_{jk}) c_jc_k2\pi(2a_+)^{ -\nu_j - \frac{1}{2} }{\normalfont \textrm{exp}}(-|h|a_k) , ~ h < 0.
\end{align*}
Thus, if one of the marginal covariances has exponential form, the cross-covariance function is proportional to an exponential covariance function for half of its domain.

Finally, consider when $\nu_j = \nu_k= 1/2$ and $a_j \neq a_k$.
Since $\besselK_{1/2}(z) = (\pi/2z)^{1/2}\textrm{exp}(-z)$, \eqref{eq:whittaker} becomes \begin{align*}
C_{jk}(h) = \Re(\sigma_{jk})\frac{(a_ja_k)^{\frac{1}{2}}}{a_+}\times\begin{cases}
    \textrm{exp}\left(-a_j|h|\right) & \textrm{if } h >0 \\
    \textrm{exp}\left(-a_k|h|\right) & \textrm{if } h <0   
\end{cases}.
\end{align*}This function is proportional to the probability density function of an asymmetric Laplace distribution. 
Thus, when both marginal distributions have an exponential covariance function, the cross-covariance developed here has a similar form, proportional to the density of a (potentially asymmetric) Laplace distribution. 
The cross-covariance form can also be obtained using exponentials when $\nu_j = \nu_k = 3/2$ and $a_j\neq a_k$ using the expression for $K_{3/2}(z)$.
In Table \ref{tab:distribution_relationships}, we summarize these relationships.

\begin{table}[ht]
    \centering
    \begin{tabular}{|c|c|c|}\hline
        Setting & $a_j = a_k$ & $a_j \neq a_k$  \\ \hline
        $\nu_j = \nu_k = 1/2$ & Laplace & asymmetric Laplace \\ 
       $\nu_j = \nu_k$ & Bessel function & Bessel function \\ 
       $\nu_j \neq \nu_k$ & gamma difference & gamma difference  \\ \hline
   \end{tabular}
    \caption[Relationships between multivariate Mat\'ern and distributions]{Relationships between multivariate Mat\'ern cross-covariances with real directional density and probability density functions of distributions. With different parameter settings, the Mat\'ern cross-covariances presented here are proportional to the probability density functions of these distributions.}
    \label{tab:distribution_relationships}
\end{table}

\end{remark}

\begin{remark}[Normalization]

\normalfont
We primarily take the normalization using $c_j$ and $c_k$ as defined in \eqref{eq:matern_normalization} (the ``original'' normalization), ensuring that we obtain a Mat\'ern covariance with a value at $h = 0$ of $\Re(\sigma_{jj})$ for each $j$. 
Also notice, for visualization purposes only, we can also apply the normalization suggested by \eqref{eq:whittaker} to obtain $C_{jk}(0) = \Re(\sigma_{jk})$, which we use for some of the panels of Figure \ref{fig:type_pos}. 
\end{remark}

\begin{remark}[Visualization of cross-covariances]
We plot resulting cross-covariances for different parameter values and normalizations in Figure \ref{fig:type_pos}.
In general, we see that imbalances between parameters $\nu_j$ and $\nu_k$ or $a_j$ and $a_k$ introduce asymmetries into the cross-covariances. 
Changing the parameters $\nu_j$ and $a_j$ primarily change the behavior of the cross-covariance over positive lags. 
We next formalize this observation.

\begin{figure}[ht]
\centering
\includegraphics[width = .45 \textwidth]{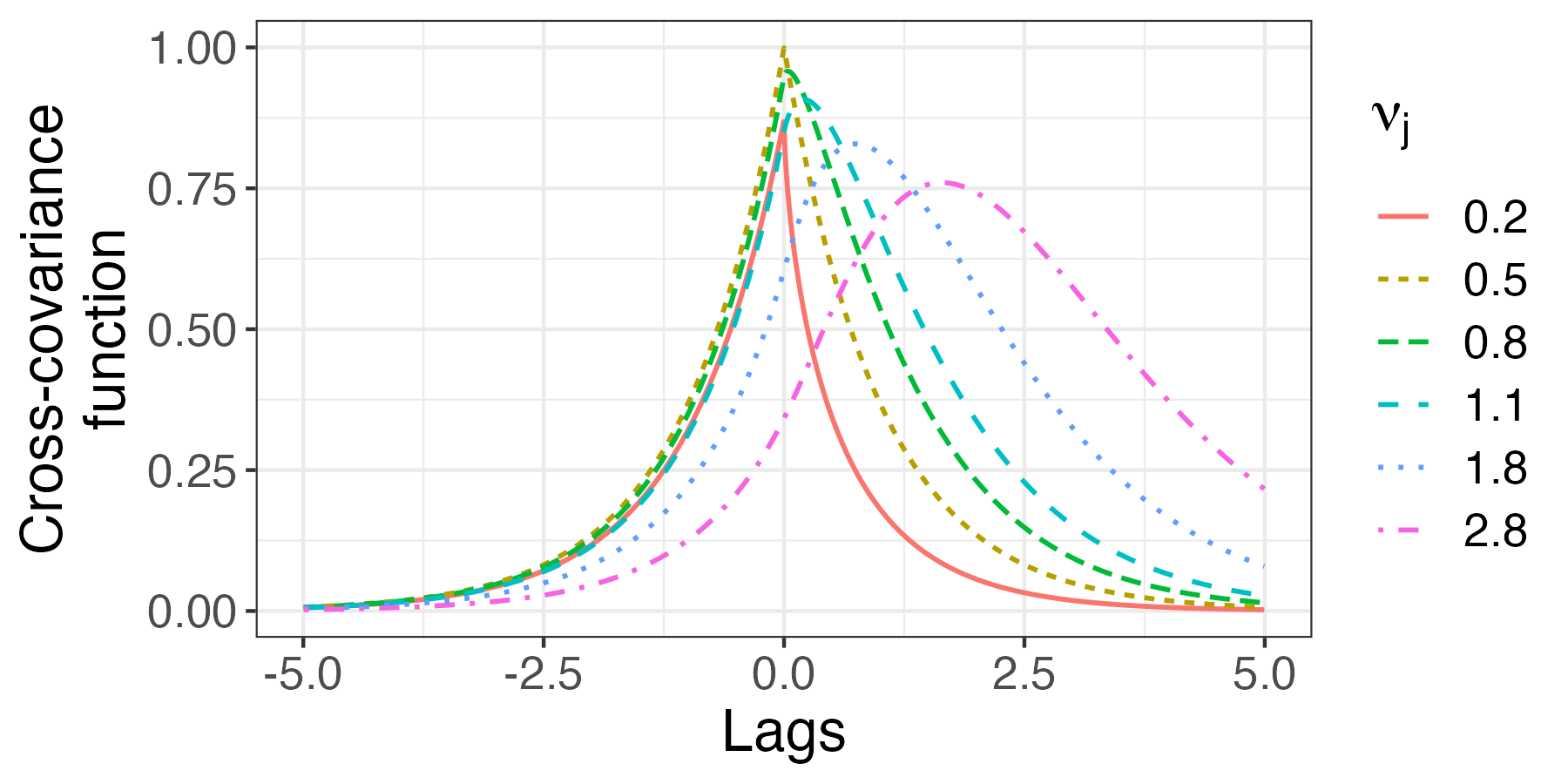}
\includegraphics[width = .45 \textwidth]{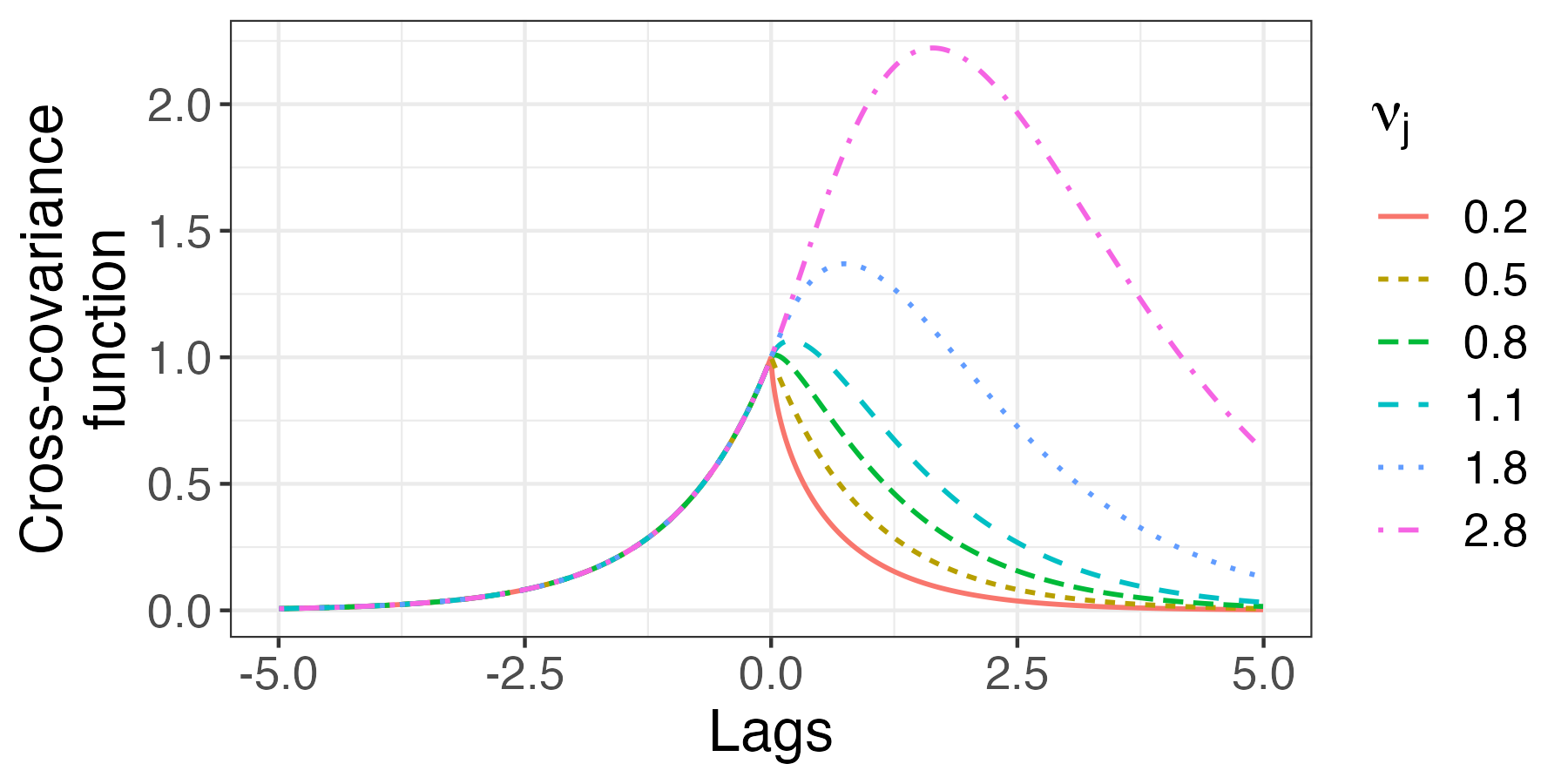}

\includegraphics[width = .45 \textwidth]{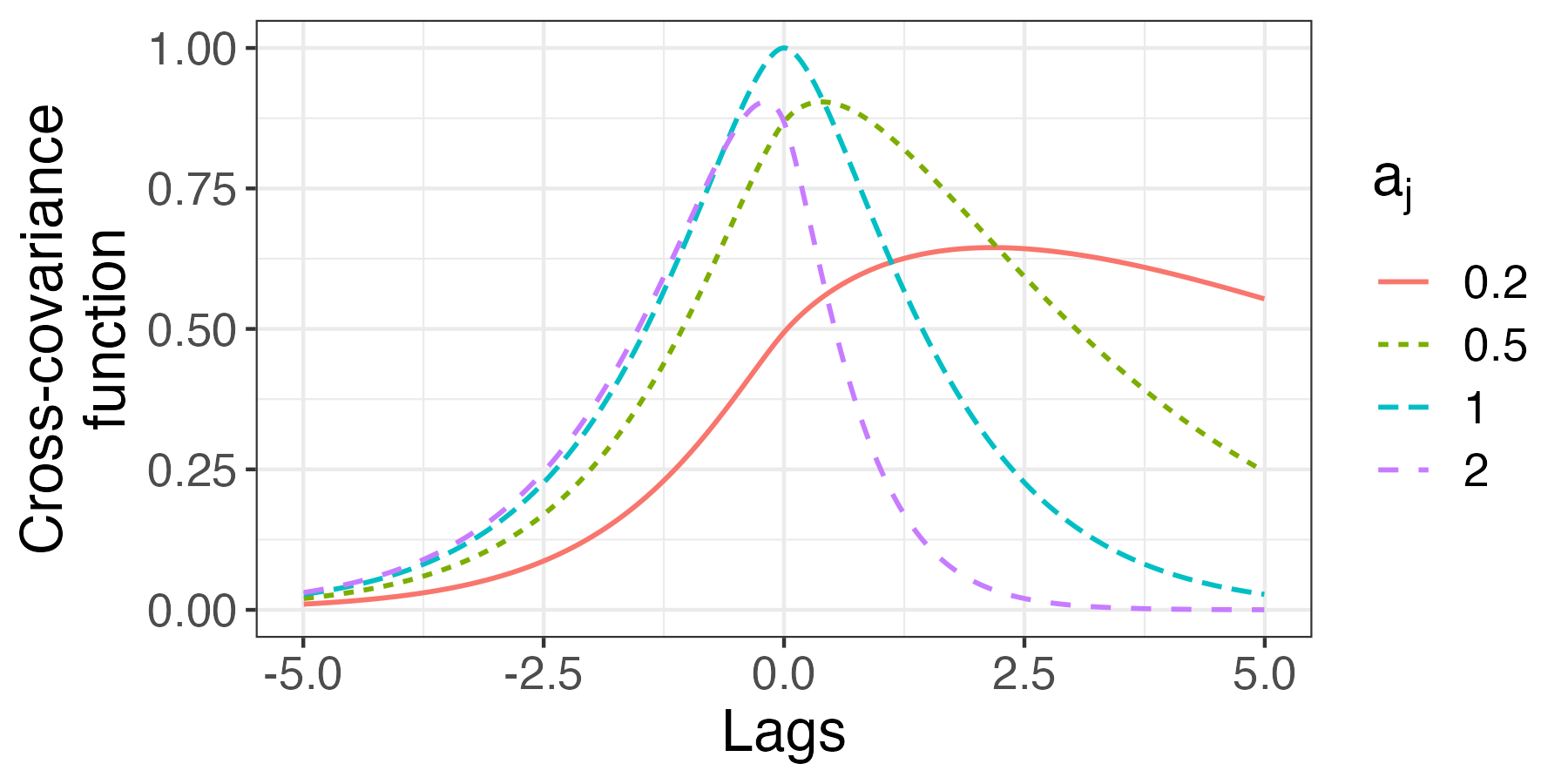}
\includegraphics[width = .45 \textwidth]{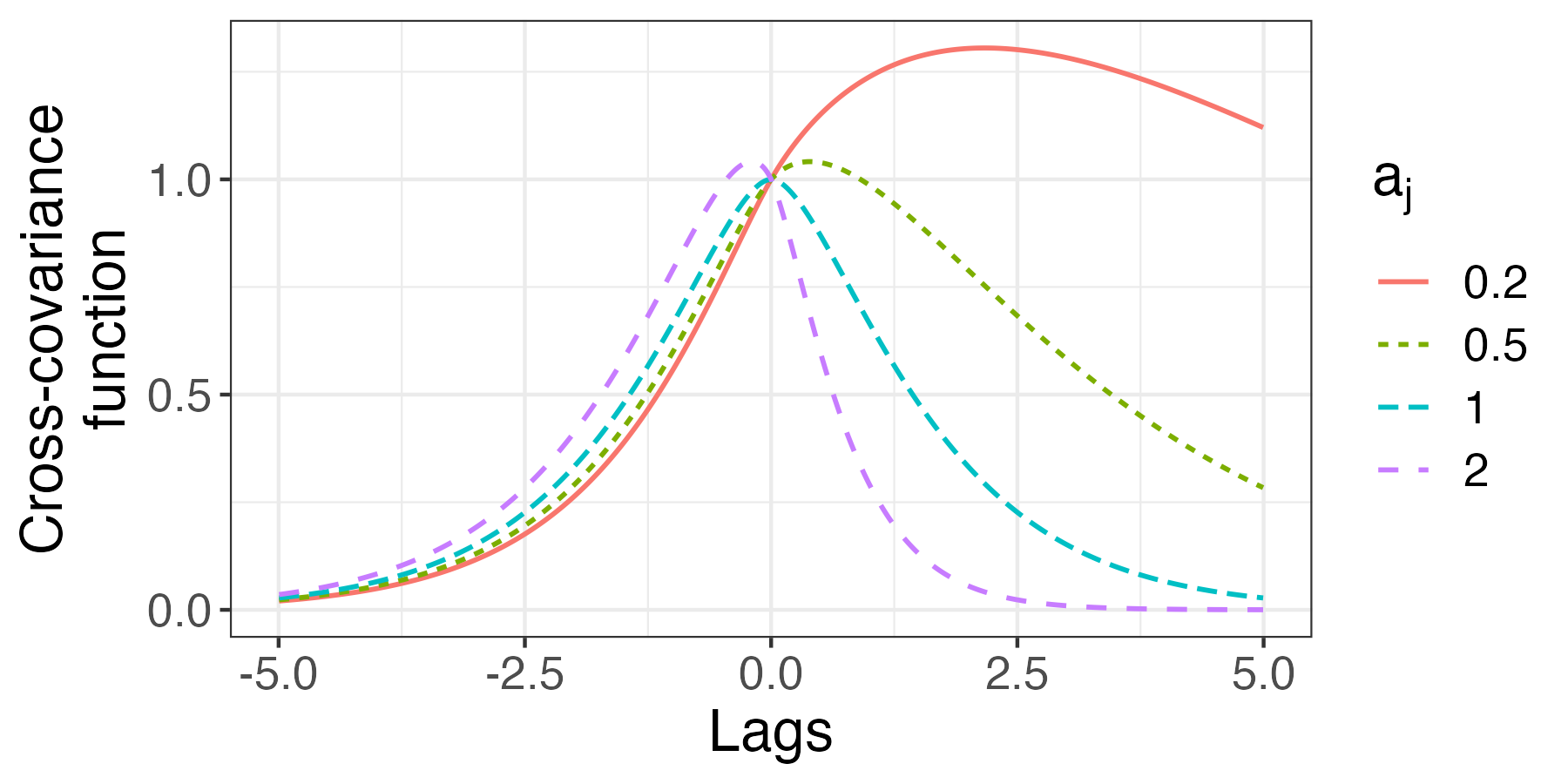}

\caption[Cross-covariance examples for $d=1$ and real directional measure]{Mat\'ern cross-covariances with $d=1$ and $\sigma_{jk} = 1$ (Top) holding fixed $a_j = a_k = 1$ and $\nu_k= 0.5$, we plot the cross-covariance function when varying $\nu_j$ under (Upper Top) the original normalization and (Lower Top) the alternative normalization; (Bottom) Holding fixed $a_k = 1$, $\nu_j= \nu_k = 0.8$, we plot the cross-covariance function when varying $a_j$ under (Upper Bottom) the original normalization and (Lower Bottom) the alternative normalization.}\label{fig:type_pos}
\end{figure}
\end{remark}

\begin{remark}[Expansion at $h \to \pm \infty$]\normalfont
The asymptotic expression for the Mat\'ern covariance is \begin{align*}
    {\cal M}(|h|; a, \nu, \sigma_{11}) 
    \overset{h \to \infty}{\sim} \sigma_{11}\frac{2^{\frac{1}{2} - \nu}\sqrt{\pi}}{\Gamma(\nu)}(a|h|)^{\nu - \frac{1}{2}} \textrm{exp}(-a|h|),
\end{align*}by using the asymptotic expressions for the function $\besselK_\nu(z) \overset{z \to \infty}{\sim} (\pi/(2z))^{1/2}  \textrm{exp}(-z)$ \citep[see Section 10.40 of][]{NIST:DLMF}. 
We will also use the asymptotic expansion of the function $\whittW_{\kappa, \mu}(h)$ to analyze the asymptotic expression of the new cross-covariance \citep{NIST:DLMF}. 
In particular, we combine \eqref{eq:asymp_whitt_inf} and  \eqref{eq:whittaker} to obtain
\begin{align*}
    C_{jk}(h)\overset{h\to \infty}{\sim}     \Re(\sigma_{jk}) c_{jk}^*(a_j|h|)^{\nu_j- \frac12} \textrm{exp}(-a_j|h|), 
\end{align*}so that, up to a positive constant that does not depend on $h$ defined by \begin{align*}
    c_{jk}^* &= a_j^{\frac{1}{2}}a_k^{\nu_k}\sqrt{\frac{\Gamma(\nu_k + \frac{1}{2})}{\Gamma(\nu_j + \frac{1}{2})}}\frac{2\sqrt{\pi}(2\alpha_+)^{-\nu_k - \frac12}}{\sqrt{\Gamma(\nu_j)\Gamma(\nu_k)}} ,
\end{align*}the cross-covariance decays like the Mat\'ern covariance class with parameters $\Re(\sigma_{jk})$, $\nu_j$, and $a_j$ in one direction. 
Of course, $\Re(\sigma_{jk})<0$ may hold to represent negative correlation between the two processes. 
Alternately, as $h \to -\infty$, \begin{align*}
    C_{jk}(h)&\overset{h\to -\infty}{\sim}
   c^*_{kj}(a_k|h|)^{\nu_k- \frac12} \textrm{exp}(-a_k|h|),
\end{align*}so that the cross-covariance decays (up to a constant) like the Mat\'ern covariance class with parameters $\Re(\sigma_{jk})$, $\nu_k$, and $a_k$ in the other direction.
Thus, the cross-covariance has the attractive property of reflecting the nature of the individual covariances; this relationship is also supported in Figure \ref{fig:type_pos}.
\end{remark}

\begin{remark}[Implementation] \normalfont
The function $\whittW_{\kappa, \nu}(z)$ is implemented as the \texttt{whittakerW} function in the R package \texttt{fAsianOptions} \citep{fAsianOptions}, the \texttt{mpmath.whitw} function in Python, the \texttt{whittakerW} function in MatLab, and the \texttt{WhittakerW} function in Mathematica. 
\cite{hancova_practical_2022} extensively evaluate computational infrastructure for computing \eqref{eq:whittaker} and related forms, though this computation is not always straightforward. 
For example, the \texttt{whittakerW} function in the R package \texttt{fAsianOptions} does not give proper results when $\nu_+$ is a half-integer.
One may also use the implementation of $U(\cdot, \cdot, \cdot)$ in \cite{galassi2002gnu} for computations involving $\whittW_{\kappa, \nu}(z)$.
\end{remark}

\begin{remark}[Comparison]\normalfont
A similar covariance function was studied in Section 5 of \cite{lim_tempered_2021} as a multifractional Ornstein-Uhlenbeck process. In particular, for a univariate process $X(s), s \in \mathbb{R}$, they let the exponent $\nu(s)> 0$ be H\"older continuous with $|\nu(\tilc_1) - \nu(\tilc_2)| \leq \xi|\tilc_1 - \tilc_2|^\kappa$ for constants $\xi> 0$ and $\kappa > 0$ and develop a covariance given by \begin{align*}
    &\Cov(X(\tilc_1), X(\tilc_2)) = \frac{(\tilc_1-\tilc_2)^{\nu_+(\tilc_1,\tilc_2) - \frac{1}{2}}}{\Gamma(\nu(\tilc_1) + \frac{1}{2})(2a)^{\nu_+(\tilc_1,\tilc_2) + \frac{1}{2}}}\\
    &~~~~~~~~~~\times \whittW_{\nu_-(\tilc_1,\tilc_2),  \nu_+(\tilc_1,\tilc_2)}(2a(\tilc_1-\tilc_2))
\end{align*}for $\nu_+(\tilc_1, \tilc_2) = \frac{\nu(\tilc_1) + \nu(\tilc_2)}{2}$ and $\nu_-(\tilc_1, \tilc_2) = \frac{\nu(\tilc_1) - \nu(\tilc_2)}{2}$. 
Therefore, when $a = a_j = a_k$, the new Mat\'ern cross-covariances are related to covariances of this multifractional Ornstein-Uhlenbeck process with varying parameter $\nu(s)$.
\end{remark}

\begin{remark}[Marginal cross-covariance]\normalfont
One important quantity of interest for a cross-covariance function is the marginal cross-covariance between the processes, that is, $\mathbb{E}[Y_j(0)Y_k(0)]$. 
We have established the relation \begin{align*}
    C_{jk}(0) &= \Re(\sigma_{jk}) \frac{a_j^{\nu_j}a_k^{\nu_k}}{a_+^{2\nu_+}}\frac{\Gamma(2\nu_+)}{\sqrt{\Gamma(2\nu_j)\Gamma(2\nu_k)}}.
\end{align*}
When $a_j = a_k$ and $\nu_j = \nu_k$, we have the expected $C_{jk}(0) = \Re(\sigma_{jk})$. 
However, when $a_j\neq a_k$ or $\nu_j \neq \nu_k$, we instead have $C_{jk}(0) < \Re(\sigma_{jk})$. 
Intuitively, when the processes have different behavior, having high marginal correlation between them is challenging while maintaining the validity of the model. 
The model here adapts to this automatically, while the multivariate Mat\'ern of \cite{gneiting_matern_2010} initially allows $C_{jk}(0) = \Re(\sigma_{jk})$ for all parameter values yet then needs to further constrain possible values of $\Re(\sigma_{jk})$.
\end{remark}

\begin{remark}[Mode]\normalfont
\cite{hancova_practical_2022} suggest a numerical strategy in finding the mode of the probability density function of a gamma difference distribution using its derivative; in our case, the mode of the cross-correlation function corresponds to the lag of maximal absolute correlation between the two processes. 
Although a closed-form expression for the mode does not appear to be directly available, their approach may be applied here to find the lag and strength of maximal correlation between the processes. 
\end{remark}

By simplifying the model in $d=1$ to processes such that $\Im(\sigma_{jk}) = 0$, we provide a model that links two Mat\'ern processes with a natural description of the cross-dependence. 
This new family of cross-covariance functions breaks the symmetry ($C_{jk}(h) = C_{jk}(-h)$) and diagonal-dominance ($\sup_h |C_{jk}(h)| = |C_{jk}(0)|$) assumptions of the multivariate Mat\'ern of \cite{gneiting_matern_2010} through imbalances between $\nu_j$ and $\nu_k$ or, alternatively, $a_j$ and $a_k$.
It is only necessary to estimate one additional parameter, $\Re(\sigma_{jk})$, compared to estimating the parameters of the processes independently. 
Validity of the cross-covariance model is immediately available, and the model reduces to familiar models or forms for certain parameter values. 
These properties make this model for multivariate processes more attractive in multiple aspects compared to that of \cite{gneiting_matern_2010} for the time-series setting. 

\subsection{Cross-covariances with imaginary directional measure}\label{sec:im_ent}

We now turn to cases where $\Im(\sigma_{jk})\neq 0$, which opens up an additional class of flexible cross-covariance functions. 
Here, we take the simplistic case when $\Re(\sigma_{jk})=0$, and then discuss the full model where $\Re(\sigma_{jk}) \neq 0$ and $\Im(\sigma_{jk}) \neq 0$ in Section \ref{sec:review}. 
Notice that the value of $\Im(\sigma_{jk})$ is still constrained by the self-adjoint and positive properties of $\mb{\Sigma}_H$. 
For example, one must have $2|\Im(\sigma_{jk})| \leq \sigma_{jj} + \sigma_{kk}$.

Closed-form cross-covariances for such models have been challenging to find, yet we have had success in certain situations. 
One tool we will use is the Hilbert transform of a real function $C(h)$, which we define as 
\begin{align}\label{eq:hilbert-transform}
     {\cal H}[C](h) &= \frac{1}{\pi} \int_{-\infty}^\infty \frac{C(u)}{h-u} du.
\end{align}
See \cite{king2009hilbert} for a comprehensive study of the Hilbert transform. 
Let $\mathcal{F}$ denote the Fourier transform (and $\mathcal{F}^{-1}$ its inverse), which is connected to the Hilbert transform by \citep{king2009hilbert} \begin{align*}
    \mathcal{H}[C](h) = \mathcal{F}^{-1}[-\I \textrm{sign}(\cdot )\mathcal{F}[C](\cdot)](h).
\end{align*}
Using the notation of \eqref{eq:initial_breakdown} and that $\mathcal{F}$ and $\mathcal{H}$ are linear, notice, then, that \begin{align*}
    C_{jk}^\Im(h) &= -\mathcal{F}^{-1}[-\I \textrm{sign}(\cdot )g_{jk}(\cdot )] = -\mathcal{H}[C_{jk}^\Re](h).
\end{align*}
That is, the cross-covariances with purely imaginary directional measure are the negative Hilbert transform of the cross-covariances with real directional measure.
Since a function and its Hilbert transform are orthogonal, we obtain that $\int_\mathbb{R} C_{jk}^\Re(h)C_{jk}^\Im(h) dh = 0$. 
The component $C_{jk}^\Im(h)$ thus represents dependence that is fundamentally different or opposite of $C_{jk}^\Re(h)$, opening a new class of flexibility in the cross-covariance functions, though the interpretation of the imaginary component of the cross-covariance is often unclear.

When the spectral density is purely imaginary, (for example, when $\nu_j = \nu_k$, $a_j  = a_k$, $\Re(\sigma_{jk}) = 0$, and $\Im(\sigma_{jk})\neq 0$), the cross-covariance is odd so that $C_{jk}(h) = -C_{jk}(-h)$. 
A process with odd cross-covariance was simulated in Section 4.4 of \cite{emery_improved_2016}, based on models built in  \cite{de_iaco_covariance_2003} and \cite{grzebyk_multivariate_1994}.
This is in the context of a linear model of coregionalization and thus does not allow for processes of different smoothnesses. 
Odd functions also arise from derivatives of covariance functions, thus modeling cross-covariance between a process and its derivative \citep[for example, see][]{solak2002derivative}. 
We believe we introduce here the first odd cross-covariances between Mat\'ern processes.

We first outline the cases for which we can find closed-form cross-covariances.
We mention some functions that represent the cross-covariance for some values of the parameters. 
The modified Bessel function of the first kind is defined as\begin{align}
    \besselI_{\nu}(z) &= \sum_{m=0}^\infty \frac{(\frac{1}{2}z)^{\nu + 2m}}{\Gamma(m + 1)\Gamma(\nu + m + 1)}.
\end{align}
The modified Struve function of the first kind is defined as \begin{align}
    \struveL_{\nu}(z) = \sum_{m=0}^\infty \frac{(\frac{1}{2}z)^{2m + \nu + 1}}{\Gamma(m + \frac{3}{2})\Gamma(m + \nu + \frac{3}{2})}.\label{eq:Lstruve}
\end{align}
See, for example, Section 3.7 of \cite{watson_treatise_1995} and Chapter 11 of \cite{NIST:DLMF} respectively for more information on these functions. 
We will also use the exponential integrals: \begin{align*}
    \textrm{E}_1(x) &= \int_1^\infty  \frac{e^{-tx}}{t} dt,
    & \textrm{Ei}(x) &= -\int_{-x}^\infty \frac{e^{-t}}{t} dt.
\end{align*}
For $x > 0$, the integral $\textrm{Ei}(x)$ is understood through the principal value. 

\begin{theorem}\label{thm:im_ent} 
Using the notation of \eqref{eq:initial_breakdown}, we have 
\begin{equation}\label{e:thm:im_ent}
 C_{jk}^\Im(h) = -{\rm \mathcal{H}}[C_{jk}^\Re](h),
\end{equation}
where ${\cal H}$ stands for the Hilbert transform defined in \eqref{eq:hilbert-transform}.

In particular, assuming that $\Re(\sigma_{jk}) = 0$, for several important special cases, we obtain the following closed-form expressions
of the cross-covariance:

\begin{enumerate}
    \item Suppose that $\nu= \nu_j = \nu_k > 0$ and $\nu \neq m/2$ for $m \in \mathbb{N}$, and $a = a_j = a_k > 0$. Then, the cross-covariance function based on \eqref{eq:new_cc} is written in closed form as \begin{align*}
C_{jk}(h) &= \Im(\sigma_{jk})
\frac{\pi{\rm sign}(h)}{2\cos(\pi \nu)}\frac{2^{1-\nu}}{\Gamma(\nu)}\left(a|h|\right)^{\nu}  \\
&~~~~~~~~~\times \left(\struveL_{-\nu}(a|h|) - \besselI_{\nu}(a|h|) \right).
\end{align*}
\item Suppose that $\nu_j = \nu_k = 1/2$, and let \begin{align*}
    &R(h, a_j, a_k) = \frac{-{\rm sign}(h)}{\pi}\\\
    &~~~\times \left(e^{a_j|h|}{\normalfont E}_1(a_j|h|) + e^{-a_k|h|}{\rm Ei}(a_k|h|)\right).
\end{align*}Then, the cross-covariance is written in closed form as \begin{align*}
    C_{jk}(h)   &= \Im(\sigma_{jk})\frac{(a_ja_k)^{\frac{1}{2}}}{a_+}(\mathbb{I}(h \leq 0) R(h, a_j, a_k) \\
    &~~~~~~~~~~~~~~~~~~~~~~~~~~~~+ \mathbb{I}(h > 0)R(h, a_k, a_j)).
\end{align*}Notice that if $a= a_j = a_k$, this reduces to $C_{jk}(h) = \Im(\sigma_{jk}) R(h, a, a)$.

\item Suppose that $\nu_j = \nu_k = 3/2$ and $a= a_j = a_k$. Then, the cross-covariance is written in closed form as \begin{align*}
C_{jk}(h)   &=\Im(\sigma_{jk})\bigg((a|h| + 1)R(h, a, a) \\
&~~~~~~~~~~~~~~~~~~~~~~~~~~~- \frac{2ahe^{a|h|}}{\pi} {\rm Ei}(-a|h|)\bigg).
\end{align*}

\end{enumerate}

\end{theorem}

\begin{proof}  Relation \eqref{e:thm:im_ent} has already been argued before the statement of the theorem.

We begin with the proof of Claim 1. The result for $\nu < 1/2$ follows from the Hilbert transform of $|h|^{\nu}\mathcal{K}_\nu(a|h|)$ presented as (8I.2) in Table 1.8I of Appendix 1 of \cite{king2009hilbert}. 
However, we present below the more general case that uses integral representations.

Focusing on $\Im(\sigma_{jk})\I\int_\mathbb{R}e^{\I hx} g_{jk}( x) {\rm sign}(x)  dx$, by substituting in $e^{\I x} = \cos(x) + \I\sin(x)$ and using even and odd properties of sine and cosine, we obtain
\begin{align*}
\begin{split}
    &\int_\mathbb{R}e^{\I hx} g_{jk}( x) \textrm{sign}(x)dx  \\
    &=\int_0^\infty \cos(hx) g_{jk}(x)dx - \int_0^\infty \cos(hx)g_{jk}(-x) dx \\
    & + \I\int_0^\infty \sin(hx) g_{jk}(x)dx + \I\int_0^\infty \sin(hx)g_{jk}(-x) dx. 
    \end{split}
\end{align*}
Since we assume here that $a_j = a_k$ and $\nu_j = \nu_k$, the function $g_{jk}(x) = g_{jk}(-x)$ is symmetric, and one sees that \begin{align*} &\Im(\sigma_{jk})\I\int_\mathbb{R}e^{\I hx}g_{jk}( x){\rm sign}(x)dx\\
&~~~~~~~= -2\Im(\sigma_{jk}) \int_0^\infty\sin(hx) g_{jk}(x) dx.\end{align*}

Then, we adjust this expression to see \begin{align*}
&\mathbb{E}\left[Y_j(\tilc + h) Y_k(\tilc)\right]
=-2\Im(\sigma_{jk})c_j^2\textrm{sign}(h) \\ 
&~~~~~~~~~\times \int_0^\infty \sin(|h|x)\left(a^2 + x^2\right)^{-\nu- \frac{1}{2}}dx,\end{align*}
since $(a + \I x)^{-\nu - \frac{1}{2}}(a - \I x)^{-\nu - \frac{1}{2}} = (a^2 + x^2)^{-\nu -\frac{1}{2}}$ and $c_j = c_k$.
Then, we can directly apply 3.771 (1) of \cite{GR_table_2015}, which also imposes the conditions on $\nu$. Particularly, we have a cross-covariance of \begin{align*}
   C_{jk}(h) &=-\Im(\sigma_{jk})c_j^2\sqrt{\pi}\Gamma\left(\frac{1}{2}-\nu\right)\textrm{sign}(h) \\
   &~~~~~~~~~\times \left(\frac{|h|}{2a}\right)^{\nu} \left(\besselI_{\nu}(a|h|) - \struveL_{-\nu_j}(a|h|)\right) 
\end{align*}when $\nu >0$ and $\nu \neq m/2$ for $m \in \mathbb{N}$. 

Using Euler's reflection formula for the Gamma function, we see $$\Gamma\left(\frac{1}{2} - \nu\right) = \frac{\pi}{\Gamma\left(\frac{1}{2} + \nu\right)\cos(\pi \nu)}.$$ 
Substituting the value of $c_j$ gives the final expression. 

To prove Claim 2, we use the Hilbert transform of $e^{-a|h|}$ given in (3.3) of Table 1.3 of Appendix 1 in \cite{king2009hilbert}, which is $-R(h, a, a)$. 
When $a_j = a_k = a$, the result is immediate. 
For $a_j \neq a_k$, linearity of the Hilbert transforms can be used to find the Hilbert transform of the asymmetric Laplace function.

For Claim 3 where $\nu= 3/2$ and $a = a_j = a_k$, since the Mat\'ern covariance is $\sigma^2 (1 + a|h|)e^{-a|h|}$, we find the Hilbert transform of $|h|e^{-a|h|}$, which is $a|h|R(h, a, a)- 2\pi^{-1} ahe^{a|h|}\textrm{Ei}(-a|h|)$. 
We defer details to Appendix \ref{app:hilbert_trans}. 
Combining with Claim 2 gives the final form. 
\end{proof}

As before, we remark to provide more insight regarding those formulas in 1--3 above. 
In contrast to Theorem \ref{thm:re_ent}, we substitute in the value of $c_j$ and $c_k$, so that, for Claim 1, the form is similar to that of the Mat\'ern covariance with a term of $2^{1-\nu}(a|h|)^\nu/\Gamma(\nu)$.  

\begin{remark}[Visualization and description]
In Figure \ref{fig:im_ent}, we plot the general shape of this cross-covariance function for varying $\nu = \nu_j = \nu_k$ and $a =a_j = a_k$. 
In this case, the cross-covariance function is an odd function, so that the process $\{Y_j(s)\}$ may be positively correlated with process $\{Y_{k}(s)\}$ for some lags and negatively correlated for others. 
This also implies that two processes with this cross-covariance would be uncorrelated marginally in $\tilc$. 
This is an interesting, unusual model, exhibiting a lack of marginal correlation as well as positive and negative dependence over non-zero lags.
The cross-covariance appears to be an odd function only when $\nu = \nu_j = \nu_k$ and $a = a_j = a_k$, where the real and imaginary parts of the spectral density are analytically identifiable.
\end{remark}

\begin{figure}[ht]
\centering
\includegraphics[width = .45 \textwidth]{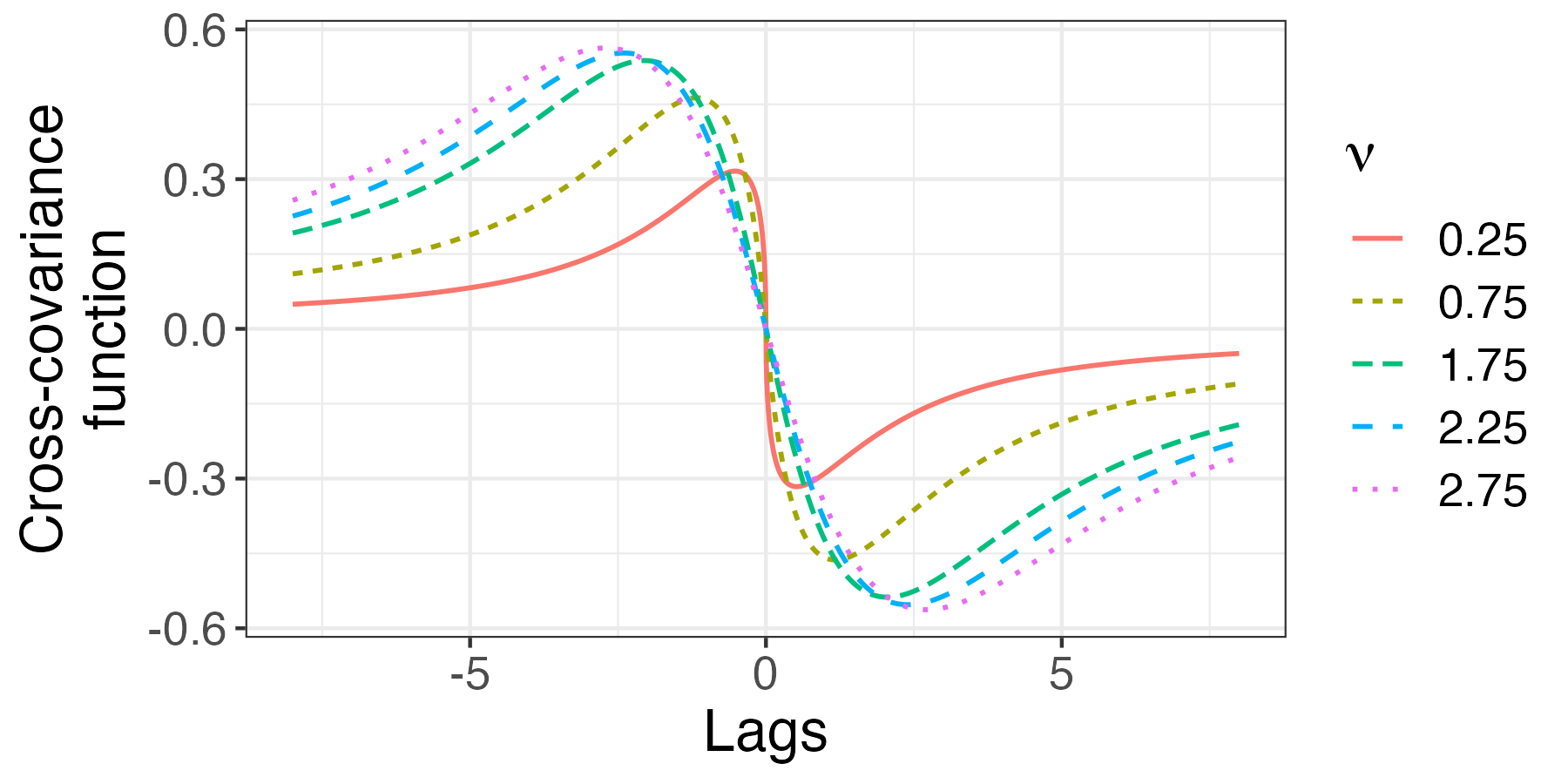}
\includegraphics[width = .45 \textwidth]{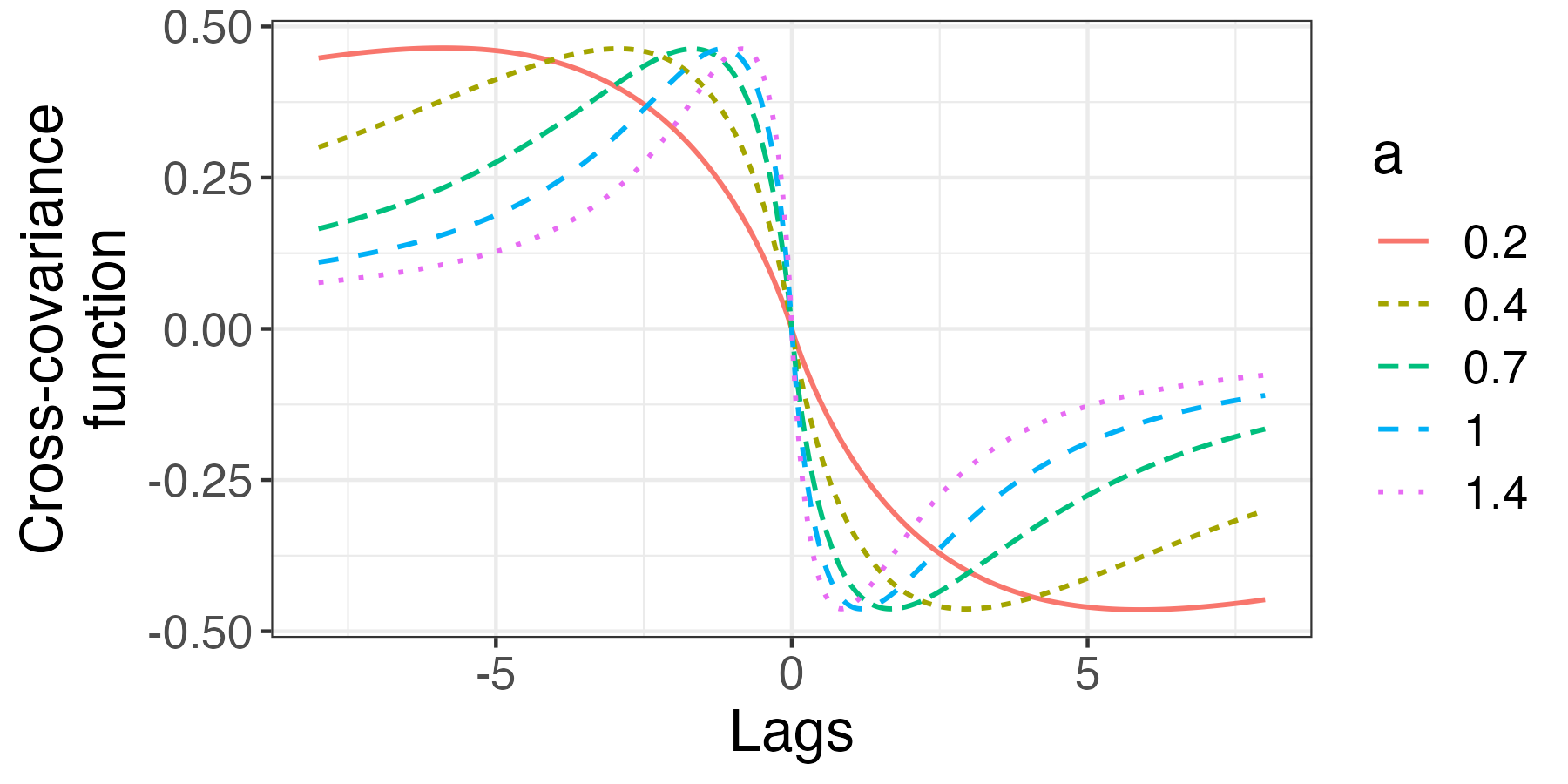}
\caption[Cross-covariance examples: $d=1$, imaginary directional measure, same parameters]{The cross-covariance function $C_{jk}(h)$ with $\sigma_{jk} = {\rm \I}$. (Top) With $a = a_j = a_k = 1$ and various values of $\nu = \nu_j = \nu_k$. (Bottom) With $\nu = \nu_j = \nu_k = 0.75$ and various values of $a = a_j = a_k$. 

}\label{fig:im_ent}
\end{figure}

\begin{remark}[Other cross-covariances when $\nu$ is a positive half-integer]\normalfont
For $\nu_j = \nu_k = 5/2, 7/2, \dots$ and $a_j = a_k$, we expect that the cross-covariances can be found through tedious algebra involving Hilbert transforms, similar to the derivation for $\nu_j = \nu_k = 3/2$ and $a_j = a_k$.
These cross-covariances can be well-defined and computed as a limit as $\nu \to m/2$ or through the spectral density representation of the cross-covariance.
\end{remark}

\begin{remark}\normalfont
Based on 12.2.6 of \cite{abramowitz_handbook_1972}, the asymptotic expression of \begin{align*}
    \besselI_{\nu}(z) - L_{-\nu}(z) \overset{z\to \infty}{\sim} \frac{2^{\nu+1}}{\pi^{\frac{3}{2}}}\Gamma\left(\nu + \frac{1}{2}\right)z^{-\nu-1}\cos(\pi\nu)
\end{align*}holds. 
The asymptotic expansion for large $|h|$ of the cross-covariance is $\Im(\sigma_{jk})$ multiplied by\begin{align}\begin{split}
&\frac{\pi\textrm{sign}(h)}{2\cos(\pi \nu)} \frac{2^{1-\nu}}{\Gamma(\nu)}\left(a|h|\right)^{\nu} \left(\struveL_{-\nu}(a|h|) - \besselI_{\nu}(a|h|)\right) 
\\&~~~~~ \overset{|h|\to \infty}{\sim} \frac{-\textrm{sign}(h)2\Gamma(\nu + \frac{1}{2})}{\sqrt{\pi} \Gamma(\nu)a |h|} .\label{eq:asympt_im}\end{split}
\end{align}The expansion suggests that the cross-covariance decays like $\Im(\sigma_{jk})\textrm{sign}(h) a^{-1}|h|^{-1}$, which is much larger in modulus compared to the Mat\'ern covariance. 
This is surprising yet supported by our implementation (see Figure \ref{fig:im_ent} as well as Figure \ref{fig:combo_ent} later).

For $\nu_j = \nu_k = 1/2$, we use the asymptotic expansion of $E_1(h) \overset{h \to \infty}{\sim} e^{-h}/h$ and $\textrm{Ei}(h)\overset{h \to \infty}{\sim} e^{h}/h$ given in 6.12 of \cite{NIST:DLMF}. We obtain
\begin{align*}
    R(h, a_j, a_k) \overset{h \to \infty}{\sim} \frac{-1}{\pi}\left(\frac{1}{a_jh} + \frac{1}{a_kh}\right),
\end{align*}so that \begin{align*}
    C_{jk}(h) \overset{h\to\infty}{\sim}\Im(\sigma_{jk})\frac{2\sqrt{a_ja_k}}{(a_j + a_k)}\frac{-1}{\pi}\left(\frac{1}{a_jh} + \frac{1}{a_kh}\right).
\end{align*}
Notice that when $a_j = a_k$, we obtain the same formula suggested by \eqref{eq:asympt_im} when $\nu = 1/2$. 
In this case, the cross-covariance decays on the order of $1/(\min\{a_j, a_k\}\cdot h)$ as $h\to\infty$.
For $\nu_j = \nu_k = 3/2$ and $a_j = a_k$, since $\textrm{Ei}(h) \overset{h \to -\infty}{\sim} e^{h}/h$, we obtain as expected \begin{align*}
    C_{jk}(h) \overset{h\to\infty}{\sim} \Im(\sigma_{jk})\frac{-4}{\pi ah}.
\end{align*}
\end{remark}

\begin{remark}[Implementation]\normalfont
The function $\besselI_{\nu}(z)$ is implemented as \texttt{besselI} in the \texttt{base} package of R, \texttt{iv} in \texttt{scipy} in Python, \texttt{besseli} in Matlab, and \texttt{besselI} in Mathematica. 
The function $\struveL_{\nu}(z)$ is implemented as \texttt{struveL} in the \texttt{RandomFieldsUtils} package of R \citep{RandomFieldsUtils}, \texttt{modstruve} in \texttt{scipy} in Python, and \texttt{StruveL} in Mathematica. 
We have also found using the series representation for $\struveL_{-\nu}(z)$ in \eqref{eq:Lstruve} works well for positive, real-valued $\nu$ and $z$.

For $\textrm{Ei}(z)$ and $E_1(z)$, we use the \texttt{expint} package of R to compute its values \citep{expint}.
\end{remark}

When taking $\Im(\sigma_{jk}) \neq 0$, we provide new closed-form cross-covariance functions in three different settings: when $a_j=a_k >0$ and $0<\nu_j = \nu_k \neq m/2$ for $m \in \mathbb{N}$; when $\nu_j = \nu_k = 1/2$; and when $\nu_j = \nu_k = 3/2$ and $a_j = a_k$.
These functions greatly improve the flexibility of Mat\'ern cross-covariances.
Still, additional range and smoothness parameters do not need to be estimated. 

The cross-covariances for other parameter settings are also of immediate interest. 
In Figure \ref{fig:im_varying_nu}, we plot examples of these cross-covariances by using fast Fourier transform approaches with their spectral densities, demonstrating that such cross-covariances exist and have interpretable shapes as one varies the parameters. 
For example, it is apparent that changing $\nu_j$ and $a_j$ will significantly alter the shape of the cross-covariance function on one half of the real line while leaving the shape of the cross-covariance function relatively intact for the other half.

For the general case, closed-form expressions of the cross-covariances are more elusive. 
We suggest, however, that one consider the generalization of the functions $\besselI_\nu(z)$ and $\struveL_\nu(z)$ to the Whittaker function $\whittM_{\mu, \nu}(z)$ and generalized modified Struve function $\struveA_{\mu, \nu}(z)$, respectively; see Section 4.4 and Chapter 5 of \cite{babister_transcendental_1967}. 
However, it is unclear if these directly correspond to our setting for general $\nu_j$ and $\nu_k$.
Furthermore, while the function $\whittM_{\mu, \nu}(z)$ is well-documented, we are unaware of any sustained research, computational formula, or implementation of the function $\struveA_{\mu, \nu}(z)$ for $\nu < 0$. 
The form of the cross-covariance when $\nu_j = \nu_k$ and $a_j = a_k$ is also related to the modified Lommel function \citep[see Equation 36 of][]{dingle_1959}, but this relation does not appear helpful in generalizing the closed-form cross-covariance functions. 

However, the spectral density represents these processes straightforwardly for general $\nu_j$, $\nu_k$, $a_j$, and $a_k$ even when closed-form representations are not available. 
Computationally, the values of the cross-covariance can be evaluated efficiently on a discrete grid of points using the fast Fourier transform \citep{cooley_algorithm_1965}, which may be considerably faster than the evaluation of the relevant special functions.

\begin{figure}[ht]
    \centering
    \includegraphics[width = .45 \textwidth]{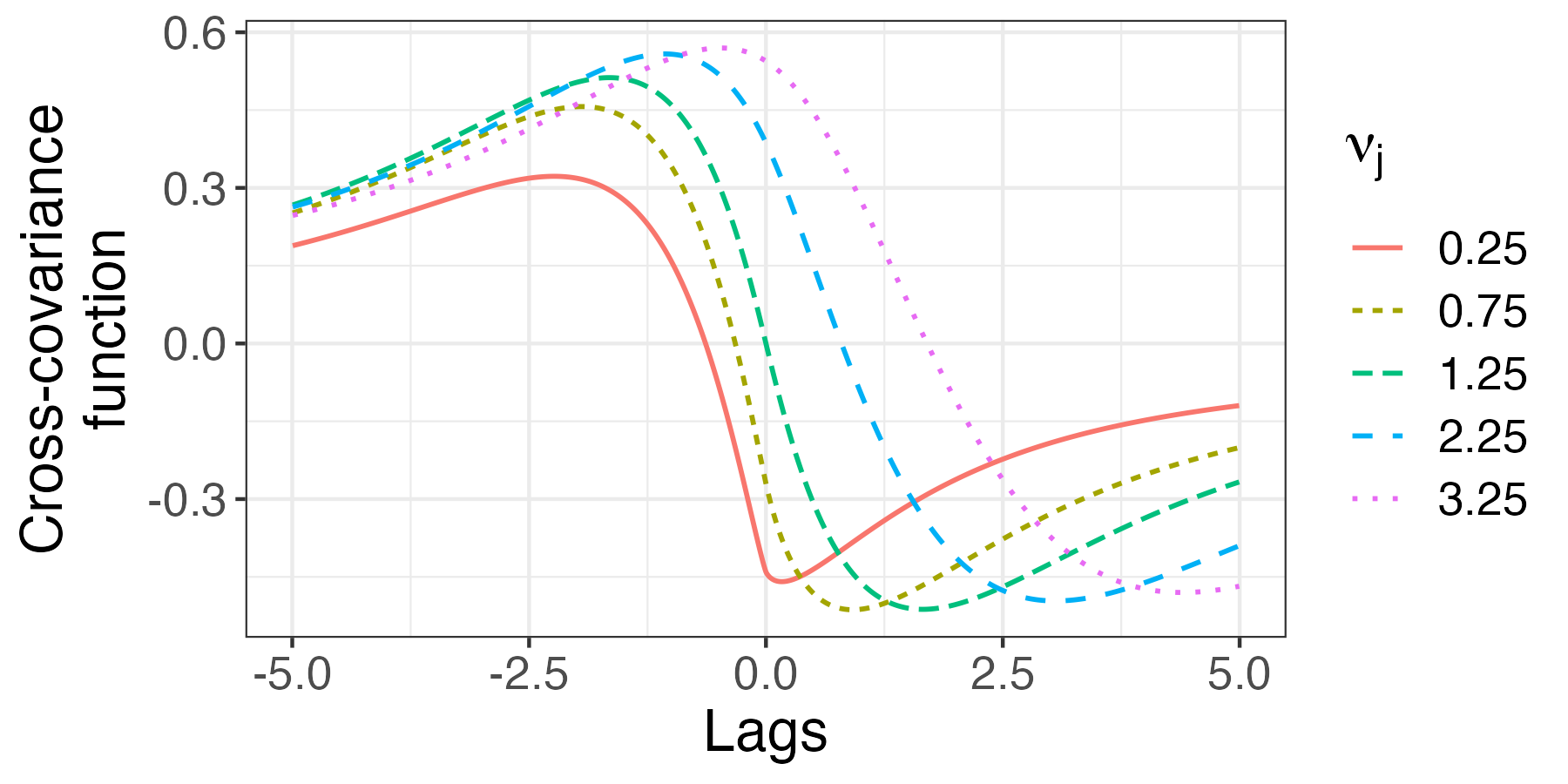}
    \includegraphics[width = .45 \textwidth]{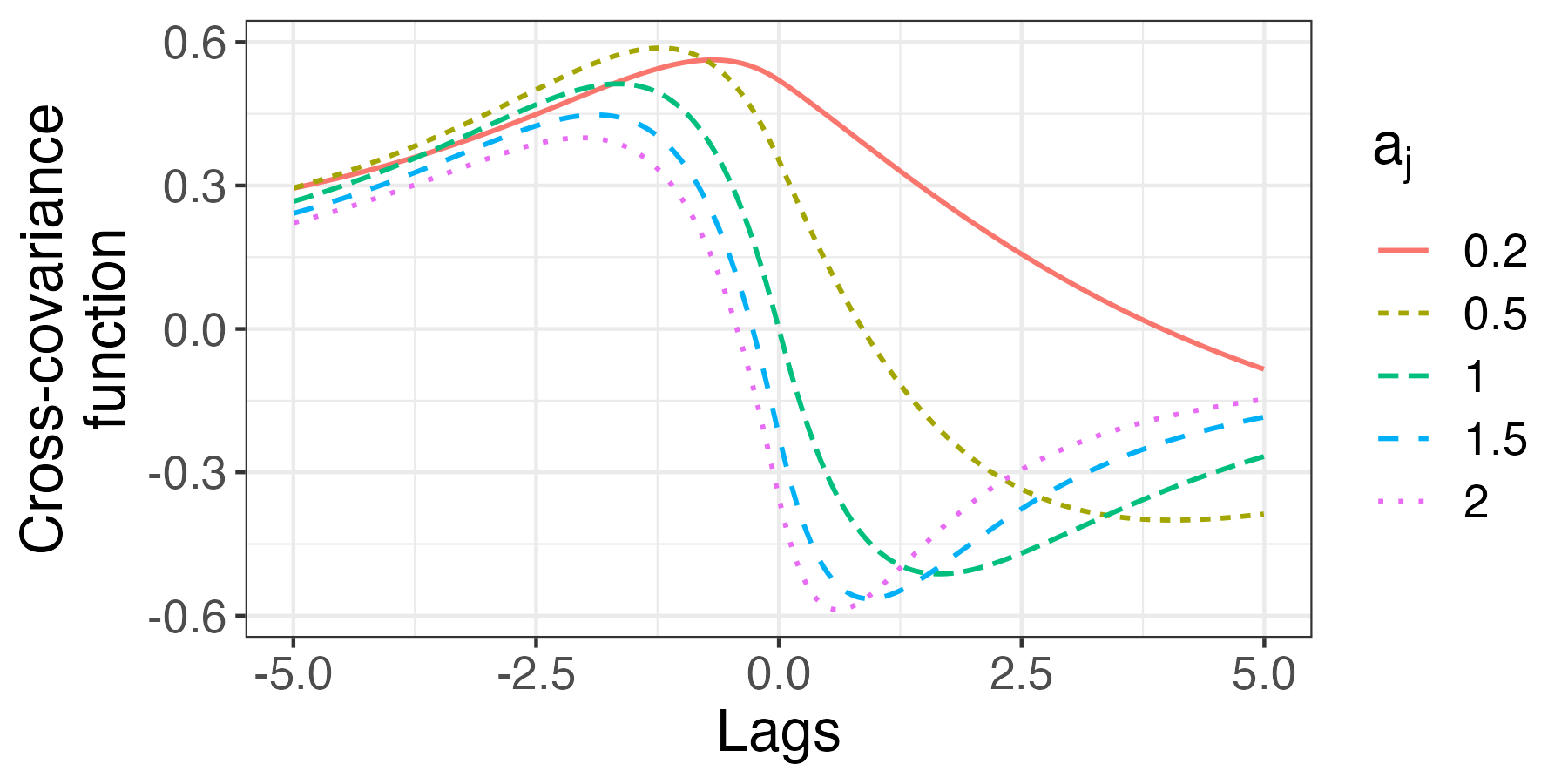}
    \caption[Cross-covariance examples: $d=1$ and imaginary directional measure]{Multivariate Mat\'ern cross-covariances with (Top) $a_j = a_k = 1$, $\nu_k = 1.25$, $\sigma_{jk} = {\rm \I}$, and $\nu_j$ varying; (Bottom) with $a_k = 1$, $\nu_j =\nu_k = 1.25$, $\sigma_{jk} = {\rm \I}$, and $a_j$ varying.}
    \label{fig:im_varying_nu}
\end{figure}

\subsection{Review of new multivariate Mat\'ern models in one dimension}\label{sec:review}

While in Theorem~\ref{thm:re_ent}, we assume $\Im(\sigma_{jk}) = 0$, and in Theorem~\ref{thm:im_ent}, we assume $\Re(\sigma_{jk}) = 0$, it is not necessary to have either be the case. 
If they are both nonzero, the cross-covariance is the sum of the respective contributions. 
Thus, to summarize, when $\mb{\Sigma}_H =[\sigma_{jk}]_{j,k=1}^p$ is positive definite and self-adjoint, we introduce a cross-covariance function represented by \begin{align*}
    \mathbb{E}[Y_j(h) Y_k(0)] &= \Re(\sigma_{jk})C_{jk}^\Re(h) + \Im(\sigma_{jk})C_{jk}^\Im(h),
\end{align*}where, after substituting $c_j$ and $c_k$ into \eqref{eq:whittaker}, \begin{align*}
    C_{jk}^\Re(h) &= \frac{2^{-\nu_+}\sqrt{2\pi}}{\sqrt{\Gamma(\nu_j)\Gamma(\nu_k)}}\frac{a_j^{\nu_j}a_k^{\nu_k}}{a_+^{\nu_+ + \frac{1}{2}}}|h|^{\nu_+ - \frac{1}{2}}e^{-ha_-} \\
    &~~~~~\times\begin{cases}
    \whittW_{\nu_-, \nu_+}(2a_+|h|)\frac{\sqrt{\Gamma(\nu_k + \frac{1}{2})}}{\sqrt{\Gamma(\nu_j + \frac{1}{2})}} & h > 0\\
    \whittW_{-\nu_-, \nu_+}(2a_+|h|)\frac{\sqrt{\Gamma(\nu_k + \frac{1}{2})}}{\sqrt{\Gamma(\nu_j + \frac{1}{2})}} & h < 0
    \end{cases}.
\end{align*}
We refer to $C_{jk}^\Im(h)$, which is available in integral form, as an extended Mat\'ern cross-correlation function for $d=1$. 
Such combinations provide a large class of potential shapes in the cross-covariance functions, with examples shown in Figure~\ref{fig:combo_ent}. 
In the specific case where $\nu = \nu_j = \nu_k \neq m/2$ for $m \in \mathbb{N}$ and $a =a_j = a_k$, the cross-covariance function becomes \begin{align*}
    &\mathbb{E}[Y_j(h) Y_k(0)] =  \frac{2^{1-\nu}}{\Gamma(\nu)}(a|h|)^{\nu}\bigg(\Re(\sigma_{jk})\besselK_{\nu}(a|h|)\\&~~~~~ + \Im(\sigma_{jk})\frac{\pi \textrm{sign}(h)}{2\cos(\pi \nu)}(\struveL_{-\nu}(a|h|) - \besselI_{\nu}(a|h|))\bigg),
\end{align*}which is decomposed as a sum of an even and an odd function. 
In the special case that $\nu_j = \nu_k = 1/2$, we obtain \begin{align*}
    &\mathbb{E}[Y_j(h) Y_k(0)] = \frac{2(a_ja_k)^{\frac{1}{2}}}{a_j + a_k}\\&\times[\Re(\sigma_{jk})\textrm{exp}\left(-|h|(a_j \mathbb{I}(h > 0) +a_k\mathbb{I}(h<0))\right) \\ & + \Im(\sigma_{jk})\left(\mathbb{I}(h \leq 0) R(h, a_j, a_k) + \mathbb{I}(h > 0)R(h, a_k, a_j)\right)].
\end{align*}
In Figure~\ref{fig:simulation}, we plot realizations of the multivariate Mat\'ern process for two different parameter settings. 
For two processes with different smoothness and real-valued $\sigma_{jk}$, one can easily pick out correlation between the processes; however, when $\sigma_{jk}$ is imaginary-valued, the dependence between the processes is harder to pick out visually since the processes are uncorrelated marginally.

\begin{figure}[ht]
\centering
\includegraphics[width = .45 \textwidth]{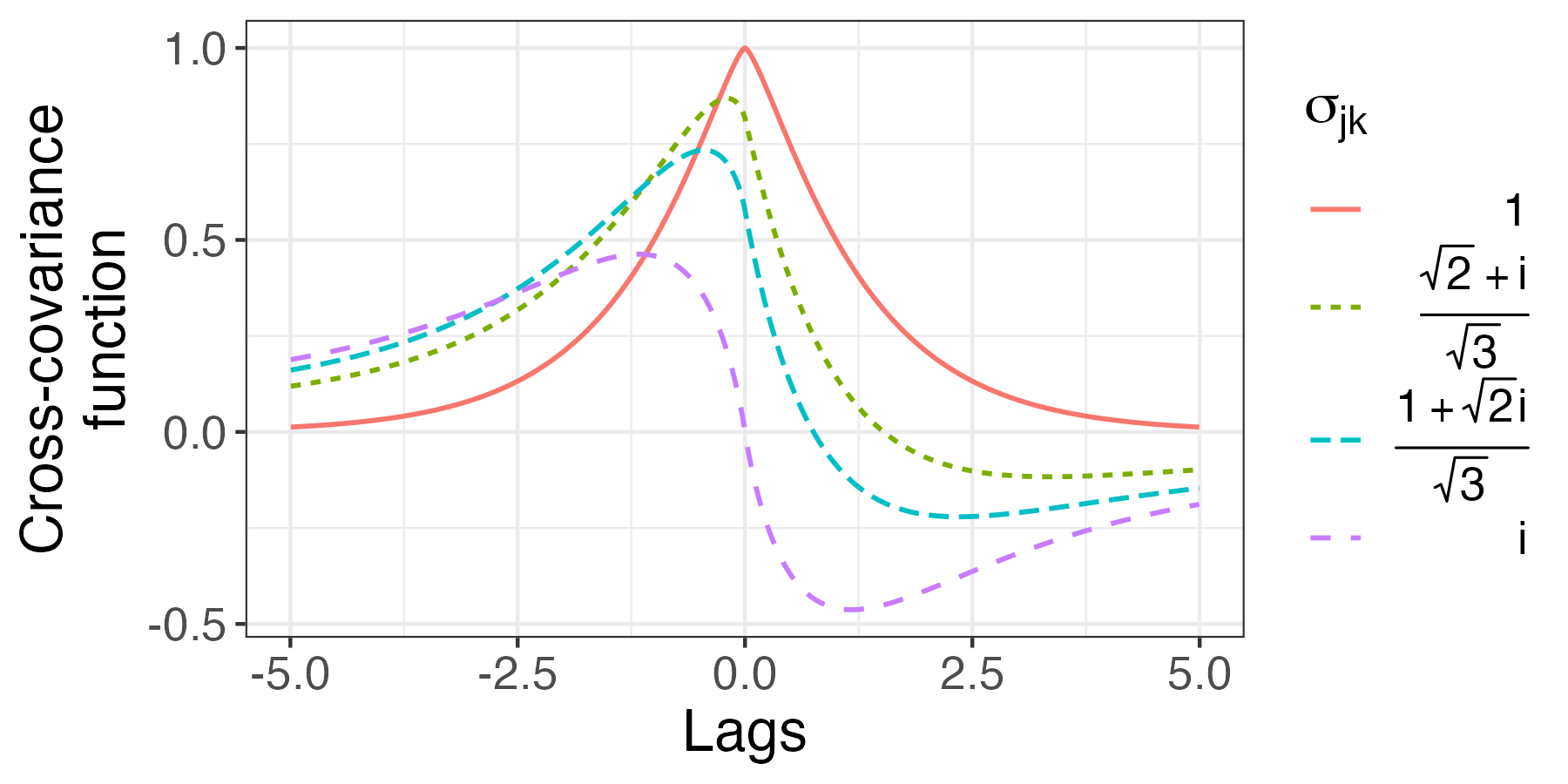}
\includegraphics[width = .45 \textwidth]{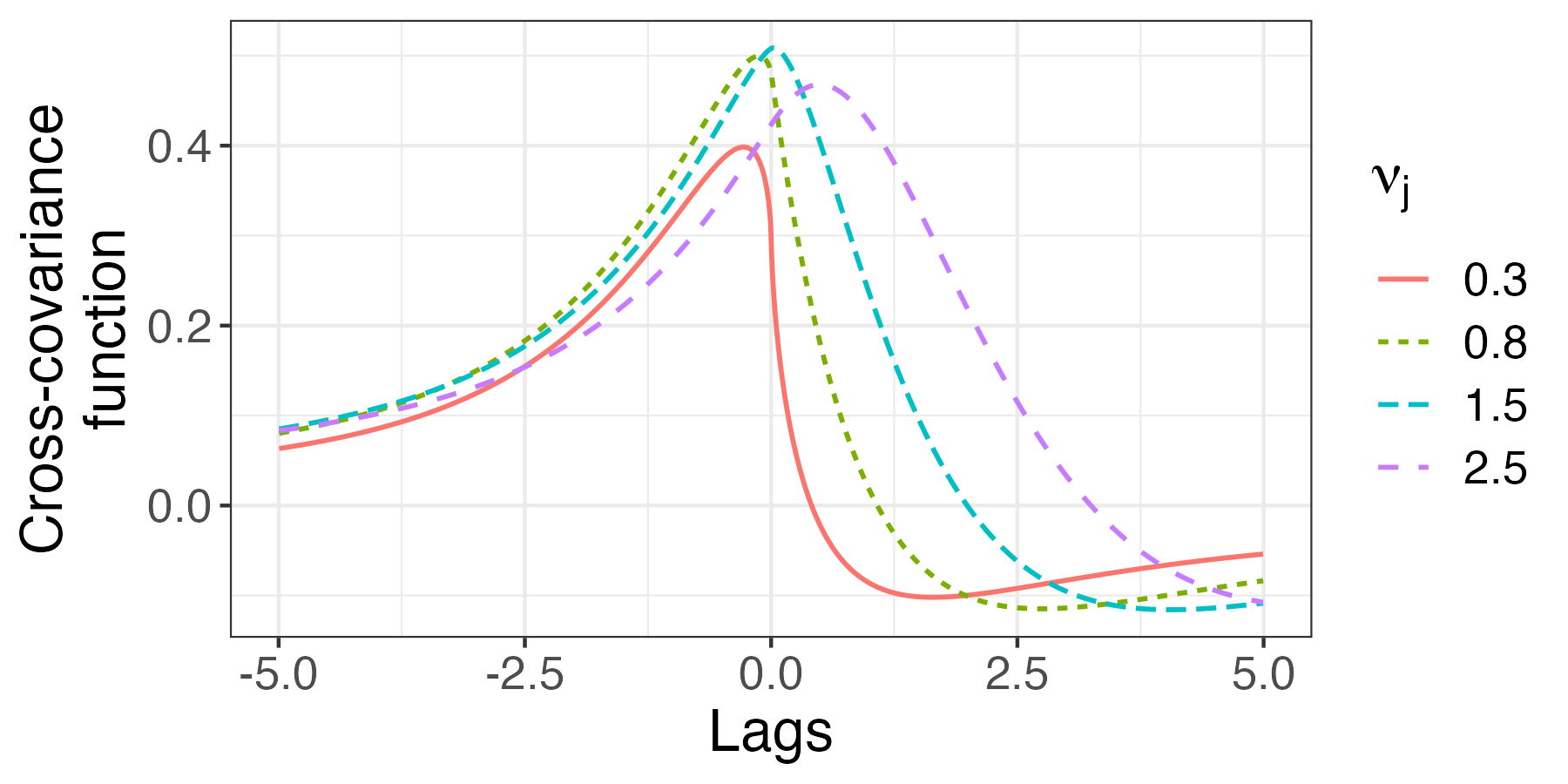}
\caption[Cross-covariance examples for $d=1$ and complex directional measure]{(Top) Multivariate Mat\'ern cross-covariances for varying values of $\sigma_{jk}$, with $d=1$, $\nu_j = \nu_k = 0.75$, and $a_j = a_k = 1$. (Bottom) Multivariate Mat\'ern cross-covariances for varying values of $\nu_j$, with $d=1$, $\sigma_{jk} = 0.5 + 0.4 {\rm \I}$, $\nu_k= 0.5$, and $a_j = a_k = 1$. }\label{fig:combo_ent}
\end{figure}

\begin{figure}[ht]
\centering
\includegraphics[width = .23 \textwidth]{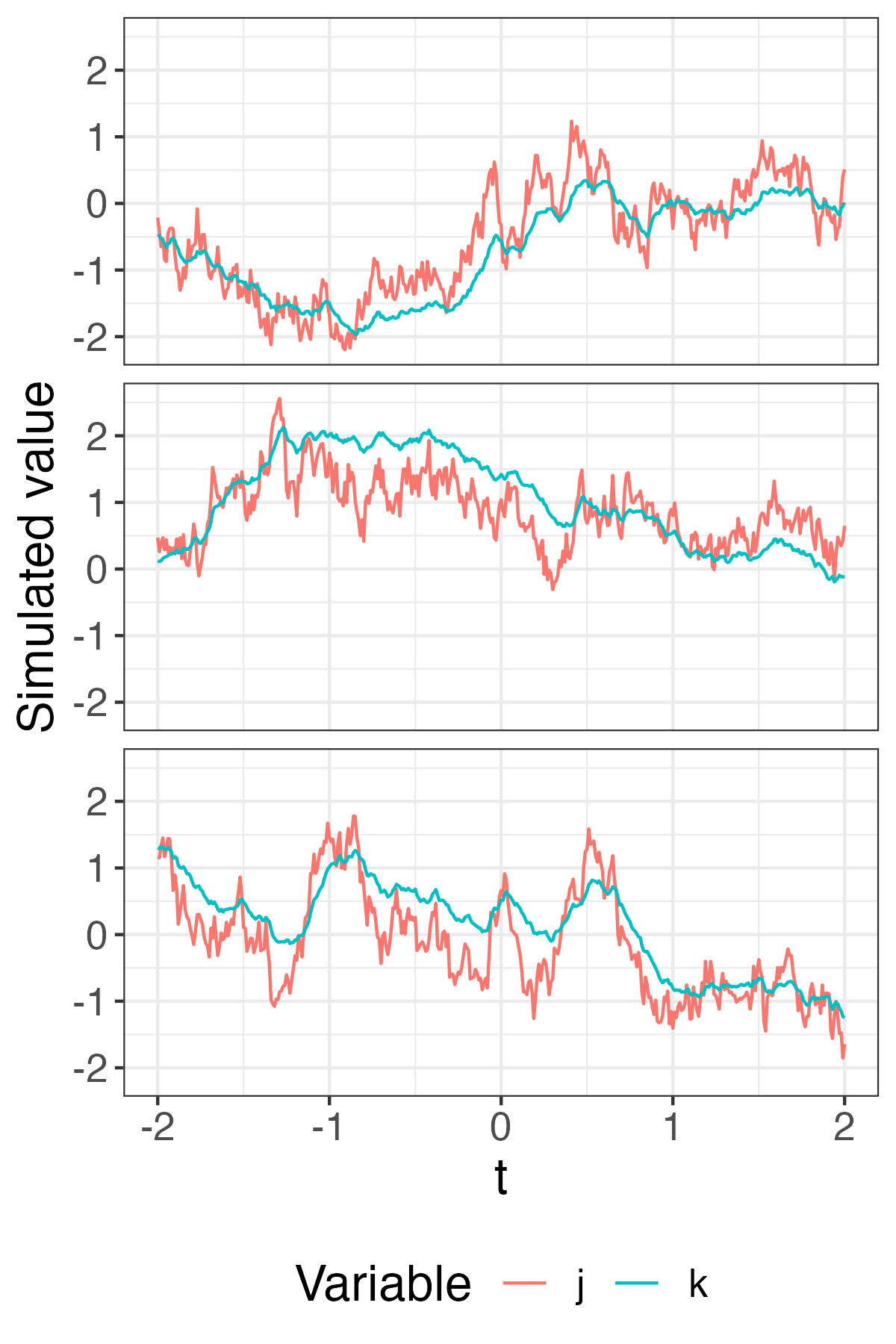}
\includegraphics[width = .23 \textwidth]{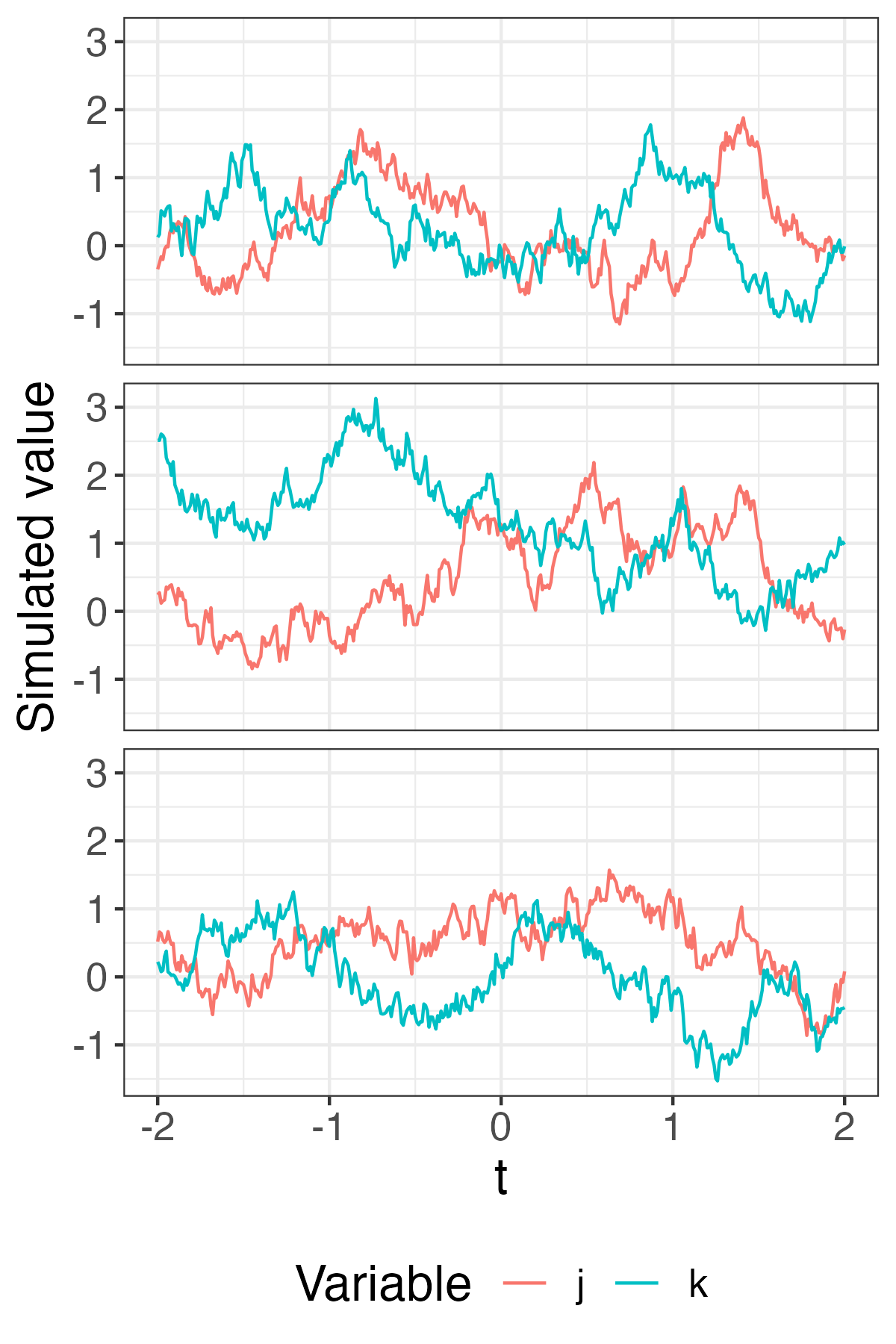}\caption[Example simulations for $d=1$]{Three realizations of bivariate Mat\'ern processes for $d=1$, with $a_j = a_k = 1$ and $\sigma_{jj} = \sigma_{kk} = 1$. (Left) For $\nu_j = 0.4$, $\nu_k = 0.8$, and $\sigma_{jk} = 0.90$. (Right) For $\nu_j = \nu_k = 0.8$ and $\sigma_{jk} = -0.95{\rm \I}$.}\label{fig:simulation}
\end{figure}

\section{Spatial extensions of multivariate Mat\'ern models}\label{sec:spatial}

In this section, we will extend our approach to the random field setting, i.e., to processes indexed by $\R^d$ for $d\ge 2$.  The spatial setting 
calls for substantially more flexible models that can accommodate anisotropy.  For a general outlook and the motivation of our approach from the perspective of tangent processes, see the Supplement. 

In the spatial setting, it is more convenient to consider integration in polar coordinates.  Namely, for $\mb{s}\in \R^d\setminus\{\mb{0}\}$,
we let $r = \|\mb{s}\|$ and $\mb{\theta}:=\mb{s}/\|\mb{s}\|$, so that $\mb{s}$ is represented with the pair of its
radial and directional components $(r,\mb{\theta}) \in (0,\infty)\times \us^{d-1}$, where $\us^{d-1}$ denotes the unit sphere in $\R^d$. 

Also, we will use matrix exponentiation and powers to describe our integrands. For all $c>0$, we let $c^{\mb{H}}:=\exp\{ \log(c) \mb{H}\}$, 
where, for a square matrix $\mb{A}$, its exponential $\exp\{\mb{A}\}$ is defined as:
$$
e^{\mb{A}}= \sum_{n=0}^\infty\frac{\mb{A}^n}{n!}.
$$
Note that the latter series converges absolutely in any matrix norm.  We will further use the definition of $\mb{A}^{\mb{B}}$ for two 
matrices of equal size as ${\rm exp}(\log(\mb{A})\mb{B})$, whenever the matrix logarithm is well-defined.  
One can use, for example, the Gregory series
\begin{equation}\label{e:Gregory-log}
\log(\mb{A}) =  \sum_{n=0}^\infty \frac{-2}{2n+1} \Big( (\mb{I}-\mb{A})(\mb{I}+\mb{A})^{-1} \Big)^{2n+1},
\end{equation}
which is convergent provided that the eigenvalues of $\mb{A}$ have positive real parts \citep[cf.][]{hingham:2008,cardoso_conditioning_2018,barradas_iterated_1994}. 
In most cases, we will apply this to a diagonal matrix $\mb{A} = \textrm{diag}(\lambda_1, \lambda_2, \dots, \lambda_p)$, in which case we have the expansion
\begin{align*}
    \log(\mb{A}) &=\textrm{diag}\left(\sum_{n=0}^\infty \frac{-2}{2n+1} \left[\frac{1- \lambda_j}{1+ \lambda_j}\right]^{2n+1}, j=1, \dots, p\right) \\
    &= \textrm{diag}\left(\log(\lambda_j), j=1, \dots, p\right).
\end{align*}




We will then define matrix versions of the parameters. 
First, define a diagonal matrix of inverse range parameters as $\amat = \textrm{diag}(a_1, \dots, a_p)$, and similarly take $\mb{c} = \textrm{diag}(c_1, \dots, c_p)$ for normalization constants $c_j > 0$; we will extend their definitions for general $d$ later. 
Let $\Nu = \textrm{diag}(\nu_1, \dots, \nu_p)$ for positive $\nu_j$.
Finally, let $\Imat_d$ again be the identity matrix of dimension $d \times d$.

\begin{remark}
    One may also consider the more general case for $\Nu$. For example, if $\Nu$ and $\amat$ commute, the definition below may be straightforwardly applicable to where $\mb{\nu} \in \mathbb{R}^{p \times p}$ is a real, symmetric, positive-definite matrix that is potentially non-diagonal.
    One case that accommodates this more-flexible $\Nu$ (and perhaps the only notable case) is where $\amat = a\Imat_p$ for $a > 0$.
\end{remark}

Following Definitions \ref{def:xi-measure} and \ref{def:Hermitian}, let $\rmeasure(dr, d\mb{\theta})$ be a Hermitian zero-mean $\mathbb{C}^p$-valued random measure, 
with orthogonal increments that satisfies:
\begin{align}
    \mathbb{E}\left[\rmeasure(dr, d\mb{\theta})\rmeasure(dr, d\mb{\theta})^*\right]&= r^{d-1}dr\measure(d\mb{\theta}).\label{eq:random_measure}
\end{align}
For concreteness, we take $\rmeasure(dr, d\mb{\theta})$ to be Gaussian, though we briefly discuss other choices in Section \ref{sec:nongaussian}.
Here, the control measure $\sd(d\mb{\theta})$ is a $\mathbb{T}_+$-valued measure on $\us^{d-1}$, which is also Hermitian, i.e., $\sd(d\mb{\theta}) = \overline{\sd(-d\mb{\theta})}$ (cf.\ Definition \ref{def:Hermitian}). 
This generalizes $\sd(dx)$ in Section \ref{sec:mm_intro}. 
Furthermore, we take $\sd(d\mb{\theta})$ to be finite so that $\left\lVert \sd(\us^{d-1})\right\rVert_{F} < \infty$ for the Frobenius norm $\lVert \cdot\rVert_{F}$.

\subsection{Proposed model}

We propose to study the multivariate Mat\'ern process generated by 
\begin{align}\begin{split}
    &\{\Y(\ti)\}_{\ti\in \mathbb{R}^d} \overset{fdd}{=} \bigg\{\int_0^\infty \int_{\us^{d-1}} e^{\I\langle \ti, r\mb{\theta}\rangle}\\
    &~~\times\mb{c}(\amat + \ang(\mb{\theta})\I r \Imat_p)^{-\Nu - \frac{d}{2}\Imat_p} \rmeasure(dr,d\mb{\theta})\bigg\}_{\ti \in \mathbb{R}^d}.\label{eq:matern_stochastic}\end{split}
\end{align}
We provide more details about the representation of \eqref{eq:matern_stochastic} and introduce the function $\ang: \us^{d-1} \to \{-1,1\}$. 
Under the assumption of diagonal $\Nu$ and $\mb{c}$, working with the exponential and logarithm definitions, we see that \begin{align*}
   & \mb{c}(\amat +\ang(\mb{\theta}) \I r\Imat_p)^{-\Nu - \frac{d}{2}\Imat_p}
    \\
    &~~~~~= \textrm{diag}\left(c_k (a_k +\ang(\mb{\theta})\I r)^{-\nu_k- \frac{d}{2}}, k=1, \dots, p\right),
\end{align*}which will simplify our analysis. 
We take $\ang: \us^{d-1} \to \{-1,1\}$ to be a function such that $\ang(-\mb{\theta}) = - \ang(\mb{\theta})$. 
The function $\ang(\cdot)$ ensures that the spectral density in \eqref{eq:matern_stochastic} is Hermitian and the cross-covariances are real 
(cf.\ Proposition \ref{p:Y-real} and Definition \ref{def:Hermitian}). 

\begin{remark} In the case $d=1$, the above model reduces to the one we have already introduced in Section \ref{Introduce_model}. Indeed, in this case
$\us^{d-1} = \us^0 = \{-1,1\}$, and the only two possible Hermitian functions are $\ang({\theta}) = \pm\textrm{sign}({\theta})$.  
Then, $\sd(d\theta) = \Re(\mb{\Sigma}_H) + \I\Im(\mb{\Sigma}_H) \textrm{sign}({\theta}) d\theta$, where $\mb{\Sigma}_H = [\sigma_{jk}]_{j,k=1}^p$. 
If $\Nu$ is also diagonal and we take $\ang({\theta}) = \textrm{sign}({\theta})$, this results in cross-covariances of \begin{align*}
 &C_{jk}(h) = c_jc_k\int_0^\infty  \int_{\us^0} e^{\I h r}(a_j + \textrm{sign}({\theta})\I r )^{-\nu_j - \frac{1}{2}} \\ &~~~~~\times \left(\Re(\sigma_{jk}) + \I\Im(\sigma_{jk}) \textrm{sign}({\theta})\right) \\ &~~~~~\times(a_k - \textrm{sign}({\theta})\I r )^{-\nu_k - \frac{1}{2}} d\theta dr \\
 &~~~~~~~~~~~~= c_jc_k\int_{-\infty}^\infty  e^{\I h r}(a_j + \I r )^{-\nu_j - \frac{1}{2}}\\ &~~~~~\times\left(\Re(\sigma_{jk}) + \I\Im(\sigma_{jk}) \textrm{sign}(r)\right) (a_k - \I r )^{-\nu_k- \frac{1}{2}} dr,
\end{align*}
which recovers the representation in \eqref{eq:d=1-C-definition} where the integral is in Cartesian coordinates on $\R$. 
When $\ang(\theta)= -\textrm{sign}(\theta)$, we obtain reversed versions of the covariances (cf.\ Proposition \ref{prop:reflected}).
\end{remark}

We next demonstrate that the marginal processes of $\Y(\ti)$ are Mat\'ern for a certain choice for $\sd(d\mb{\theta})$. 
Here and throughout, let $d\mb{\theta}$ represent the uniform probability distribution on $\us^{d-1}$. 
\begin{proposition}\label{prop:Matern}
Suppose the diagonal of $\sd(d\mb{\theta})$ has constant spectral density with respect to $d\mb{\theta}$. That is, $$
{\normalfont \textrm{diag}}(\sd(d\mb{\theta})) =\left[ \sigma_{11} d\mb{\theta}, \sigma_{22} d\mb{\theta}, \cdots, \sigma_{pp} d\mb{\theta}\right]
$$ for $\sigma_{jj} >0$. 
Define the normalization constants as \begin{align}
    c_j = 
    \frac{a_j^{\nu_j}\sqrt{\Gamma(\nu_j + \frac{d}{2})}}{\pi^{\frac{d}{4}}\sqrt{\Gamma(\nu_j)}},\label{eq:matern_normalization}
\end{align}
and let $\mb{Y}(\mb{s}):= (Y_j(\mb{s}))_{j=1}^p,\ \mb{s}\in\R^d$ be the multivariate Mat\'ern process given by \eqref{eq:matern_stochastic}.
Then, for each $j=1,\cdots,p$, the marginal process $\{Y_j(\ti)\}$ has the Mat\'ern covariance functions with inverse range parameter $a_j$, smoothness parameters $\nu_j$, 
and variance parameters $\sigma_{jj}$.
\end{proposition}

\begin{proof}
Writing $\sd(d\mb{\theta}) = [\sde_{jk}(d\mb{\theta})]_{j=1,k=1}^p$, we see\begin{align*}
    &C_{jj}(\hv) =c_j^2 \int_0^\infty \int_{\us^{d-1}} e^{\I \langle \hv, r\mb{\theta}\rangle} \left(a_j + \ang(\mb{\theta})\I r\right)^{-\nu_j - \frac{d}{2}}\\
    &~~~~~~~~\times \sde_{jj}(d\mb{\theta})\left(a_j - \ang(\mb{\theta}) \I r\right)^{-\nu_j - \frac{d}{2}} dr \\
     &=c_j^2\sigma_{jj}\int_0^\infty \int_{\us^{d-1}} e^{\I \langle \mb{h}, r\mb{\theta}\rangle} \left(a_j^2 +r^2\right)^{-\nu_j - \frac{d}{2}}r^{d-1} d\mb{\theta} dr
\end{align*}
since $(a + b\I)^d(a - b\I)^d = \left(a^2 + b^2\right)^d$ and $\ang(\mb{\theta})^2 = 1$.
Define the Bessel function \citep[see][]{watson_treatise_1995} as\begin{align*}
    \besselJ_{\nu}(z) = \sum_{m=0}^\infty \frac{(-1)^m}{\Gamma(m+1)\Gamma(m + \nu + 1)}\left(\frac{z}{2}\right)^{2m + \nu}.
\end{align*}
We next use the representation of the Bessel function  and the inversion formula on page 43 and 46, respectively, of \cite{stein_interpolation_2013} to obtain 
\begin{align}\begin{split}
    &C_{jj}(\hv) = c_j^2(2\pi)^{\frac{d}{2}}\sigma_{jj} \left\lVert \hv\right\rVert^{\frac{-d+2}{2}}\\
    &~~~~~\times\int_0^\infty \besselJ_{(d-2)/2}(r\left\lVert \hv\right\rVert)\left(a_j^2 + r^2\right)^{-\nu_j - \frac{d}{2}}r^{\frac{d}{2}} dx.\label{eq:bessel_J}\end{split}
\end{align}
Thus, the Fourier transform on $\mathbb{R}^d$ reduces to the Hankel transform (involving the Bessel function $J_\nu(z)$) on $\mathbb{R}$. 
The representation \eqref{eq:bessel_J} is proportional to the Mat\'ern covariance in \eqref{eq:matern_form}; see the representation of the Mat\'ern spectral density on page 49 of \cite{stein_interpolation_2013}. 

Formally showing this, we use 6.565 (4) of \cite{GR_table_2015} to obtain \begin{align*}
     C_{jj}(\hv) &= c_j^2 2^{1 - \nu_j}\pi^{\frac{d}{2}}\sigma_{jj} \left\lVert \hv\right\rVert^{\nu_j} \frac{a_j^{-\nu_j}}{\Gamma(\nu_j + \frac{d}{2})}\besselK_{\nu_j}(a_j \left\lVert \hv\right\rVert).
\end{align*}
Then, using the expression $\besselK_{\nu}(z) \sim 2^{-1}\Gamma(\nu)(z/2)^{-\nu},$ as $z\to 0$ \citep[see Section 10.30 of][]{NIST:DLMF}, we see  \begin{align*}
    C_{jj}(\mb{0}) &= c_j^2\pi^{\frac{d}{2}}\sigma_{jj}  \frac{\Gamma(\nu_j)}{a_j^{2\nu_j}\Gamma(\nu_j + \frac{d}{2})}.
\end{align*}Therefore, from our choice of $c_j$, we obtain $C_{jj}(\mb{0}) = \sigma_{jj}$ and \begin{align*}C_{jj}(\hv)&=\sigma_{jj} \frac{2^{1 - \nu_j}}{\Gamma(\nu_j)} \left(a_j\left\lVert \hv\right\rVert\right)^{\nu_j} \besselK_{\nu_j}(a_j \left\lVert \hv\right\rVert)\\
&= \mathcal{M}( \lVert\hv\rVert; a_j, \nu_j, \sigma_{jj}).\end{align*}
This completes the proof. 
\end{proof}

The Mat\'ern-type model defined in \eqref{eq:matern_stochastic} depends on the choice of the {\em directional measure} $\sd(d\mb{\theta})$.
This, in the spatial setting $(d\ge 2)$ leads to a great amount of flexibility, since $\sd(d\mb{\theta})$ is in fact an 
infinite-dimensional parameter. Therefore, in contrast to the case $d=1$, one cannot expect to obtain a canonical ``spatial extension'' 
of the classical Mat\'ern model even in the scalar-valued regime $p=1$.  In what follows, we will offer several natural approaches, where
we begin with more basic models for $\sd(d\mb{\theta})$ and then gradually build more complexity in the model.


Since the integrals above are not often available in closed form for most parameter settings, we approximate the integrals for the plots in this section, and Fourier transform approaches can be used to do this quickly \citep{averbuch_fast_2006}.
Furthermore, the processes can be simulated efficiently. 
Throughout this section, we use the simulation approach of \cite{emery_improved_2016} as recommended in \cite{alegria_bivariate_2021}, using the spectral density and importance sampling to simulate the multivariate process.

\subsection{Cross-covariances with real directional measure}

Here, we consider the case where $\sd(d\mb{\theta}) = \Sig d\mb{\theta}$ for some self-adjoint and positive-definite matrix $\Sig = [\sigma_{jk}]_{j,k=1}^p$. 
In contrast to $\Sig_H$ of Section \ref{Introduce_model}, the matrix $\Sig$ must be a real-valued matrix in this case, due to the requirement $\sd(d\mb{\theta}) = \overline{\sd(-d\mb{\theta})}$. 
In this setting, we obtain a cross-covariance of \begin{align*}
    &C_{jk}(\mb{h}) = \sigma_{jk}c_jc_k\int_{-\infty}^\infty \int_{\us^{d-1}, \ang(\mb{\theta}) = 1} e^{\I \langle \hv, r\mb{\theta}\rangle} \\&~~~~~~~\times\left(a_j + \I r\right)^{-\nu_j - \frac{d}{2}}\left(a_k - \I r\right)^{-\nu_k - \frac{d}{2}}|r|^{d-1}d\mb{\theta} dr.
\end{align*}
This integral can be somewhat simplified if one assumes that $\{\mb{\theta} | \ang(\mb{\theta}) = 1\} = \{\mb{\theta} | \langle \mb{\theta}, \mb{\theta}^*\rangle >0\}$ for some fixed $\mb{\theta}^* \in \us^{d-1}$ and the standard Euclidean inner product $\langle \cdot, \cdot \rangle$. The entire spectral density is only real in the case of $a_j = a_k$ and $\nu_j = \nu_k$, for which the cross-covariance is proportional to a Mat\'ern covariance.
However, finding closed-form expressions for the integral in full generality is not straightforward. 
We show simulated processes and their cross-covariances in Figure~\ref{fig:simulated_spatial_constant}. 

\begin{figure}[ht!]
\centering
 \includegraphics[width = .30\textwidth]{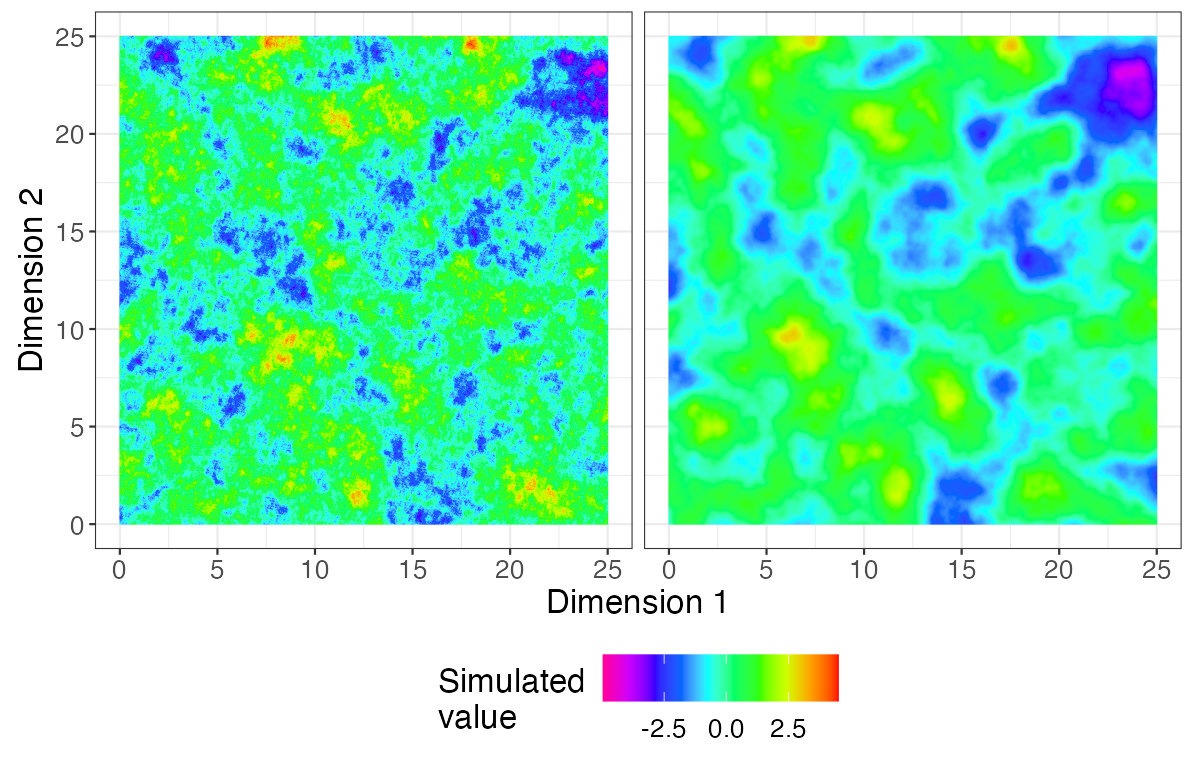}
\includegraphics[width = .15\textwidth]{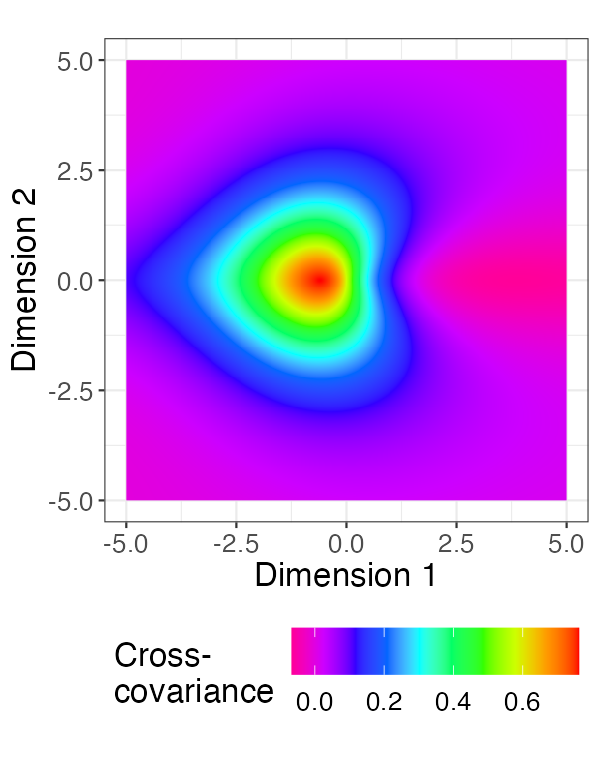}

 \includegraphics[width = .30\textwidth]{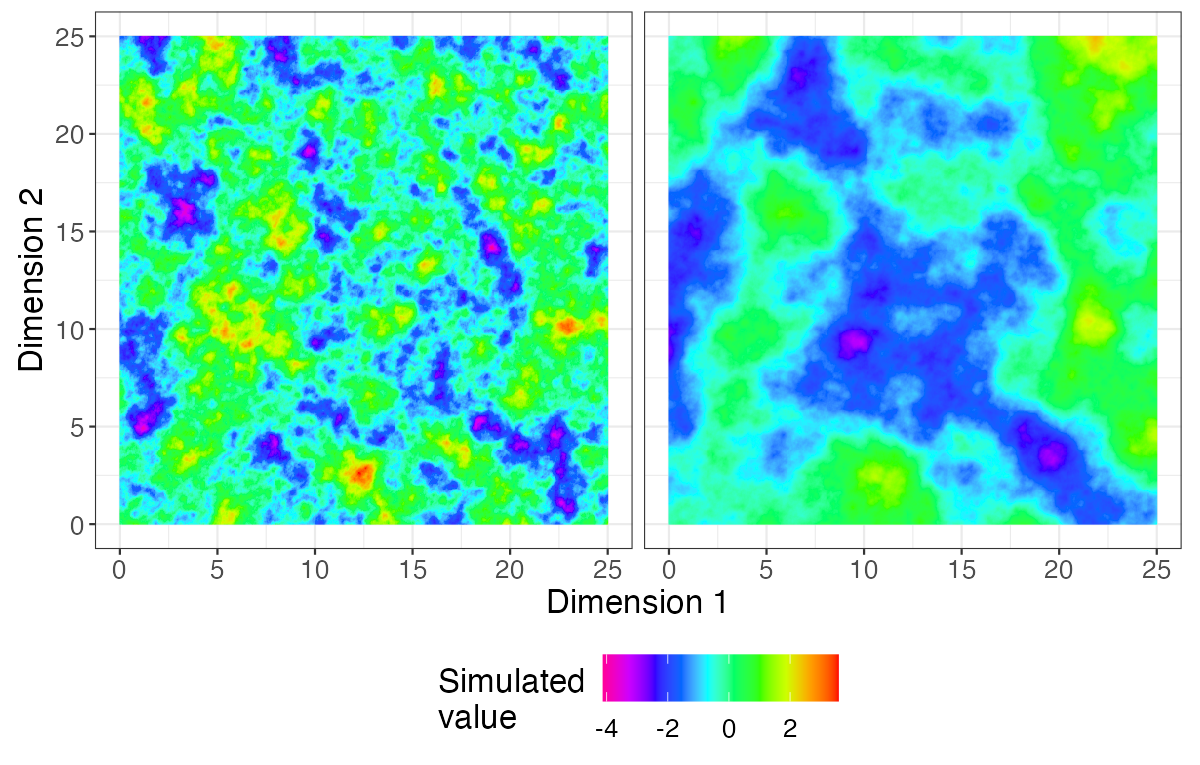}
\includegraphics[width = .15\textwidth]{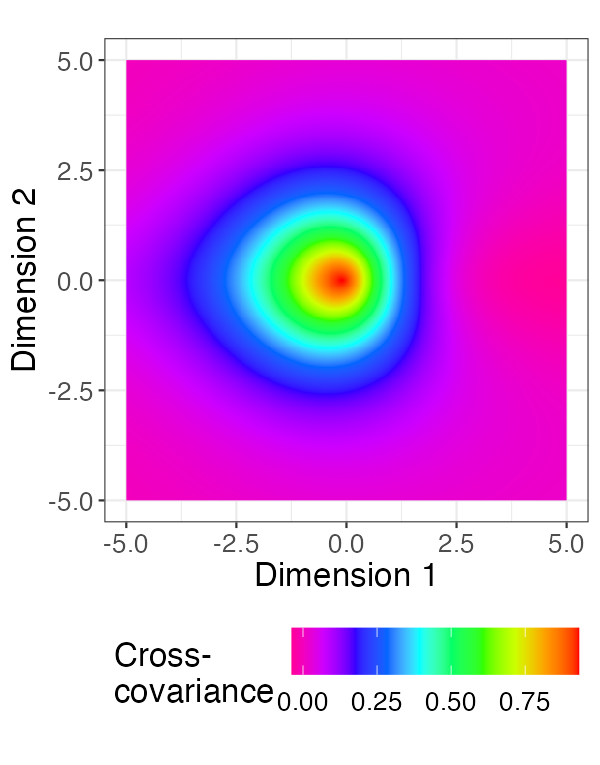}

 \includegraphics[width = .30\textwidth]{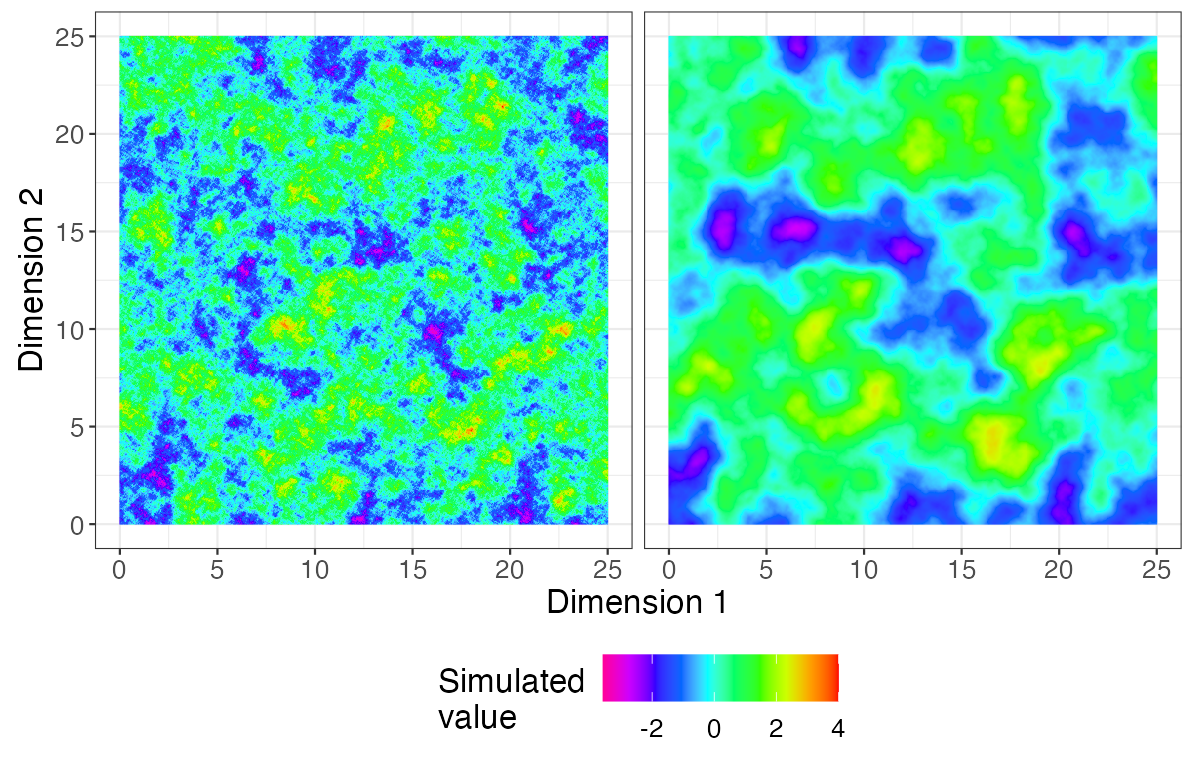}
\includegraphics[width = .15\textwidth]{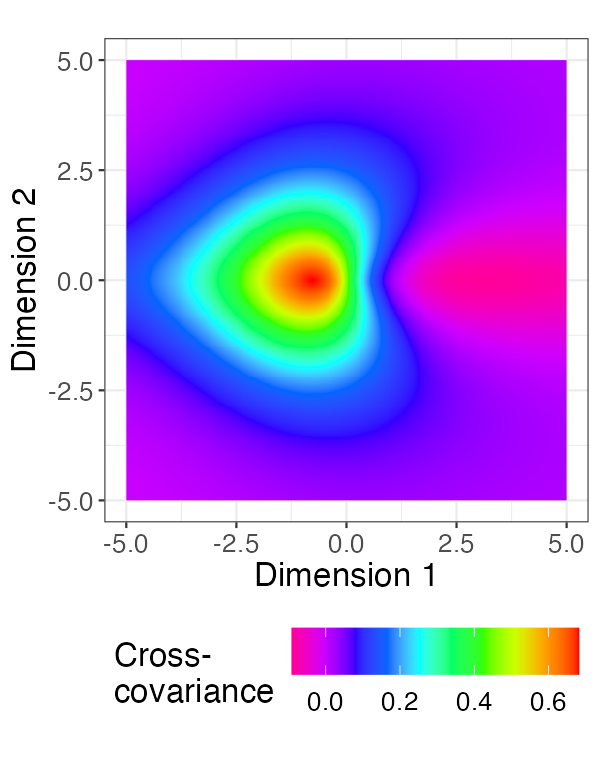}

\caption[Cross-covariance examples for $d=2$ and real directional measure]{Simulated realizations of process $j$ (Left) and Process $k$ (Middle) over a grid from $0$ to $25$; (Right) their cross-covariance function. 
For each process, we take $\sd(d\mb{\theta}) = \Sig d\mb{\theta}$, $\sig_{jj} = \sig_{kk} = 1$, $\sig_{jk} = 0.85$, $\ang(\mb{\theta}) = \textrm{sign}(\theta_1)$ where $\theta_1$ is the first entry of $\mb{\theta}$.
(Top) $a_j = a_k = 1$, $\nu_j = 0.5$, $\nu_k = 1.5$.
(Middle) $a_j = 1.2$, $a_k = 0.8$, and $\nu_j = \nu_k = 1$. 
(Bottom) $a_j = 1.1$, $a_k = 0.9$, and $\nu_j = 0.5$, and $\nu_k = 1.5$.} 
\label{fig:simulated_spatial_constant}
\end{figure}

\subsection{Cross-covariances with imaginary or varying directional measure}\label{sec:d2vary}

Next, we explore the flexibility that $\sd(d\mb{\theta})$ provides and the resulting cross-covariances.
We first outline one particular area where the result is directly comparable to the $d=1$ case. Suppose that $\sde_{jk}(d\mb{\theta}) = \textrm{sign}(\theta_1)\I d\mb{\theta}$, where $\theta_1$ is the first entry of the vector $\mb{\theta}$.
We choose this axis arbitrarily, and such a formulation can be extended to the conditions of $\langle \mb{\theta}^*, \mb{\theta}\rangle >0$ and $<0$, and so forth for some fixed vector $\mb{\theta}^*$. 
Finally, let $\nu = \nu_j = \nu_k>0$, $\nu \notin m/2$ for $m \in \mathbb{N}$, and $a = a_j = a_k>0$, so that we are left with a cross-covariance of \begin{align*}
    C_{jk}(\hv) &= c_jc_k \int_{0}^\infty \int_{\us^{d-1}} e^{\I \langle \hv, r\mb{\theta}\rangle} \left(a^2 + r^2\right)^{-\nu - \frac{d}{2}} \\&~~~~~~~~ \times \I \textrm{sign}(\theta_1) r^{d-1} d\mb{\theta} dr.
\end{align*}
First, take $\hv$ to have first entry $0$ so that we may write $\hv\in \textrm{span}\{\mb{e}_2,\mb{e}_3,\dots, \mb{e}_d\}$ where $\mb{e}_q$ is the $q$-th standard basis function.
Then, as the integrand is an odd function of $\theta_1$, $$\int_{\us^{d-1}} e^{\I\langle \hv, r\mb{\theta}\rangle}\textrm{sign}(\theta_1)d\mb{\theta}=0.$$
This implies an axis of reflection across $\mb{e}_1$ where $C_{jk}(\hv) = 0$ for any $\hv$ with first entry $0$. 
Alternately, when $\hv$ points in the orthogonal direction, that is, $\hv = b\mb{e}_1$ for $b \in \mathbb{R}$, we outline in the Supplement how then  $$C_{jk}(b\mb{e}_1)\propto \textrm{sign}(b)\left|ba_j^{-1}\right|^{\nu_j} \left(\struveL_{-\nu_j}(a_j|b|) - \besselI_{\nu_j}(a_j|b|)\right).$$ 
The cross-covariance in this direction has the same form as in the $d=1$ case of imaginary entries, a natural generalization to the $d > 1$ case. 
We also confirm this result numerically in the code accompanying this paper. 
An example of this cross-covariance is plotted in the top panel of Figure~\ref{fig:simulated_spatial}. 
To develop this closed form to general $\hv$, we have explored incomplete cylindrical functions as described in \cite{agrest_general_1971}, yet this does not appear to properly generalize the results here.

The above is just one example of the flexibility that the measure $\sd$ provides, and there is a broad array of potential symmetries in the domain of the process. 
\cite{didier_domain_2018} provide more insight about the possibilities, characterizing all domain and range symmetries in the case $p = d = 2$ for operator fractional Brownian fields. 
As one example, we plot another cross-covariance with more complicated $\sd(d\mb{\theta})$ in the bottom panel of Figure~\ref{fig:simulated_spatial}.

\begin{figure}[ht!]
\centering
 \includegraphics[width = .30\textwidth]{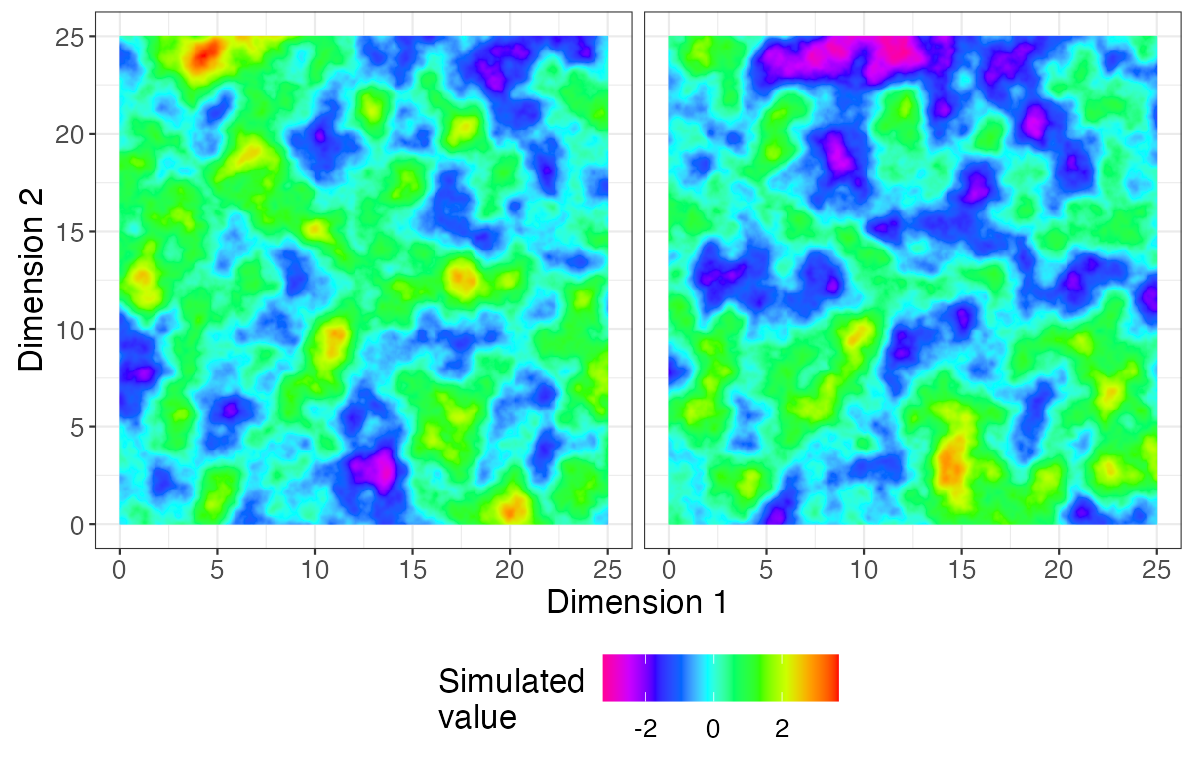}
\includegraphics[width = .15\textwidth]{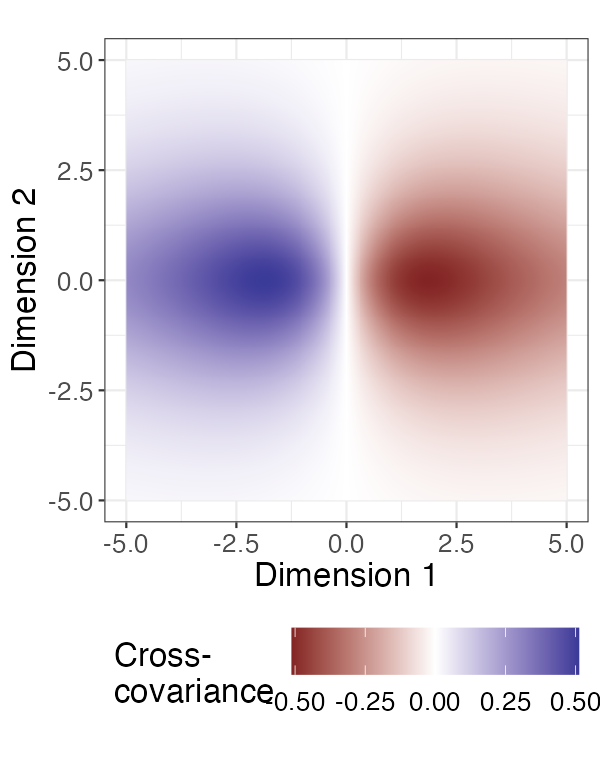}

 \includegraphics[width = .30\textwidth]{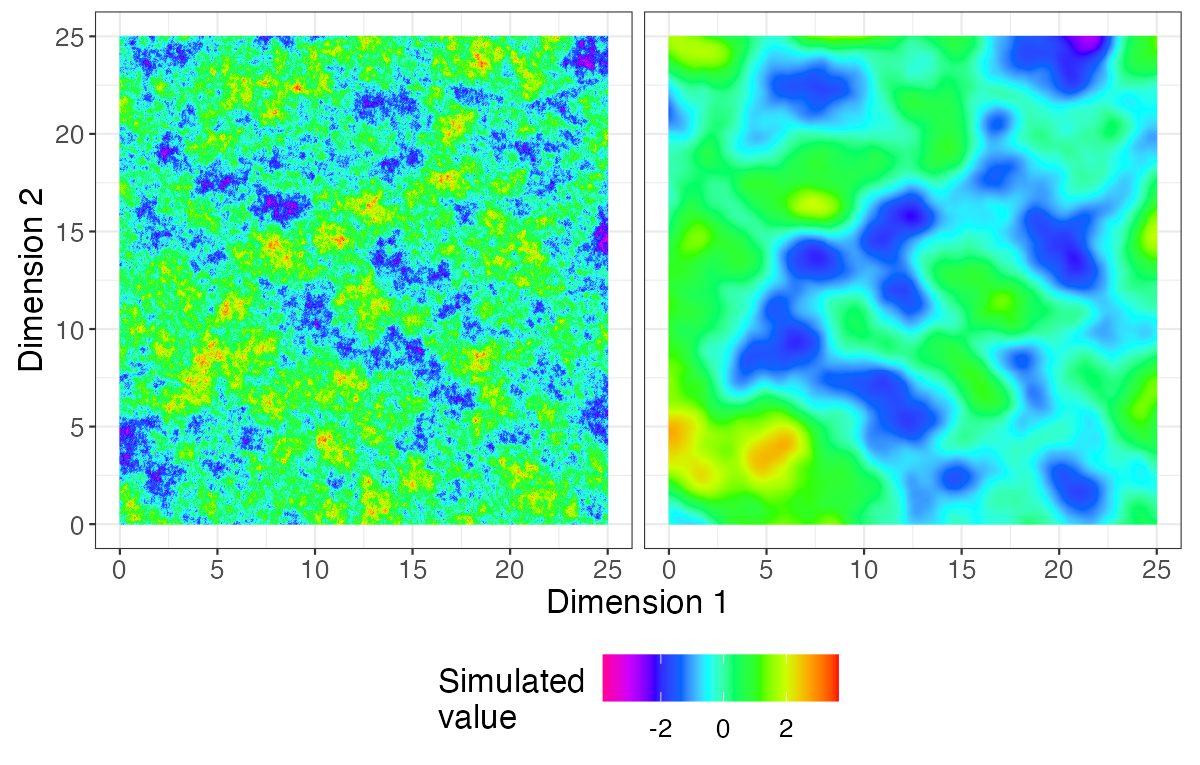}
\includegraphics[width = .15\textwidth]{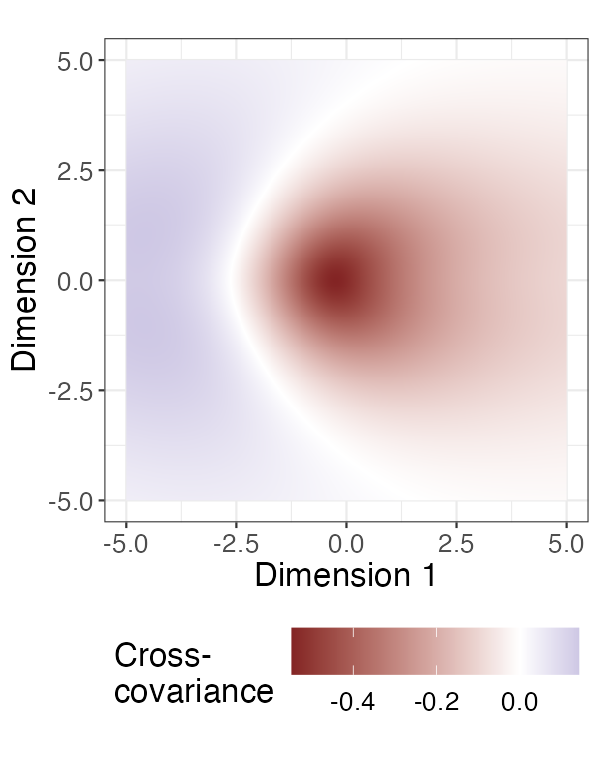}

 \includegraphics[width = .30\textwidth]{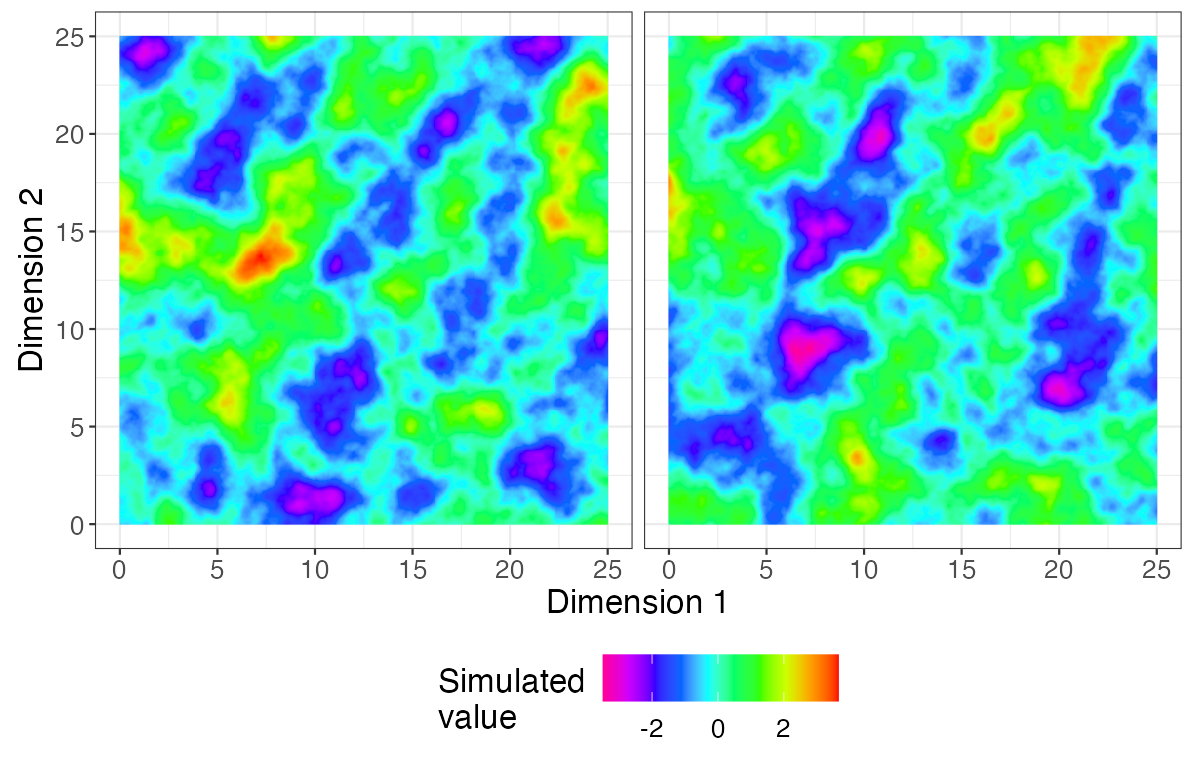}
\includegraphics[width = .15\textwidth]{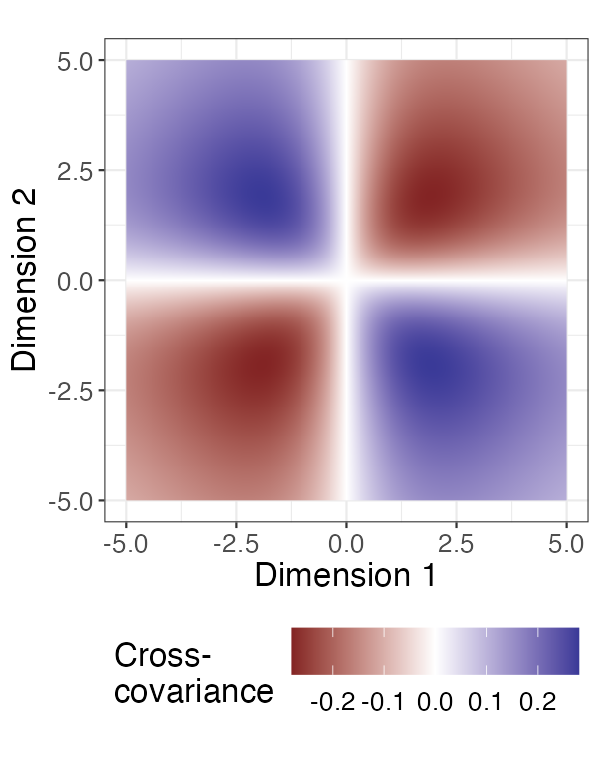}

\caption[Cross-covariance examples for $d=2$ and imaginary directional measure]{Simulated realizations of process $j$ (Left) and Process $k$ (Middle) over a grid from $0$ to $25$; (Right) their cross-covariance function. 
For each process, we take $a_j = a_k = 1$, $\ang(\mb{\theta}) = \textrm{sign}(\theta_1)$, $\sde_{jj}(d\mb{\theta}) = \sde_{kk}(d\mb{\theta}) =   d\mb{\theta}$, and normalization A.
(Top) With $\nu_j = \nu_k= 1.5$,  $\sde_{jk}(d\mb{\theta}) = 0.97\I\textrm{sign}(\theta_1)d\mb{\theta}$. 
(Middle) With $\nu_j = 0.4$, $\nu_k = 2.5$, and   $\sde_{jk}(d\mb{\theta}) = 0.97\I\textrm{sign}(\theta_1)d\mb{\theta}$. 
(Bottom)  With $\nu_j = \nu_k = 1.5$ and $\sde_{jk}(d\mb{\theta}) = 0.97\left(\textrm{sign}(\theta_1) \times \textrm{sign}(\theta_2)\right) d\mb{\theta}$.} 
\label{fig:simulated_spatial}
\end{figure}

By considering general complex forms for $\sd(d\mb{\theta})$, we introduce a wide class of spatial cross-covariance functions for $d\geq 2$. 
They are a natural extension of new multivariate Mat\'ern models when $d=1$. 
There is a wide amount of model flexibility, especially with respect to the input domain. 
For practitioners, it may be helpful to understand which type of asymmetries are likely and impose relevant restrictions on the functions $\sd(d\mb{\theta})$ and $\ang(\mb{\theta})$. 
These computations for such a flexible model are ameliorated by computational approaches for multivariate Fourier transforms on Euclidean or polar grids; for example, see \cite{averbuch_fast_2006}. 
\section{Simulation Studies and Computation}\label{sec:computation_simulation}

In this section, we provide simulation studies 
focusing on estimation, testing, and prediction for the imaginary component of the proposed models. 
We consider the settings of $d=1$ and $d=2$ in Sections \ref{sec:sim_time_series} and \ref{sec:sim_spatial}, respectively.

\subsection{The $d=1$ case}\label{sec:sim_time_series}
We consider a bivariate Mat\'ern process observed on $n=300$ points uniformly generated in $[0,1]$, with both coordinates observed on the same points. 
For the first coordinate of the process, we take $\nu_1 = 0.5$, $a_1 = 8$, and $\sigma_{11} = 1$ and 
for the second coordinate, we take $\nu_2 = 0.75$, $a_2 = 12$, and $\sigma_{22} = 1$. 
This evaluates performance for two processes with different parameters.
For the cross-covariances, we generate from three different models:
\begin{itemize}
    \item $\sigma_{12} = 0.4$,
    \item $\sigma_{12} = 0.4\mathbbm{i}$,
    \item $\sigma_{12}= 0.4 + 0.4 \mathbbm{i}$,
\end{itemize}to evaluate the variance parameterizations introduced in Section \ref{Introduce_model}.
For each setting, we compute $200$ independent realizations.  For each realization of the process, we consider {\it estimating} two different models from this data, 
based on the model in Section \ref{Introduce_model} with parameters $\nu_1$, $\nu_2$, $a_1$, $a_2$, $\sigma_{11}$, $\sigma_{22}$, $\Re(\sigma_{12})$, and (potentially) 
$\Im(\sigma_{12})$.  The first takes models in Section \ref{Introduce_model} with $\sigma_{12}$ complex, so that both $\Re(\sigma_{12})$ and $\Im(\sigma_{12})$ are estimated.
The second assumes $\sigma_{12} \in \mathbb{R}$ and $\Im(\sigma_{12}) =0$. 
In the supplement, we also consider the alternate factorization in \cite{bolin_multivariate_nodate} and the multivariate Mat\'ern of \cite{gneiting_matern_2010}. 
We use maximum likelihood estimation and numerically optimize the likelihood using the L-BFGS-B routine in R \citep{byrd_limited_1995}.

First, we consider a hypothesis testing problem for $\sigma_{12} \in \mathbb{R}$. Specifically, we consider: 
\begin{align*}
    H_0&: \Im(\sigma_{12}) = 0, \\
    H_1&: \Im(\sigma_{12}) \neq 0.
\end{align*}
Since the null hypothesis is nested in the more general, unconstrained model $\Im(\sigma_{12})\in \R$, the standard likelihood ratio 
testing methodology applies.  Specifically, we evaluate the statistic  
$\lambda := 2\log({\cal L}_1/{\cal L}_0)$, where ${\cal L}_0$ and ${\cal L}_1$ are the likelihoods for the constrained ($\Im(\sigma_{12})=0$) and 
unconstrained models. Appealing to the classical asymptotic theory, under the null hypothesis the $p$-value is approximated by
$\P[\chi > \lambda]$, where $\chi$ follows the $\chi^2$ distribution with $1$ degree of freedom.  
Thus, we evaluate a likelihood ratio test between their real and complex versions, using the $\chi^2_1$ distribution with a set threshold for $p$-values of $0.050$. 
At the nominal level $0.050$, the test rejected $0.030$ of the 200 simulations under the null hypothesis, demonstrating Type I error control. 
At this level, the test power was $0.660$ when $\sigma_{12} = 0.4\I$ and $0.585$ when $\sigma_{12} = 0.4 + 0.4\I$, rejecting the null hypothesis over the majority of simulations. 
The empirical cumulative distribution functions of the $p$-values are plotted in Figure \ref{fig:lrt}. 
When the null hypothesis is true, the $p$-values are approximately uniform (Kolmogorov-Smirnov $p$-value 0.1725) and lie close to the line $y=x$, again indicating decent type I error control.
Under the two alternative hypotheses considered, the substantial power is demonstrated by the deviations from the line $y=x$. 
In general, the simulation study establishes a well-performing test of the null hypothesis of the imaginary component of $\sigma_{12}$ is $0$.

\begin{figure}
    \centering
    \includegraphics[width = .47\textwidth]{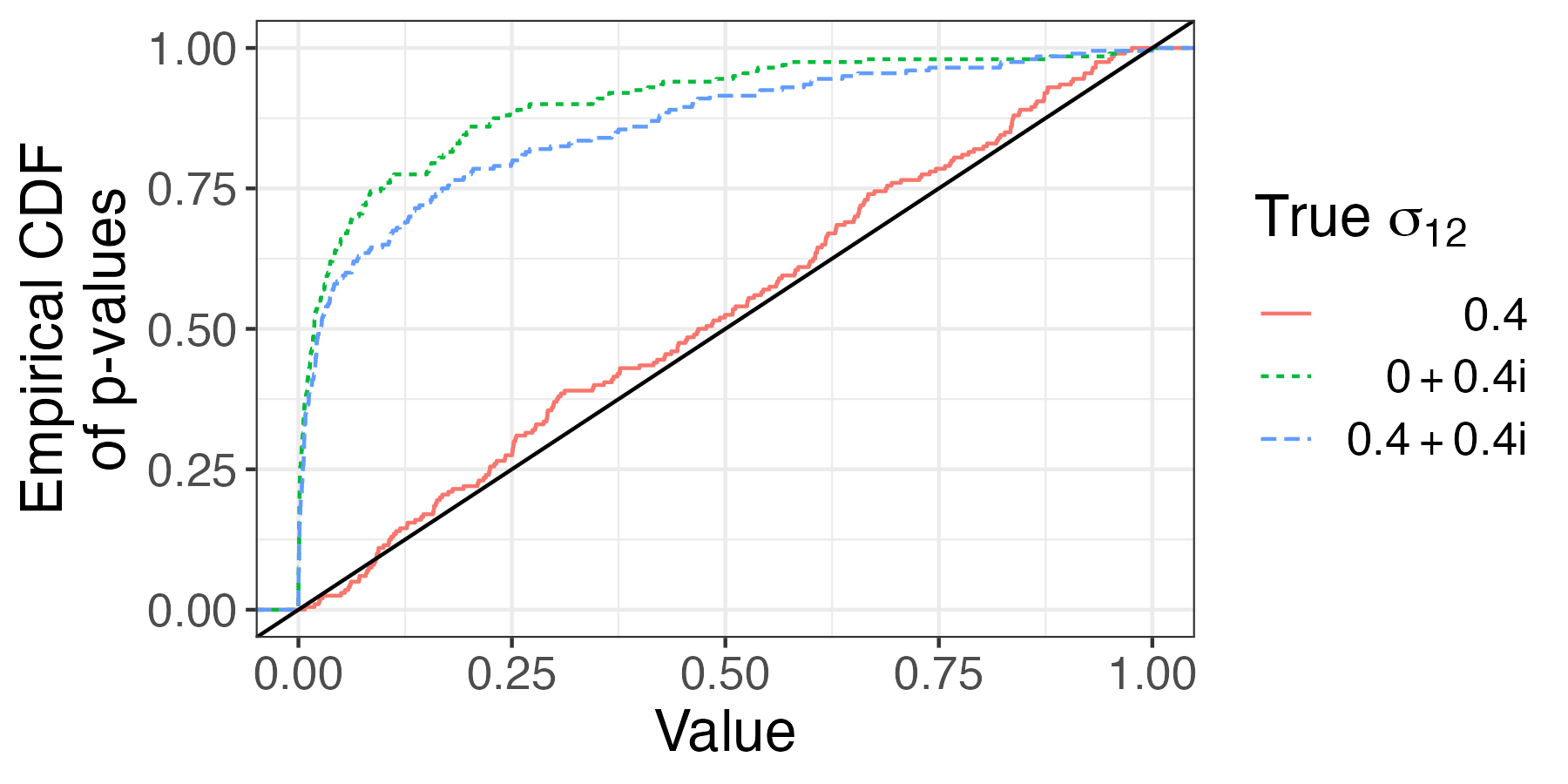}
    \caption{Empirical cumulative distribution of p-values under the three settings for $\sigma_{12}$.}
    \label{fig:lrt}
\end{figure}

We turn next to estimation of $\sigma_{12}$.
In Table \ref{tab:par_est}, we summarize the estimates of $\sigma_{12}$ for each of the models and true data generating models. 
The model that assumes $\sigma_{12} \in \mathbb{R}$ performs well when this is the true setting, with little parameter bias and standard deviation at $0.23$.
When the real model is misspecified, however, the estimated parameter $\Re(\sigma_{12})$ has higher variance over simulations, with sample standard deviation around $0.30$ compared to $0.23$. 
In contrast, the model that assumes $\sigma_{12} \in \mathbb{C}$ results in estimates with sample standard deviations approximately the same no matter the setting. 
Furthermore, the complex model estimates the components $\sigma_{12}$ very well on average across all three settings, with minimal bias in comparison to the true values.
Interestingly, the results suggest that $\Im(\sigma_{12})$ is not necessarily harder to estimate compared to $\Re(\sigma_{12})$ and introducing $\Im(\sigma_{12})$ as a parameter does not substantially detract from the estimation of $\Re(\sigma_{12})$.

\begin{table}[ht]\footnotesize
    \centering
        \scalebox{1}{
        \setlength\tabcolsep{3pt}
    \begin{tabular}{|c|c c|c c|c c|}\hline
          & $\Re(\sigma_{12})$ &$\Im(\sigma_{12})$  & $\Re(\sigma_{12})$ &$\Im(\sigma_{12})$ &  $\Re(\sigma_{12})$ &$\Im(\sigma_{12})$  \\ 
          Truth & $0.40$ & $0.00$ & $0.00$ & $0.40$ & $0.40$ & $0.40$\\ \hline
         Real & \begin{tabular}{c}$0.40$ \\ 
          (0.23)\end{tabular} & - & \begin{tabular}{c}$0.01$ \\ 
          (0.32)\end{tabular} & -  & \begin{tabular}{c}$0.52$ \\ 
          (0.30)\end{tabular} & - \\
          Complex & \begin{tabular}{c}$0.38$ \\ 
          (0.22)\end{tabular}& \begin{tabular}{c}$-0.02$ \\ 
          (0.19)\end{tabular} & \begin{tabular}{c}$0.00$ \\ 
          (0.20)\end{tabular} &   \begin{tabular}{c}$0.38$ \\ 
          (0.20)\end{tabular} &  \begin{tabular}{c}$0.37$ \\ 
          (0.21)\end{tabular} &  \begin{tabular}{c}$0.36$ \\ 
          (0.22)\end{tabular}  \\ \hline
    \end{tabular}
    }
    \caption{Mean (without parentheses) and sample standard deviation (in parentheses below the mean value) over 200 simulations of real and complex parts of the estimate of $\sigma_{12}$. The first row ``Truth'' indicates the true data-generating model for $\sigma_{12}$.
}
    \label{tab:par_est}
\end{table}


Finally, in Table \ref{tab:pred}, we compare the approaches in terms of prediction at unobserved locations from a random uniform sample of locations, summarized by root-mean-squared error averaged over simulations. 
To emphasize the role of the cross-covariance, we predict our testing data for variable 1 using {\it only} observed data from variable 2, and vice versa. 
When $\sigma_{12} = 0.4$, the real and complex models perform similarly, around $0.91$ for both variable 1 and variable 2. 
However, when there is an imaginary component to the model, prediction is considerably detracted when using the real model, with substantial gaps of $0.065$ to $0.095$ between the RMSEs of the models. 
Overall, we find that accounting for the imaginary component of the model, when present, is necessary for proper estimation and prediction in this setting. 

\begin{table}[ht]\footnotesize
    \centering
    \begin{tabular}{|c|c|c|c|c|c|}\hline
          & Response & \begin{tabular}{c} $\sigma_{12}=$ \\ $0.4$  \end{tabular}& \begin{tabular}{c} $\sigma_{12}=$ \\ $ 0.4 \mathbbm{i}$  \end{tabular} & \begin{tabular}{c} $\sigma_{12}=$ \\ $0.4 + 0.4 \mathbbm{i}$  \end{tabular} \\ \hline
         Real & 1 & 0.924 & 1.013 &0.911 \\ 
         Complex & 1  & 0.915 & 0.935 & 0.846 \\ 
         Real & 2 &0.912& 0.993& 0.917 \\ 
         Complex & 2 & 0.902  & 0.907& 0.827 \\ \hline
    \end{tabular}
    \caption{Average root-mean-squared-error of predictions for real and complex versions of the estimated model, broken down by the target response variable. 
    }
    \label{tab:pred}
\end{table}

\subsection{The $d=2$ case} \label{sec:sim_spatial}

We next consider the case where $d=2$. 
The marginal parameters of the two processes are taken to be the same in the first study.
In the context of Section \ref{sec:d2vary}, we take $\psi(\mb{\theta}) = \textrm{sign}(\mb{\theta}^* \mb{\theta})$.
For data generation, we take two situations of $\mb{\theta}^* = (1,0)^\top$ versus $\mb{\theta}^* = (1,1)^\top$ to parameterize $\psi(\theta)$ and $\mu_{12}(d\mb{\theta}) = \Re(\sigma_{12}) + \textrm{sign}(\mb{\theta}^\top \mb{\theta}^*)\Im(\sigma_{12}) \I$ as formulated in Section \ref{sec:d2vary}. 
As in the $d=1$ case, we consider $\sigma_{12} = 0.4$, $\sigma_{12} = 0.4\I$, and $\sigma_{12} = 0.4 + 0.4\I$.
For estimation, we consider the case where $\mb{\theta} = (1,0)^\top$, as well as when it is estimated from the data. 

We present likelihood ratio test results in Table \ref{tab:lrt_d2}.
While Type I error rates are slightly inflated when $\mb{\theta}^*$ is not estimated, performance overall is good when $\mb{\theta}^*$ is estimated, at 0.06 near the 0.05 nominal level.
When the true $\mb{\theta}^*$ is specified, the tests have more power compared to estimation of $\mb{\theta}^*$; alternatively, when $\mb{\theta}^*$ is misspecified, the tests have less power compared to having $\mb{\theta}^*$ estimated from the data. 
These results thus line up with expectations.

\begin{table}[ht]
    \centering
    \scalebox{.9}{
    \begin{tabular}{|c c|c c c|}\hline
         True $\mb{\theta}^*$ & \begin{tabular}{c}
              Implemented   \\
             $\mb{\theta}^*$
         \end{tabular}& \begin{tabular}{c} $\sigma_{12}=$ \\ $0.4$  \end{tabular}& \begin{tabular}{c} $\sigma_{12}=$ \\ $ 0.4 \mathbbm{i}$  \end{tabular} & \begin{tabular}{c} $\sigma_{12}=$ \\ $0.4 + 0.4 \mathbbm{i}$  \end{tabular} \\ \hline
         $(1,0)^\top$ & $(1,0)^\top$ & 0.12 & 0.97& 0.89\\ 
         $(1,1)^\top$ & $(1,0)^\top$ & 0.12 & 0.48& 0.48\\ 
         $(1,0)^\top$ & Estimated & 0.06 & 0.78 & 0.72\\ 
         $(1,1)^\top$ & Estimated & 0.06 & 0.98 & 0.93\\ \hline
    \end{tabular}

    }
    \caption{Proportion of simulations with significant likelihood comparison test for imaginary component with nominal level $0.05$. The first column refers to the Type I error rate, while the other two columns are the power of the test under the specified alternative. }
    \label{tab:lrt_d2}
\end{table}

In Table \ref{tab:par_est_d2}, we provide the mean and standard deviation of estimates for $\sigma_{12}$ over simulations when $\mb{\theta}^*$ is correctly specified. 
When $\mb{\theta}^*$ is estimated, the imaginary part of $\mu_{12}(d\mb{\theta})$ may only be identified up to a sign. 
As in the case $d=1$, maximum likelihood estimation estimation performs generally well. 
Compared to the $d=1$ case, parameters might have slightly more bias (at times, $0.30$ on average compared to the target $\Im(\sigma_{12}) = 0.4$), but the estimates are also less variable compared to $d=1$. 

\begin{table}[ht]\footnotesize
    \centering
    \setlength\tabcolsep{3pt}
    \begin{tabular}{|c|c c|c c|c c|}\hline
         & $\Re(\sigma_{12})$ & $\Im(\sigma_{12})$ &$\Re(\sigma_{12})$ & $\Im(\sigma_{12})$ &$\Re(\sigma_{12})$ & $\Im(\sigma_{12})$  \\ \hline
         True value & $0.40$ &$0.00$  & $0.00$ & $0.40$ & $0.40$ &$0.40$ \\ \hline
          Real   & \begin{tabular}{c}$0.42$ \\ 
          (0.10)\end{tabular} & - &  \begin{tabular}{c}$-0.06$ \\ 
          (0.12)\end{tabular} & - &  \begin{tabular}{c}$0.46$ \\ 
          (0.15)\end{tabular} & - \\ 
          Complex &  \begin{tabular}{c}$0.41$ \\ 
          (0.11)\end{tabular} &  \begin{tabular}{c}$-0.01$ \\ 
          (0.13)\end{tabular}  & \begin{tabular}{c}$-0.04$ \\ 
          (0.09)\end{tabular} &\begin{tabular}{c}$0.32$ \\ 
          (0.17)\end{tabular} & \begin{tabular}{c}$0.42$ \\ 
          (0.13)\end{tabular} &\begin{tabular}{c}$0.30$ \\ 
          (0.19)\end{tabular}\\ \hline
    \end{tabular}
    \caption{Mean (without parentheses) and sample standard deviation (in parentheses below the mean) over simulations of real and complex parts of the estimate of $\sigma_{12}$ for $d=2$ when $\mb{\theta}^*$ is correctly specified. }
    \label{tab:par_est_d2}
\end{table}

\subsection{Computation}

The code accompanying the paper is available in \url{https://github.com/dyarger/multivariate_matern}.  For these simulations and data analysis, we use fast 
 Fourier transforms \citep{frigo_fast_1999} implemented in \texttt{R}, as well as (for $d=1$) $\texttt{approx}$, the basic interpolation function in \texttt{R}, or (for $d=2$) 
 tools from the \texttt{fields} package \citep{fields}; again, see \url{https://github.com/dyarger/multivariate_matern}. 
This appears more stable than other approaches like using the special functions or na\"ively integrating. 
Another approach for computation is Fourier feature approaches \citep[for example,][]{miller_bayesian_2022, emery_improved_2016}, but its implementation and use for estimation 
in this setting would require substantial more development and validation. We also briefly compare approaches for computation in the Supplement.

\section{Data analysis}\label{sec:data_analysis}

We next demonstrate the estimation of these new multivariate Mat\'ern models in the context of time series and spatial data. 
We first examine time series of housing data that demonstrates covariance asymmetry features well-represented by the full model in Section \ref{Introduce_model}, similar to an analysis alongside \cite{mu2024gaussian}.
We then analyze ocean temperature data collected by Argo floats in the setting where $d=2$, which is studied in \cite{bolin_multivariate_nodate}.
In the Supplement, we also present an analysis on the air pressure and temperature data in the Pacific Northwest of North America \citep[see][]{gneiting_matern_2010, apanasovich_valid_2012, bolin_multivariate_nodate, cressie_multivariate_2016, hu_multivariate_2013}. 

For model estimation, we employ maximum likelihood estimation under a Gaussian assumption. 
Let $\mb{X}_{obs,1} = [X_1(\ti_i)]_{i=1}^{n_1} \in \mathbb{R}^{n_1}$ and $\mb{X}_{obs,2} = [X_2(\ti_i)]_{i=1}^{n_2} \in \mathbb{R}^{n_2}$ be vectors of the two variables, respectively.
We suppose that \begin{align*}
\begin{pmatrix} \mb{X}_{obs, 1} \\ 
\mb{X}_{obs, 2}
\end{pmatrix}&\sim  \mathcal{N}_{n_1+n_2}\Bigg(\begin{pmatrix} \mb{\mu}_1 \\  \mb{\mu}_2 \end{pmatrix}, \\
&~~~~~~~~~\mb{\Gamma} := \begin{pmatrix}
\mb{\Sigma}_{11} + \gamma^2_1\Imat_{n_1} & \mb{\Sigma}_{12}
 \\ 
\mb{\Sigma}_{12}^\top & \mb{\Sigma}_{22} + \gamma^2_2\Imat_{n_2}
\end{pmatrix}\Bigg),
\end{align*}where $\mb{\mu}_1$ and $\mb{\mu}_2$ are mean vectors and $\mb{\Sigma}_{jk}$ are covariance matrices. 
Adding the nugget effect terms of $\gamma^2_1$ and $\gamma^2_2$ is commonly done in spatial statistics applications including previous ones with this data \citep[for example,][]{gneiting_matern_2010, apanasovich_valid_2012}. 
Since closed-form representation of the maximum likelihood estimates are not generally available, we again numerically optimize the likelihood using the L-BFGS-B routine in R \citep{byrd_limited_1995}.

\subsection{Housing market data}

The housing market data is made available from Redfin, a national real estate brokerage \citep{redfin}. 
We consider $p=2$ variables ``median sale price''  (variable $j=1$) and ``inventory'' (variable $j=2$), a measure of how many homes are on the market, in the San Diego, California metropolitan area, which are summarized weekly beginning in January 2017, with a total of $n_1 = n_2 =383$ weeks. 
After removing the means $\mb{\mu}_1$ and $\mb{\mu}_2$ from each series with a standard LOESS estimator and standardizing the two variables, we estimate parameters of the covariance model. 

The data is presented in the upper panel of Figure \ref{fig:housing}. 
On the lower panel, we plot the cross-correlation function, using three different models, the time-series nonparametric estimator and the real (assuming $\Im(\sigma_{12}) = 0$) and complex  ($\Im(\sigma_{12})$ estimated) versions of the model in Section \ref{Introduce_model}. 
The sample marginal correlation between the processes is $0.049$.
The nonparametric estimator demonstrates asymmetry, with positive correlation for negative lags and negative correlations for positive lags, suggesting the ``Complex'' model is appropriate. 
The estimated maximized log-likelihoods were $-1001.722$ for the ``Complex'' model and $-1004.684$ for the ``Real'' model, resulting in a p-value of $0.0149$ for the nested likelihood ratio test, demonstrating that the ``Complex'' model provides a statistically significant improvement over the ``Real'' model. 
The parameter estimates are presented in Table \ref{tab:housing}, with similar parameter estimates between the ``Real'' and ``Complex'' models for their shared parameters. 

The asymmetric cross-correlation estimates also make sense when interpreted conceptually.
For negative lags, we compare the median sale price with the inventory in the future, which exhibits more positive correlation. 
In other words, if median sale prices are high, inventory will later increase. 
This could be attributed to more homeowners willing to sell their houses when prices have been higher. 
Alternately, for positive lags, we compare the inventory with the sale price in the future, which exhibits negative correlation.
That is, if inventory is high, the median sale price will decrease in the near future. 
This could be attributed to increases in inventory providing more competition among sellers and more options for buyers, which in turn will drive down prices. 
Such explanations of the estimated cross-correlation require time for reaction, thus resulting in greater correlation at non-zero lags as compared to the marginal 
correlations at zero lag. 

\begin{figure}[ht]
    \centering
    \includegraphics[width = .43\textwidth]{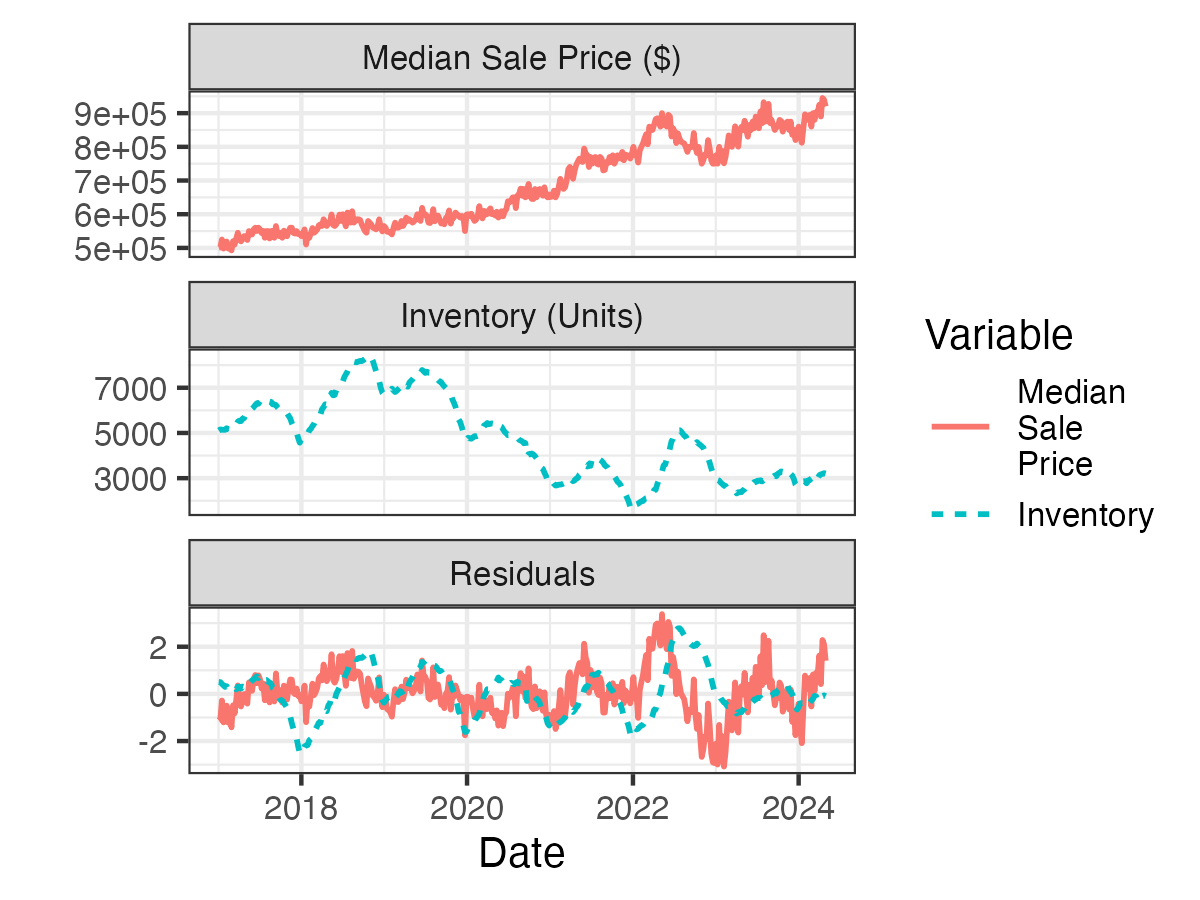}
    
    \includegraphics[width = .41\textwidth]{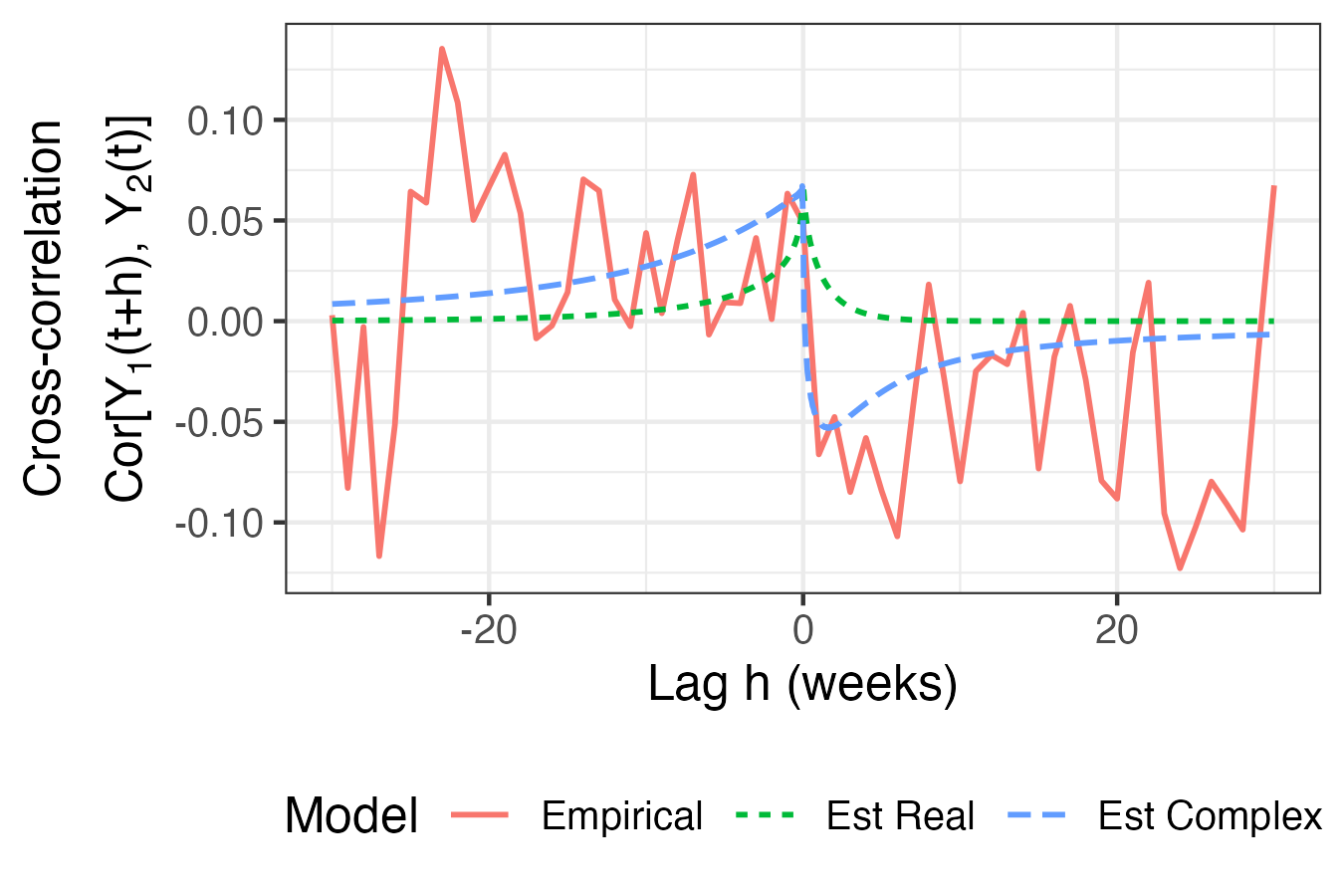}

    \caption{(Top) Original and normalized data. (Bottom) Cross-correlation functions. 
    ``Empirical'' refers to nonparametric estimators from the \texttt{stats::ccf} functions in \texttt{R}. ``Est Real'' and ``Est Complex'' refer to models in Section \ref{Introduce_model}, assuming $\Im(\sigma_{12}) = 0$ and both $\Re(\sigma_{12})$ and $\Im(\sigma_{12})$ free, respectively. }
    \label{fig:housing}
\end{figure}
\begin{table}[ht]
    \centering
    \begin{tabular}{|c|cc|c|cc|}
        & Real & Complex & & Real & Complex \\ \hline
      $\nu_1$  & 0.24  & 0.23 & $\frac{\Re(\sigma_{12})}{\sqrt{\sigma_{11}\sigma_{22}}}$ & 0.09 & 0.08\\
      $\nu_2$ & 0.10 & 0.10 & $\frac{\Im(\sigma_{12})}{\sqrt{\sigma_{11}\sigma_{22}}}$ & - & 0.23\\
      $a_1$ & 0.55 & 0.52 & $\gamma_1^2\big/\sigma_{11}$ & 0.0003 & 0.0003\\
      $a_2$ & 0.12 & 0.14 & $\gamma_2^2\big/\sigma_{22}$ & 0.0515 &  0.0532\\
      $\sigma_{11}$ & $9.2\cdot 10^{8}$& $9.4\cdot 10^{8}$\\
      $\sigma_{22}$ & $7.9\cdot 10^{5}$& $7.6\cdot 10^{5}$ \\
    \end{tabular}
    \caption{\add{Parameter estimates for the housing data (variable 1: median sale price; variable 2: inventory). A dash indicates that the given parameter does not exist under that model.}}
    \label{tab:housing}
\end{table}

\subsection{Argo temperature data}

\add{For the Argo dataset presented here and Pacific Northwest data presented in the Supplement}, we use $\mb{s} \in \mathbb{R}^2$ and have $p=2$.
To evaluate covariances and cross-covariances, we use a 2-dimensional inverse Fourier transform implemented through \cite{frigo_fast_1999} on a regular and fine grid.
Covariance values are then interpolated onto the actual distances between sites. 
For $\mb{\mu}_1$ and $\mb{\mu}_2$, we subtract the empirical mean assumed constant over space. 

\add{As the models are more complicated for $d=2$,} we define the models compared in more detail:
\begin{enumerate}

\item[(IM)] Independent Mat\'ern, where the processes are estimated independently, so that $\mb{\Sigma}_{12} = \mb{0}$, and $\mb{\Sigma}_{11}$ and $\mb{\Sigma}_{22}$ are defined by univariate Mat\'ern covariances. 

\item[(SCF)] A ``single covariance function'' model, with\begin{align*}
\mb{\Gamma} &= \begin{pmatrix} \sigma_{11} & \sigma_{12} \\ \sigma_{12} & \sigma_{22} \end{pmatrix} \otimes \mb{\Sigma}_{single} + \begin{pmatrix}
 \gamma^2_1\Imat_n & 
 \mb{0}_{n \times n}\\ 
 \mb{0}_{n \times n}&  \gamma^2_2\Imat_n
\end{pmatrix}.
\end{align*}Here, $\mb{\Sigma}_{single}$ is a correlation matrix defined with a Mat\'ern covariance at the observed points, and $\otimes$ is the standard Kronecker product, so that each covariance and cross-covariance has the same shape. 

\item[(MMG)] The multivariate Mat\'ern of \cite{gneiting_matern_2010} using the parameters $\nu_1$, $\nu_2$, $\nu_{12}$, $a_1$, $a_2$, $a_{12}$, $\sigma_{11}$, $\sigma_{12}$, $\sigma_{22}$, $\gamma^2_1$, and $\gamma^2_2$.

\item[(SMM-0)] A spectral multivariate Mat\'ern as presented in Section \ref{sec:spatial}, with real directional measure and $\ang(\mb{\theta}) = \textrm{sign}(\theta_1)$ fixed and not estimated. 
The parameters are $\nu_1$, $\nu_2$, $a_1$, $a_2$, $\sigma_{11}$, $\sigma_{12}$, $\sigma_{22}$, $\gamma^2_1$, and $\gamma^2_2$. 

\item[(SMM-R)] A spectral multivariate Mat\'ern as presented in Section \ref{sec:spatial}, with real directional measure and $\ang(\mb{\theta}) = \textrm{sign}\left(\mb{\theta}^\top \mb{\theta}^*\right)$ for a fixed $\mb{\theta}^*$ that we estimate. 
\item[(SMM-C)] A spectral multivariate Mat\'ern as presented in Section \ref{sec:spatial}, with complex directional measure $\sde_{12}(d\mb{\theta}) = \Re(\sigma_{12}) + \Im(\sigma_{12})\I\textrm{sign}\left(\mb{\theta}^\top \mb{\theta}^*_1\right)d\mb{\theta}$ and the function  $\ang(\mb{\theta}) = \textrm{sign}\left(\mb{\theta}^\top \mb{\theta}^*_2\right)$ for two different estimated parameters $\mb{\theta}^*_1$ and $\mb{\theta}^*_2$. 

\end{enumerate}

The Argo data which profiles temperature and salinity measurements in the upper $2{,}000$ meters of the ocean \citep{argo2020}. 
Here, we focus on temperature measurements at the depths of $300$ meters and $1{,}500$ meters located south of New Zealand in \add{February} 2015 and fit a bivariate model to them. 
\add{The locations of the 97 measurements for $1{,}500$ meters are a subset of the locations of the 110 measurements for $300$ meters.
We plot the data in Figure \ref{fig:argo_data}. }
This data was also used in \cite{bolin_multivariate_nodate}, and an extensive analysis of Argo data using univariate spatial models is presented in \cite{kuusela_locally_2018}.
We present parameter estimates and log-likelihoods in the Supplement. 
The estimated cross-covariances are plotted in Figure \ref{fig:argo_data_fun_plot}.
The multivariate Mat\'ern with complex directional measure shows a substantially different shape of the cross-covariance compared to the other models, and there is a substantially strong imaginary component to the cross-covariance. 
While the estimated real correlation parameter for the SMM-R model is $0.678$, the estimated real and imaginary correlation parameter for the SMM-C model are $0.638$ and $0.275$, respectively. 
In this case, the nested likelihood ratio test suggests that the imaginary correlation parameter does not significantly improve model fit (with reported values in the Supplement).

\begin{figure}[ht]
    \centering
    \includegraphics[width=.235\textwidth]{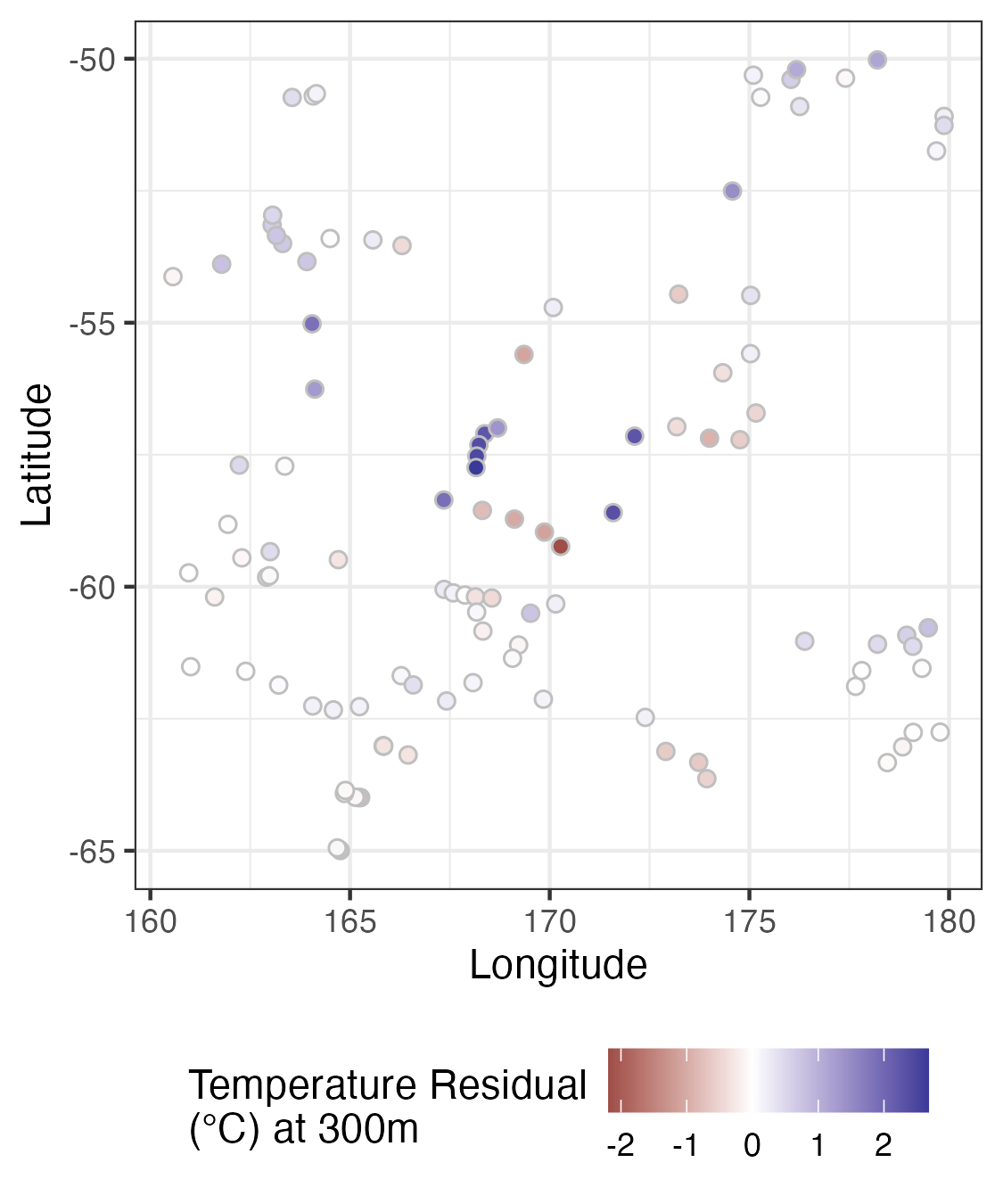}
    \includegraphics[width=.235\textwidth]{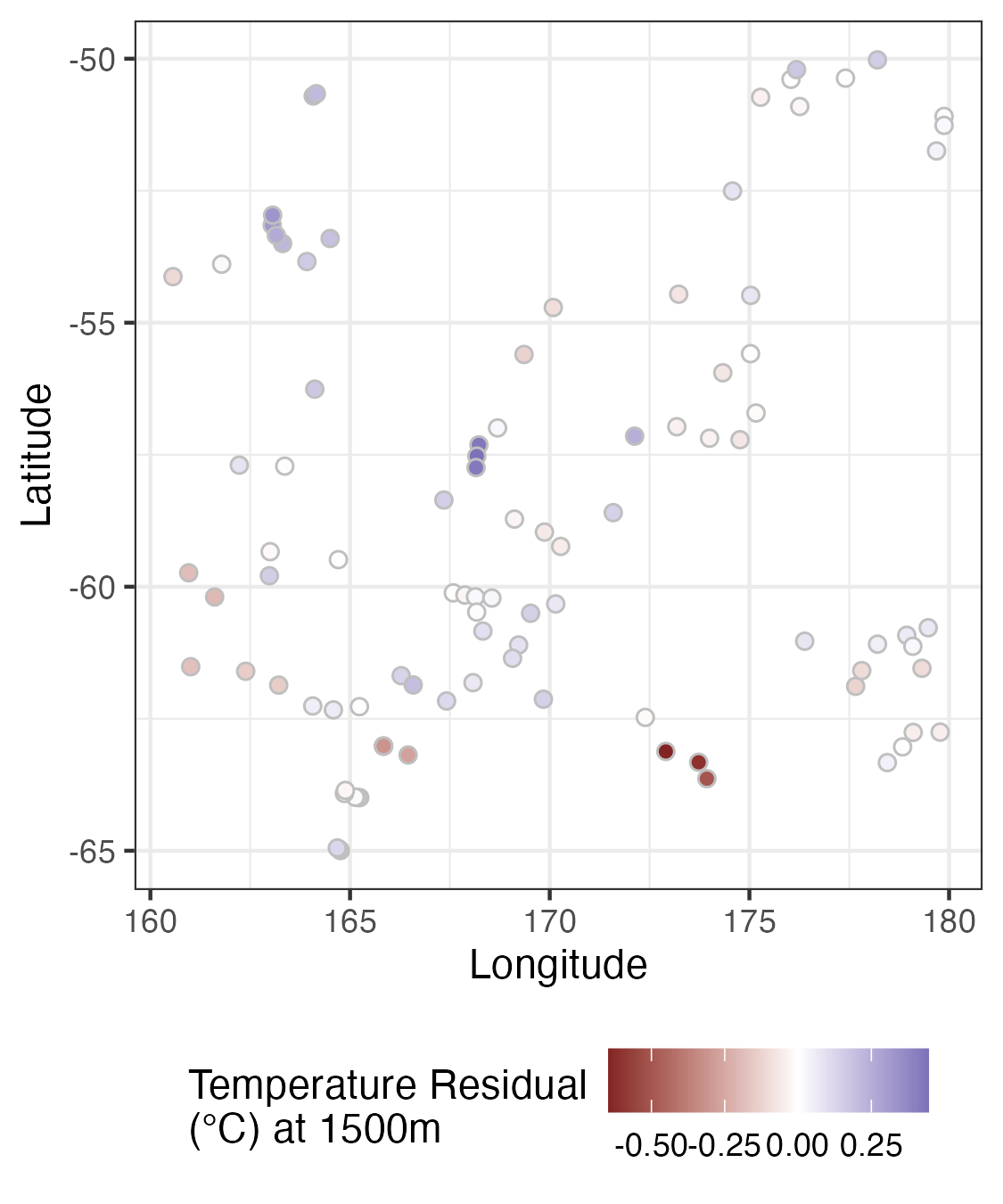}
    \caption{\add{Argo data for $300$ meters (Left) and $1{,}500$ meters (Right).}}
    \label{fig:argo_data}
\end{figure}

\begin{figure}[ht]
    \centering
    \includegraphics[scale=.41]{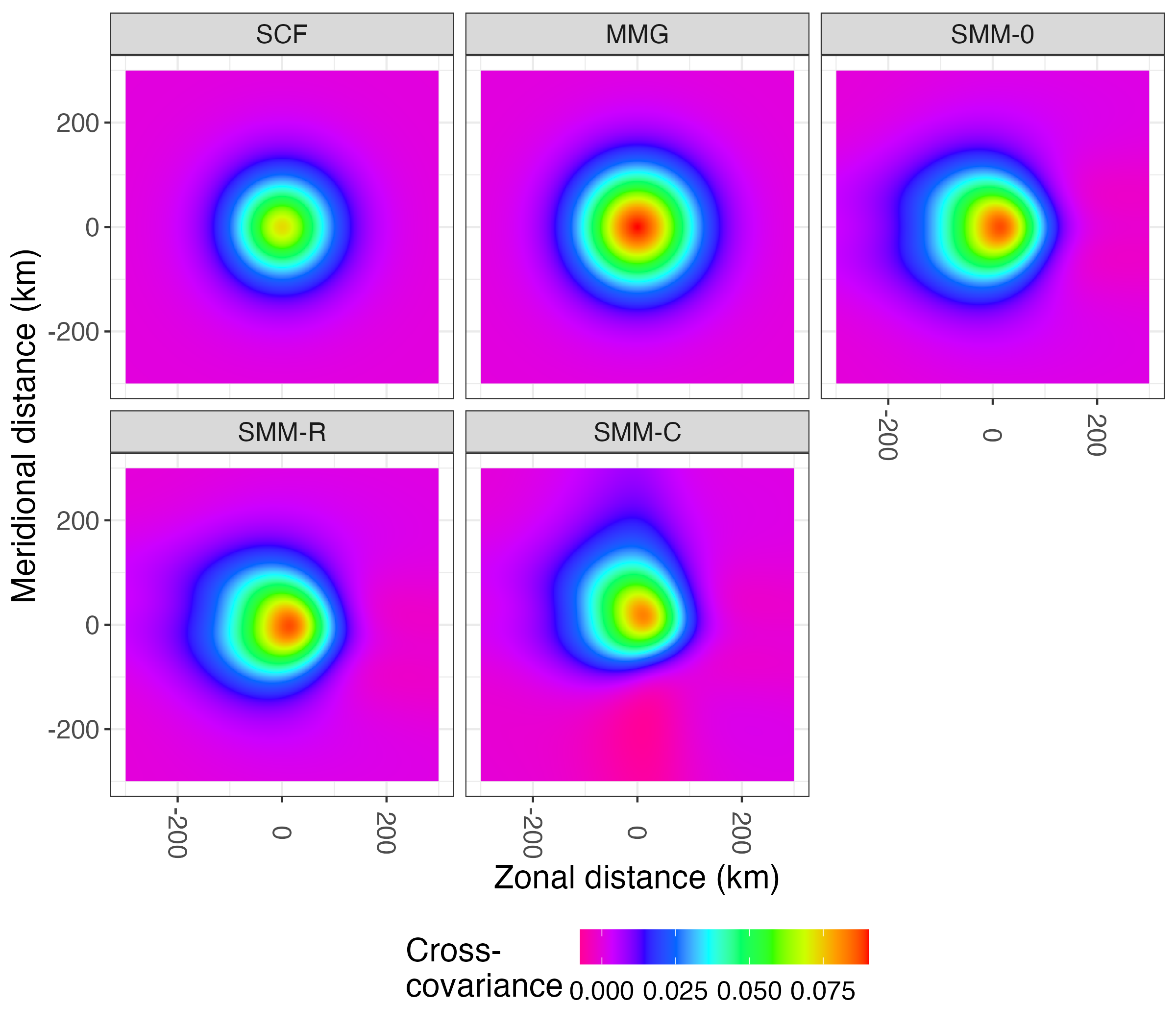}
    \caption{Estimated cross-covariance functions, Argo temperature data.}
    \label{fig:argo_data_fun_plot}
\end{figure}

\section{Discussion and extensions}\label{sec:mm_discussion}

In this work, we introduce a new class of multivariate Mat\'ern models motivated by the spectral representation of the Mat\'ern covariance. This class of models provides more flexible forms of the cross-covariance structure compared to the multivariate Mat\'ern of \cite{gneiting_matern_2010} and its more recent extensions. 
In particular, asymmetry in the cross-covariance can be modeled straightforwardly. 
Furthermore, compared to that of the multivariate Mat\'ern of \cite{gneiting_matern_2010}, there are fewer parameters for $d=1$ and validity of the cross-covariance is given without complicated restrictions on the parameters. 
We also provide clarity on how previous work on multivariate Mat\'ern models fit into the framework developed here. 
Under particular parameters for $d=1$, closed-form expressions for the cross-covariance are given.

In the spatial case of $d\ge 2$, flexible representations of the cross-covariance with respect to the domain are available. 
A potentially fruitful area of further research would be to investigate any closed-form spatial covariances. 
On the other hand, the flexibility of such models suggests that closed-form cross-covariances may be elusive, and we have explored possible approaches  \citep[for example, in][]{agrest_general_1971, babister_transcendental_1967} to no avail.
To mitigate this challenge, we demonstrate that spatial processes with these cross-covariances can be easily simulated using existing approaches for multivariate random fields, and the cross-covariance functions can be evaluated through their spectral density.

There are some research avenues that could strengthen our understanding of the models introduced here. 
First, it is unclear which parameters of the cross-covariance are identifiable under fixed-domain asymptotics. 
In the univariate setting, the parameters $a_j$ and $\sigma_{jj}$ are not individually identifiable \citep{zhang_2004}. 
Thus, it is unclear if the real and imaginary parts of $\sd(d\mb{\theta})$ would be identifiable. 
Also, the parameter $\ang(d\mb{\theta})$ is relatively mysterious, especially in the case $d\geq 2$; it may be possible that one can choose a separate $\ang(d\mb{\theta})$ for each cross-covariance when $p > 2$. 
Future work could also introduce computational methodology to use the model presented here for large-scale data analysis; such work for previous multivariate Mat\'ern models includes \cite{fahmy_vecchia_2022}. 

We conclude with potential extensions of this work that may be of interest to the spatial statistics community. 

\subsection{SPDE approach}\label{sec:spde}

The new multivariate Mat\'ern models should relate to the characterizations of the Mat\'ern model and fractional Brownian motion through stochastic partial differential equations (SPDEs) \citep{lindgren_explicit_2011, tafti_fractional_2010}. 
\cite{hu_multivariate_2013} adapt the SPDE approach to the multivariate Mat\'ern of \cite{gneiting_matern_2010} by considering the system \begin{align}
    \begin{pmatrix}
    \mathcal{L}(a_1, \nu_1) & \mathcal{L}(a_{12}, \nu_{12})\\
    0 & \mathcal{L}(a_2, \nu_2)\end{pmatrix}\begin{pmatrix} Y_1(s) \\ Y_2(s)\end{pmatrix}= \begin{pmatrix} \mathcal{W}_1 \\ \mathcal{W}_2
    \end{pmatrix}, \label{eq:SPDE}
\end{align}where $\mathcal{W}_j$ are independent (and, in most cases, Gaussian) white noise processes, $\mathcal{L}(a, \nu) = (a^2\mb{I} - \Delta)^{(\nu+ \frac{1}{2})/2}$, and $\Delta$ is the Laplacian operator. 
\cite{bolin_multivariate_nodate} consider an alternative approach that also uses the SPDE formulation as an inspiration. 
Our approach suggests a new strategy; for the example of $d=1$, one can consider differential operators generated by convolutions of $G^{a, -\mu}_+$ and $G^{a, -\mu}_-$, which are in turn characterized using fractional calculus through the Fourier transform $\mathcal{F}$ with $$\mathcal{F} \left[G_{\pm}^{a, \mu}f(y)\right] = (a \mp \I x)^{-\mu}\mathcal{F} \left[f(x)\right]$$ in a similar manner to Section 18.4 of \cite{samko_fractional_1993}.
Through their Fourier transforms, one sees that $G_-^{a, -\mu} \circ G_+^{a, -\mu} = (a^2 \Imat - \Delta)^{\mu}$ where $\circ$ is the convolution. 
Therefore, each operator associated with the Mat\'ern processes (for example, in Equation \ref{eq:SPDE}) can be formed from convolutions of $G_\pm^{a, \mu}$ operators. 
The operators $G_{\pm}^{a, -\mu}$ relate to the differential operators $(a\Imat \mp \nabla)^{\mu}$ where $\nabla$ is the derivative operator, as outlined in Section 18.4 of \cite{samko_fractional_1993}.
See Appendix B of \cite{lindgren_explicit_2011} and Sections 27 and 18.4 of \cite{samko_fractional_1993} for further details. 
However, more care must be taken to extend such an approach to $d=2, 3, \dots$ and a complex-valued directional measure.

\subsection{Functional Mat\'ern models}

Extending this model from multivariate to functional data would be of theoretical and practical interest for a variety of spatial functional data analysis applications \citep[see, for example,][]{martinez-hernandez_recent_2020}. 
\cite{shen_tangent_2022} develop a framework for covariance models in a general separable Hilbert space $\mathbb{V}$, which is a setting broader than the $\mathbb{R}^p$-valued processes here. 
In particular, they extend the results of \cite{cramer1942harmonic}, so that a process $X(\mb{s})$ taking values in $ \mathbb{V}$ can be written \begin{align*}
X(\mb{s})&= \int_{\mathbb{R}^d} e^{\I \langle \mb{s}, \mb{x}\rangle} \eta(d\mb{x}),
\end{align*}where $\eta(d\mb{x})$ is an appropriately-defined $\mathbb{V}$-valued random measure.
Similarly to Proposition 5.6 of \cite{shen_tangent_2022}, one might consider\begin{align*}
X(\mb{s})&= \int_{0}^\infty \int_{\us^{d-1}} e^{\I \langle \mb{s}, r\mb{\theta}\rangle} (a + \ang(\mb{\theta})\I r)^{-\nu - \frac{d}{2}}\mb{B}_{\sd}(dr, d\mb{\theta}),
\end{align*}where $\mathbb{E}\left[\mb{B}_{\sd}(dr, d\mb{\theta})\otimes \mb{B}_{\sd}(dr, d\mb{\theta})\right] = r^{d-1} dr \sd(d\mb{\theta})$, $\otimes$ is the outer product on $\mathbb{V}$, and $\sd(d\mb{\theta})$ is now a positive-definite and trace-class operator on $\mathbb{V}$.
Extensions to where $\nu$ and $a$ are also operators would introduce flexible covariances for the spatial functional-data setting.

\subsection{Alternate factorization of spectral density}\label{sec:alt_formal}

As suggested by \eqref{eq:eta_bolin}, one may also study the covariance given by \begin{align*}
C_{jk}(\hv)&=c_jc_k \int_0^\infty \int_{\us^{d-1}}e^{\I \langle \hv, r\mb{\theta}\rangle}\left(a_j^2 + r^2\right)^{-\frac{\nu_j}{2} - \frac{d}{4}}\\
&~~~~\times \sde_{jk}(d\mb{\theta})\left(a_k^2 + r^2\right)^{-\frac{\nu_k}{2} - \frac{d}{4}} dr.
\end{align*}This model was studied for real-valued and constant $\sd(d\mb{\theta})$ in \cite{bolin_multivariate_nodate}, and its general form in \cite{terres2018bayesian} and \cite{guinness2014multivariate}.
The form suggests that, when $\sde_{jk}(d\mb{\theta}) \propto d\mb{\theta}$, the cross-covariance is isotropic for any appropriate $a_j$, $a_k$, $\nu_j$, and $\nu_k$; alternately, when $\sd(d\mb{\theta}) \propto\I \textrm{sign}(\mb{\theta}^\top \mb{\theta}^*) d\mb{\theta}$ for $\mb{\theta}^*$ fixed, the cross-covariance would be reflected across the direction perpendicular to $\mb{\theta}^*$ for any values of the inverse range and smoothness parameters. 
This cross-covariance also allows one to avoid the introduction of $\ang(\mb{\theta})$ when $d \geq 2$. 
However, closed-form expressions of this covariance seem to be more elusive, except in some cases.
Most notably, the ``parsimonious multivariate Mat\'ern model'' introduced in \cite{gneiting_matern_2010}  is obtained if one takes $a_j = a_k$.
Another case is $d=1$, $\nu_j = \nu_k = 3/2$, and real directional measure: we see that \begin{align}\begin{split}
&C_{jk}(h)= 2c_jc_k\Re(\sigma_{jk}) \int_{0}^\infty \cos(hx) \\
&~~~~~\times\left(a_j^2 + x^2\right)^{-\frac{\nu_j}{2} - \frac{1}{4}}\left(a_k^2 + x^2\right)^{-\frac{\nu_k}{2} - \frac{1}{4}} dx,\label{eq:bolin_type_spectral}\end{split}
 \end{align}and by 3.728 (1) of \cite{GR_table_2015}, \begin{align*}
C_{jk}(h)&= c_jc_k\Re(\sigma_{jk}) \frac{\pi}{a_ja_k} \frac{a_je^{-a_k|h|} - a_k e^{-a_j|h|}}{a_j^2 - a_k^2}
\end{align*}for $a_j\neq a_k$ (the case $a_j = a_k$ is Mat\'ern). 
We plot examples of these cross-covariances in Figure \ref{fig:bolin_type}.

\begin{figure}
    \centering
    \includegraphics[width = .45 \textwidth]{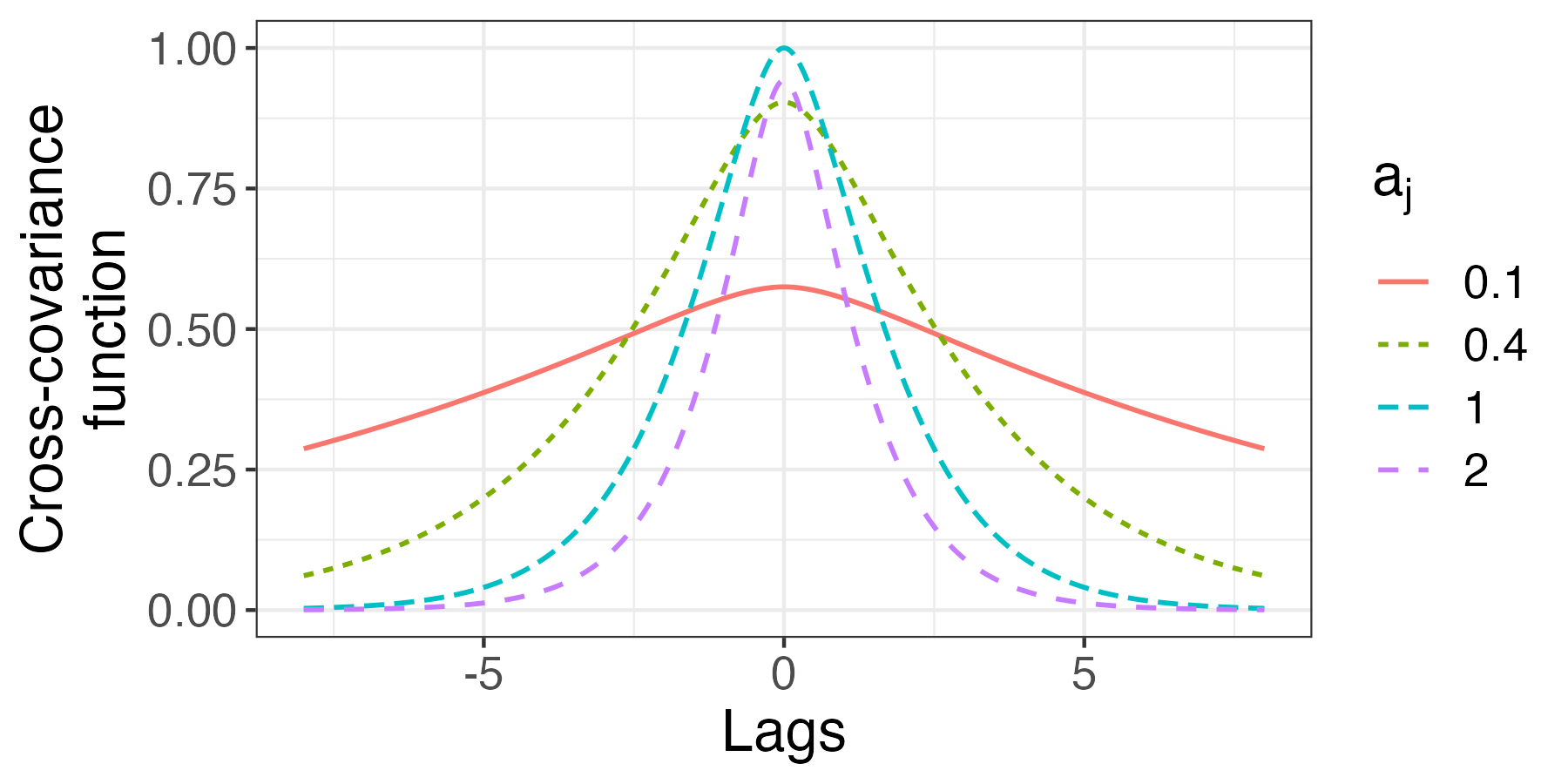}
    \caption{Multivariate Mat\'ern cross-covariances of type \eqref{eq:bolin_type_spectral} for $\Re(\sigma_{jk}) = 1$, $a_k = 1$, $\nu_j = \nu_k = 3/2$ and $a_j$ varying.}
    \label{fig:bolin_type}
\end{figure}

\subsection{Nonstationary models}

When considered across a substantial spatial domain, stationary models may be restrictive, and we include a few key references to extend this to a nonstationary model. 
For example, the multivariate case has been addressed in \cite{kleiber2012nonstationary} and \cite{kleiber_nonstationary_2015}, and \cite{emery_continuous_2018} propose an approach for simulation.
\cite{ton2018spatial} consider a Fourier feature approach for estimating nonstationary covariances using the following theorem from \cite{yaglom_correlation_1987}.
\begin{theorem}[\cite{yaglom_correlation_1987}]
    A nonstationary covariance function $C(\mb{s}_1, \mb{s}_2)$ is positive definite if and only if one can represent \begin{align*}
        C(\mb{s}_1, \mb{s}_2) &= \int_{\mathbb{R}^d \times \mathbb{R}^d} e^{\I \left(\mb{x}_1^\top \mb{s}_1 - \mb{x}_2^\top \mb{s}_2\right)}\mu(d\mb{x}_1, d\mb{x}_2),
    \end{align*}where $\mu(\cdot, \cdot)$ is the Lebesgue-Stieltjes measure associated with some positive semi-definite function $f(\mb{x}_1, \mb{x}_2)$ of bounded measure. 
\end{theorem}
This is one avenue for pursuing such models that is applicable to the spectral approach.
\cite{kleiber2012nonstationary} instead begin with the normal-scale mixture of the Mat\'ern covariance, reviewed in the Supplement. 


\subsection{Non-Gaussian models}\label{sec:nongaussian}

Throughout the paper, we have assumed that $\mb{Y}$ is Gaussian for simulations and estimation of parameters. 
Non-Gaussian random fields are of interest in a number of practical applications reflected in some development of their estimation and simulation \citep{bolin2014spatial, wallin2015geostatistical, bolin_multivariate_nodate, grigoriu1998simulation, popescu1998simulation}.
Two approaches for doing so are the SPDE approach \citep{bolin2014spatial, wallin2015geostatistical, bolin_multivariate_nodate} and a transformation approach \citep{grigoriu1998simulation, popescu1998simulation}.
To directly apply the SPDE approach, one would likely need to develop approaches suggested in Section \ref{sec:spde} before applying to asymmetric cross-covariance functions, while the transformation approach may be more directly applicable. 

An approach closely related to our formulation would be to replace the random Gaussian measure $\mb{B}(dx)$ (as in \eqref{eq:Y-d=1-def}) or $\mb{B}_{\mb{\mu}}(dr, d\mb{\theta})$ (as in \eqref{eq:matern_stochastic}) with respective non-Gaussian random measures. 
Simulation and model estimation for non-Gaussian processes $\mb{Y}$, however, would require substantial further development.




\subsection{Applications to other covariance functions} 

Factoring the spectral density and using a complex-valued variance parameterization, as done in this paper, should be considered more generally as a way to flexibly extend covariance functions to the multivariate case. 
For example, consider the squared-exponential covariance function for $d=1$ which has covariance function $C(h) = \textrm{exp}(- ah^2)$ and spectral density function $f(x) = (\pi a)^{-\frac{1}{2}}\textrm{exp}(-x^2/ (4a))/2$ \citep[see Section 2.7 of][]{stein_interpolation_2013}. 
Considering a cross-spectral density of \begin{align*}f_{jk}(x) &=\frac{1}{2\sqrt{\pi}a_j^{\frac{1}{4}}a_k^{\frac{1}{4}}} \textrm{exp}\left(\frac{-x^2}{8}\left[\frac{1}{a_j}+\frac{1}{a_k}\right]\right)\\&~~~~~~~\times (\Re(\sigma_{jk}) + \I \Im(\sigma_{jk})\textrm{sign}(x))\end{align*} leads to a cross-covariance function of \begin{align*}C_{jk}(h) &=\frac{(a_ja_k)^{\frac{1}{4}}}{\sqrt{a_+}}\bigg[\Re(\sigma_{jk})\textrm{exp}\left(-\frac{a_j a_k}{a_+}h^2\right) \\
&~~~~~~~~~~~~~~~~~~~~~- \Im(\sigma_{jk}) \frac{2}{\sqrt{\pi}} \tilde{F}\left(\sqrt{\frac{a_ja_k}{a_+}} h \right)\bigg],\end{align*}where $a_+ = (a_j + a_k)/2$ and $\tilde{F}(z)$ is Dawson's integral, an odd function of $z$. The cross-covariance reduces to the squared-exponential covariance function when $a_j = a_k$ and $\Im(\sigma_{jk}) = 0$.
One could conceivably extend this formulation to the more general powered-exponential class or other classes of models like the generalized Cauchy covariance \citep[cf.][]{moreva2023bivariate}.


\section{Acknowledgements}

We thank two anonymous reviewers for their helpful comments that improved the quality of this paper, in particular with 
respect to Sections \ref{sec:computation_simulation} and \ref{sec:data_analysis}.

The Argo data was collected and made freely available by the International Argo Program and the national programs that contribute to it.  (\url{https://argo.ucsd.edu},  \url{https://www.ocean-ops.org}). The Argo Program is part of the Global Ocean Observing System. 
\bibliographystyle{apalike}

\bibliography{main}
\appendix

\section{Hilbert transforms}\label{app:hilbert_trans}

\begin{proposition}
    The Hilbert transform of the function $-a|h|e^{-a|h|}$ is $$-a|h|R(h, a, a)+ \frac{2ahe^{a|h|}}{\pi}{\rm Ei}(-a|h|).$$
\end{proposition}
\begin{proof}

In particular, we look at the Hilbert transform, add and subtract $x$ in the numerator, and break out the integral into multiple terms: \begin{align*}
    \mathcal{H}(|h|e^{-|a|h}) &= \frac{1}{\pi}\int_{-\infty}^\infty \frac{|s|e^{-a|s|}}{x-s} ds \\
    &= \frac{1}{\pi}\int_{-\infty}^\infty \frac{(x - (x-|s|)) e^{-a|s|}}{x-s} ds\\
    &=\frac{1}{\pi}x\int_{-\infty}^\infty \frac{ e^{-a|s|}}{x-s} ds - \frac{1}{\pi}\int_{-\infty}^\infty \frac{(x-|s|) e^{-a|s|}}{x-s} ds.
    \end{align*}
    Notice that the first term is $ x \mathcal{H}(e^{-a|s|})$. Breaking down the remaining integral, we obtain
    \begin{align*}
    \mathcal{H}(|h|e^{-|a|h}) &= x \mathcal{H}(e^{-a|s|}) - \frac{1}{\pi}\int_{0}^\infty e^{-as} ds \\&~~~~~~- \frac{1}{\pi}\int_{-\infty}^0 \frac{(x+s) e^{-a|s|}}{x-s} ds \\
    &=x \mathcal{H}(e^{-a|x|}) - \frac{1}{a\pi} - \frac{1}{\pi}\int_{-\infty}^0 \frac{(x+s) e^{as}}{x-s} ds.
\end{align*}
Adding and subtracting $s$, we have \begin{align*}
     \int_{-\infty}^0 \frac{(x+s) e^{as}}{x-s} ds &=  \int_{-\infty}^0 \frac{(x-s+2s) e^{as}}{x-s} ds \\
     &=\int_{-\infty}^0  e^{as} ds + 2\int_{-\infty}^0 \frac{s e^{as}}{x-s} ds\\
     &=-\frac{1}{a} + 2\int_{-\infty}^0 \frac{s e^{as}}{x-s} ds.
\end{align*}
Finally, \begin{align*}
    \int_{-\infty}^0 \frac{s e^{as}}{x-s} ds = xe^{a|x|} \textrm{Ei}(-a|x|)
\end{align*}follows from a change in variables and the definition of $\textrm{Ei}(z)$.

Combining work, we see that \begin{align*}
     \mathcal{H}(a|h|e^{-|a|h})&= a|h| \mathcal{H}(e^{-a|h|}) + \frac{2xe^{a|h|}}{\pi}\textrm{Ei}(-a|h|).
\end{align*}

The Hilbert transform $\mathcal{H}(e^{-a|h|})$ was established as $-R(h, a, a)$ in the proof of Theorem \ref{thm:im_ent}.
\end{proof}

\section{Additional properties of multivariate Mat\'ern random fields}

\subsection{Reversed version}

We next examine the role of the choice of $\ang(\cdot)$ plays, connecting two different multivariate Mat\'ern models with different $\ang(\cdot)$ functions. 
For the $d=1$ case examined in Section \ref{Introduce_model}, this motivates own approach to look only examine models with $\ang(\cdot) = \textrm{sign}(\cdot)$, as each model with $\ang(\cdot) = -\textrm{sign}(\cdot)$ can be associated with a model that has $\ang(\cdot) = \textrm{sign}(\cdot)$.

\begin{proposition}\label{prop:reflected}
Let $\cf(\hv)$ be a multivariate Mat\'ern covariance function with parameters $\sd(d\mb{\theta})$, $\ang(\mb{\theta})$, $\Nu$, and $\mb{a}$, and let $\tilde{\cf}(\hv)$ be a multivariate Mat\'ern covariance function with respective parameters $\overline{\sd(d\mb{\theta})}$, $-\ang(\mb{\theta})$, $\Nu$, and $\mb{a}$. 
Then $\tilde{\cf}(\hv) = \cf(-\hv)$.
\end{proposition}

We briefly outline the proof. 
By the Hermitian property of $\sd(d\mb{\theta})$ and that $- \ang(\mb{\theta}) =\ang(-\mb{\theta})$, we have \begin{align*}
\tilde{C}_{jk}(\hv) &= c_jc_k \int_0^\infty \int_{\us^{d-1}} e^{\I \langle \hv, r\mb{\theta}\rangle} \left(a_j 
+\ang(-\mb{\theta})\I r\right)^{-\nu_j - \frac{d}{2}}\\
&~~~~~~~ \times \sde_{jk}(-d\mb{\theta})\left(a_k - \ang(-\mb{\theta}) \I r\right)^{-\nu_k - \frac{d}{2}}r^{d-1} dr \\
&= c_jc_k \int_0^\infty \int_{\us^{d-1}} e^{\I \langle -\hv, r\check{\mb{\theta}}\rangle} \left(a_j +\ang(\check{\mb{\theta}})\I r\right)^{-\nu_j - \frac{d}{2}}\\
&~~~~~~~ \times\sde_{jk}(d\check{\mb{\theta}})\left(a_k - \ang(\check{\mb{\theta}}) \I r\right)^{-\nu_k - \frac{d}{2}}r^{d-1} dxr\\
&=C_{jk}(-\hv)
\end{align*}by the change of variables $\check{\mb{\theta}} = -\mb{\theta}$.

We also confirm this result in the case of $d=1$ numerically in the code accompanying this paper.

\subsection{Coherence}
One way to evaluate the statistical properties of multivariate covariance functions is coherence, as investigated in \cite{kleiber_coherence_2018}. 
The measure of coherence evaluates the strength of the linear relationship between two spatial processes at a particular frequency. 
Based on \cite{kleiber_coherence_2018} and letting $f_{jk}(r,\mb{\theta})$ be the $j,k$ entry of the matrix of spectral densities in polar coordinate form, we write the coherence of the introduced model as \begin{align*}
\gamma_{\textrm{C}, jk}(r, \mb{\theta}) &=  \frac{f_{jk}(r, \mb{\theta})}{\sqrt{f_{jj}(r, \mb{\theta}) f_{kk}(r,  \mb{\theta})}}.
\end{align*}
As this is potentially a complex number, \cite{kleiber_coherence_2018} recommends evaluating the squared absolute coherence function $|\gamma_{\textrm{C},jk}(r,  \mb{\theta})|^2$.
For the new multivariate Mat\'ern model presented here, where $\Nu$ is diagonal, we have \begin{align*}
    |\gamma_{\textrm{C},jk}(r,  \mb{\theta})|^2 &= \frac{|\sde_{jk}(d \mb{\theta})|^2}{{\sde_{jj}(d \mb{\theta})\sde_{kk}(d \mb{\theta})}}.
\end{align*}
This does not depend on the frequency $r$ and only depends on the angular component of the spectral density. 
This leads to a relatively interpretable function; for example, if $\sde_{jj}(d \mb{\theta}) = \sigma_{jj} d \mb{\theta}$, $\sde_{kk}(d \mb{\theta}) = \sigma_{kk}d \mb{\theta}$, and $\sde_{jk}(d \mb{\theta}) = \left(\Re(\sigma_{jk}) + \Im(\sigma_{jk})\I \textrm{sign}(\theta_1)\right)d \mb{\theta}$ for $\sigma_{jk} \in \mathbb{C}$, we obtain \begin{align*}
    |\gamma_{\textrm{C},jk}(r,  \mb{\theta})|^2 &= \frac{|\sigma_{jk}|^2}{\sigma_{jj}\sigma_{kk}},
\end{align*}a constant function of $r$ and $ \mb{\theta}$.
While \cite{kleiber_coherence_2018} suggests that a constant coherence function reflects relative inflexibility, this is perhaps expected when there are fewer parameters to influence the behavior of the cross-covariance relationship. We also find that, even with a constant coherence function, such processes can have relatively flexible cross-covariance functions.

Our approach can also be extended to multivariate Mat\'ern models of non-constant coherence.
For example, if the multivariate Mat\'ern of \cite{gneiting_matern_2010} $\mb{C}(h) = [\sigma_{jk}^*\mathcal{M}(|h|; a_{jk}, \nu_{jk}, 1)]$ is valid in $d=1$, then the models with $\mb{\sigma} \in \mathbb{C}^{p\times p}$ in Section \ref{Introduce_model} that set cross-covariance parameters of $\nu_j = \nu_k = \nu_{jk}$ and $a_j =a_k = a_{jk}$ should also be valid if $\sigma_{jj}^* = \sigma_{jj}$ and $|\sigma_{jk}| < \sigma_{jk}^*$, allowing asymmetric covariances while retaining the flexibility of the parameters $\nu_{jk}$ and $a_{jk}$, resulting in non-constant coherences like the multivariate Mat\'ern of \cite{gneiting_matern_2010} as established in \cite{kleiber_coherence_2018}.



\end{document}


\maketitle
\tableofcontents

\section{Representation of spatial cross-covariances for imaginary directional measure}
\label{app:im}

We aim to compute \begin{align*}
    C_{jk}(\hv) &= c_jc_k \int_{0}^\infty \int_{\us^{d-1}} e^{\I \langle \hv, r\mb{\theta}\rangle} \left(a^2 + r^2\right)^{-\nu - \frac{d}{2}} \I \textrm{sign}(\theta_1) r^{d-1} d\mb{\theta} dr.
\end{align*}
Focusing on the inner integral in the special case where $\hv = b \mb{e}_1$, we have \begin{align*}
&\I\int_{\us^{d-1}} e^{\I\langle \hv, r\mb{\theta}\rangle} \textrm{sign}(\theta_1) d\mb{\theta}=\I \int_{\us^{d-1}, \theta_1 > 0}e^{\I b\theta_1r} d\mb{\theta} -  \I\int_{\us^{d-1}, \theta_1 < 0}e^{\I b\theta_1r} d\mb{\theta}. \end{align*}
Applying $e^{\I b\mb{\theta} r} = \cos(b\theta_1 r) + \I \sin(b\theta_1r)$ and using the even and odd properties of cosine and sine, we obtain \begin{align*}
    \I\int_{\us^{d-1}} e^{\I\langle \hv, r\mb{\theta}\rangle} \textrm{sign}(\theta_1) d\mb{\theta}&= -2\int_{\us^{d-1}, \theta_1 > 0} \sin(b\theta_1r) d\mb{\theta}
\end{align*}
If we approach this integral similarly to page 154 of \cite{stein_introduction_1975} or page 43 of \cite{stein_interpolation_2013} by making the change of variables to $\phi$, the angle that $\mb{\theta}$ makes with $\mb{e}_1$, we have \begin{align*}
\int_{\us^{d-1}, \theta_1 > 0}\sin(\theta_1 br)d\mb{\theta} &\propto \int_0^{\pi/2}\sin(\cos(\phi) br) \sin^{d-2}(\phi)d\phi\\&\propto(br)^{-(d-2)/2} \struveH_{(d-2)/2}(br)
\end{align*}using a standard integral representation of the Struve function $\struveH_\nu(z)$, defined as \begin{align*}
    H_{\nu}(z) =
    \sum_{m=0}^\infty \frac{(-1)^m (\frac{1}{2}z)^{2m + \nu + 1}}{\Gamma(m + \frac{3}{2})\Gamma(m + \nu + \frac{3}{2})}.
\end{align*}
See Chapter 11 of \cite{NIST:DLMF} or Chapter 12 of \cite{abramowitz_handbook_1972} for more information. 
Then, we can apply 6.814 of \cite{GR_table_2015} to the overall integral, resulting in a cross-covariance in this direction proportional to $$\textrm{sign}(b)\left|ba_j^{-1}\right|^{\nu_j} \left(\struveL_{-\nu_j}(a_j|b|) - \besselI_{\nu_j}(a_j|b|)\right).$$ 

\section{Additional simulation results}

We provide more simulation results for $d=1$, in particular looking at different models for the cross covariance:

\begin{enumerate}[label=\Alph*.]
    \item The cross-covariance \begin{align*}
        C_{12}(h) &= c_1c_2\int_\mathbb{R} e^{\I hx} (a_1 + \I x)^{-\nu_1 - \frac{1}{2}} \left[\Re(\sigma_{12}) + \I\textrm{sign}(x)\Im(\sigma_{12})\right] (a_2 - \I x)^{-\nu_2 - \frac{1}{2}} dx,
    \end{align*}as presented in Section 3 of the main text. The results are the presented in the main text. 

    \item The cross-covariance \begin{align*}
        C_{12}(h) &= c_1c_2\int_\mathbb{R} e^{\I hx} (a_1^2 +  x^2)^{-\frac{\nu_1}{2} - \frac{1}{4}} \left[\Re(\sigma_{12}) + \I\textrm{sign}(x)\Im(\sigma_{12})\right] (a_2^2 +  x^2)^{-\frac{\nu_2}{2} - \frac{1}{4}} dx,
    \end{align*}the alternate factorization considered in, for example, \cite{bolin_multivariate_nodate}.

    \item The cross-covariance \begin{align*}
        C_{12}(h) &= c_{12}\int_\mathbb{R} e^{\I hx} (a_{12}^2 +  x^2)^{-\nu_{12} - \frac{1}{2}} \left[\Re(\sigma_{12}) + \I\textrm{sign}(x)\Im(\sigma_{12})\right]dx,
    \end{align*}with additional parameters $a_{12}$ and $\nu_{12}$ and $c_{12} = a_{12}^{2\nu_{12}}\Gamma(\nu_{12} + 1/2)/[\pi^{1/2}\Gamma(\nu_{12})]$, which corresponds to the multivariate Mat\'ern of \cite{gneiting_matern_2010} when $\Im(\sigma_{12}) = 0$. 
\end{enumerate}

We consider both generating data from these models and estimating these models. 
Due to the additional parameters and validity constraints in Model C, we are unsure if we have appropriately implemented it; the main goal of the simulation study is instead to evaluate real versus complex $\sigma_{12}$.  
We present results for the likelihood ratio test in Table \ref{tab:lrt}. 
In general, we see that the test performs well when estimating model A and B, with Type I errors near the 0.050 levels; however, Type I error seems better controlled when estimating the correct cross-covariance model. 
When estimating with model C, Type I errors are larger. 
All test have substantial power above the $0.05$ level under the two alternative hypotheses considered here. 

\begin{table}[ht]
    \centering
    \begin{tabular}{|c|c|c|c|c|}\hline
         True model & Estimated model 
         &
         $\sigma_{12} = 0.4$ & $\sigma_{12} = 0.4\mathbbm{i}$ & $\sigma_{12} = 0.4 + 0.4 \mathbbm{i}$  \\ \hline
         A & A 
         & 0.03 & 0.66& 0.59\\ 
         A & B
         & 0.08 & 0.65 & 0.90\\ 
         A & C & 0.14 & 0.70 & 0.80\\ \hline
         B & A & 0.10 & 0.62 & 0.40\\ 
         B & B & 0.06 & 0.70 & 0.74\\ 
         B & C & 0.14 & 0.73 & 0.73\\ \hline
         C & A & 0.06 & 0.44 & 0.28\\ 
         C & B & 0.06 & 0.48 & 0.50\\ 
         C & C & 0.17 & 0.53 & 0.55\\ \hline
    \end{tabular}
    \caption{Proportion of simulations with significant likelihood comparison test for imaginary component with nominal level $0.050$. The column $\sigma_{12} = 0.4$ refers to the Type I error rate. }
    \label{tab:lrt}
\end{table}

In Table \ref{tab:par_est}, we compare estimation of $\sigma_{12}$ when the true model is A for all estimating models A, B, and C. 
Even if the cross-covariance model is misspecified, the general direction and strength of the components of $\sigma_{12}$ are estimated well using maximum likelihood. 

\begin{table}[ht]
    \centering
    \begin{tabular}{|c|c c|c c|c c|}\hline
          & $\Re(\sigma_{12})$ &$\Im(\sigma_{12})$  & $\Re(\sigma_{12})$ &$\Im(\sigma_{12})$ &  $\Re(\sigma_{12})$ &$\Im(\sigma_{12})$  \\ \hline
          True parameter & $0.40$ & $0.00$ & $0.00$ & $0.40$ & $0.40$ & $0.40$\\ \hline
          A (Complex) & $0.38$ & $- 0.02$ & $0.00$ &  $0.38$ & $0.37$  & $0.36$\\ 
          A (Real) & $0.40$ & - & $0.01$ & -  & $0.52$ & - \\ 
         B (Complex) & $0.38$ &  $0.07$  & $-0.11$ & $0.37 $ & $0.28$ & $0.45$\\ 
          B (Real) & $0.38$ & - & $-0.11$ & - & $0.28$ & - \\ 
          C (Complex)& $0.42$ &$0.10$ & $-0.09$ & $0.43 $ & $0.29$ &$0.48$\\ 
          C (Real)& $0.47$ & - & $-0.02$  & - & $0.44$ & -\\ \hline
    \end{tabular}
    \caption{Mean over 200 simulations of real and complex parts of the estimate of $\sigma_{12}$ when the true model is A. }
    \label{tab:par_est}
\end{table}

\section{Additional data analysis details}\label{app:data}

\subsection{Data analysis on Pacific Northwest pressure and temperature data}

The data consists of $n=157$ of bivariate measurements indexed in $\mathbb{R}^2$. 
After converting relevant distances to kilometers, the maximum distance between two locations is $1{,}561.6$ kilometers. 
Thus, for the discrete Fourier transform, we use $2^{10} = 1{,}024$ points for each dimension evenly-spaced between $-2{,}500$ and $2{,}500$ kilometers. 

We now present the results comparing the models as applied to the data. 
For equal comparison, we use the Fourier-transform-based optimization of the likelihood for each model. 
In Table \ref{tab:tab1}, we provide details of the estimated maximized log-likelihoods of the models. 
First, in comparing the log-likelihoods, the multivariate Mat\'ern of \cite{gneiting_matern_2010} and the models introduced here have a higher maximum likelihoods and mostly lower Akaike information criteria (AIC) compared to the independent Mat\'ern and single covariance function models. 
This suggests moderately better fits of SMM-0 and SMM-R in fitting the cross-covariance of pressure and temperature. 
In comparing the three SMM models, we see that increasing complexity of the model results in a larger maximized log-likelihood yet also higher AIC values. 
This suggests that increased complexity in the SMM-C model may not be suitable for this dataset.

We also present the optimized parameters in Tables \ref{tab:tab2} and \ref{tab:tab3} for temperature (process 1) and pressure (process 2). 
For the most part, the estimated parameters and log-likelihoods are also consistent with previous analyses of the data \citep{gneiting_matern_2010}; differences \citep[for example, an estimated $\nu_1 = 1.50$ is reported in][]{gneiting_matern_2010} are likely due to the Fourier transform estimation scheme. 
The estimated parameters of SMM models are mostly consistent with each other and the multivariate Mat\'ern of \cite{gneiting_matern_2010}. 
The SMM models introduced in this paper have a stronger correlation parameter compared to the other models, which may be expected since the cross-correlation functions often take maximum absolute value substantially less than $1$. 
The estimate of $\Im(\sigma_{12})$ suggests marginal evidence of covariance asymmetry through $\sigma_{12}$ for this data.

\begin{table}[ht]
    \centering
    \begin{tabular}{c|c|c|c|}
       Model & Log-likelihood & $\#$ Parameters &  AIC \\ \hline
       IM & $-1272.304$ & 8 &  $2560.608$ \\ 
       SCF  & $-1266.664$ & 7 & $2547.328$ \\ 
       MMG & $-1261.418$ & 11 &  $2544.836$\\ 
       SMM-0 & $-1261.336$ & 9 & $ 2540.672$ \\ 
       SMM-R & $-1260.610$ & 10 & $2541.220$ \\ 
       SMM-C & $-1260.537$ & 12 & $2545.074$ \\ 
    \end{tabular}
    \caption{Multivariate Mat\'ern maximum likelihoods, number of parameters, and Akaike information criterion (AIC) values.}
    \label{tab:tab1}
\end{table}

\begin{table}[ht]
    \centering
    \begin{tabular}{c|c cc cc c|}
       Model & $\nu_{1}$  & $\nu_{2}$ & $\nu_{12}$ & $a_1$ & $a_2$& $a_{12}$ \\ \hline 
       IM &  $7.00$ & $0.56$& - & $3.3 \cdot 10^{-2}$ &$1.0 \cdot 10^{-2}$  & -\\ 
       SCF  & $0.50$ & - & - & $6.1 \cdot 10^{-3}$ & - & -   \\ 
       MMG & $4.27$ & $0.59$ & $3.23$ & $2.4\cdot 10^{-2}$ & $1.1 \cdot 10^{-2}$ & $0.024$  \\ 
       SMM-0  & $1.13$ & $0.55$ &- &$8.7 \cdot 10^{-3}$ &$9.6 \cdot 10^{-3}$ &- \\ 
       SMM-R & $1.14$ & $0.56$ &- &$8.6 \cdot 10^{-3}$ &$9.9 \cdot 10^{-3}$ &- \\ 
       SMM-C  & $1.16$ & $0.56$ &- &$8.7 \cdot 10^{-3}$ &$9.9 \cdot 10^{-3}$ &- \\ 
    \end{tabular}
    \caption{Multivariate Mat\'ern estimated smoothness and range parameters for the Northwest temperature and pressure data. A dash indicates that the given parameter does not exist under that model.}
    \label{tab:tab2}
\end{table}

\begin{table}[ht]
    \centering
    \begin{tabular}{c|cccccc|}
       Model &  $\sqrt{\sigma_{11}}$ & $\sqrt{\sigma_{22}}$ & $\frac{\Re(\sigma_{12})}{\sqrt{\sigma_{11}\sigma_{22}}}$ &$\frac{\Im(\sigma_{12})}{\sqrt{\sigma_{11}\sigma_{22}}}$ & $\gamma_{1}$  & $\gamma_{2}$  \\ \hline
       IM &   $ 224.9$ & $2.63$ & - & - & $71.37$ & $1.27\cdot 10^{-2}$ \\ 
       SCF  & $198.6$ & $3.02$ & $-0.389$ & - & $48.03$ & $1.70\cdot 10^{-3}$ \\ 
       MMG & $ 225.7$ & $2.63$ & $-0.554$ & -  & $72.03$ & $2.24 \cdot 10^{-2}$   \\ 
       SMM-0 &  $226.4$ & $2.66$ & $-0.661$ & -  & $68.16$ &  $3.64 \cdot 10^{-3}$  \\ 
       SMM-R &  $226.2$ &$2.65$ & $-0.683$ & - & $69.19$ &  $3.64 \cdot 10^{-3}$   \\ 
       SMM-C &  $227.4$ &$2.65$ & $-0.685$ & $-0.039$ & $69.38$ &  $3.64 \cdot 10^{-3}$  \\ 
    \end{tabular}
    \caption{Estimated variance parameters for the Northwest temperature and pressure data. A dash indicates that the given parameter does not exist under that model.}
    \label{tab:tab3}
\end{table}

In Figure \ref{fig:data_fun_plot}, we plot the estimated cross-covariance functions for each of the models. 
We see that the single covariance function and the multivariate Mat\'ern of \cite{gneiting_matern_2010} are isotropic and covariance-symmetric, while each of the SMM have more flexible form. 
Furthermore, the estimated cross-covariance between the processes is substantially larger for lags of approximately 50-300 kilometers for the SMM models compared to the other models.

\begin{figure}[ht]
    \centering
    \includegraphics[scale=.41]{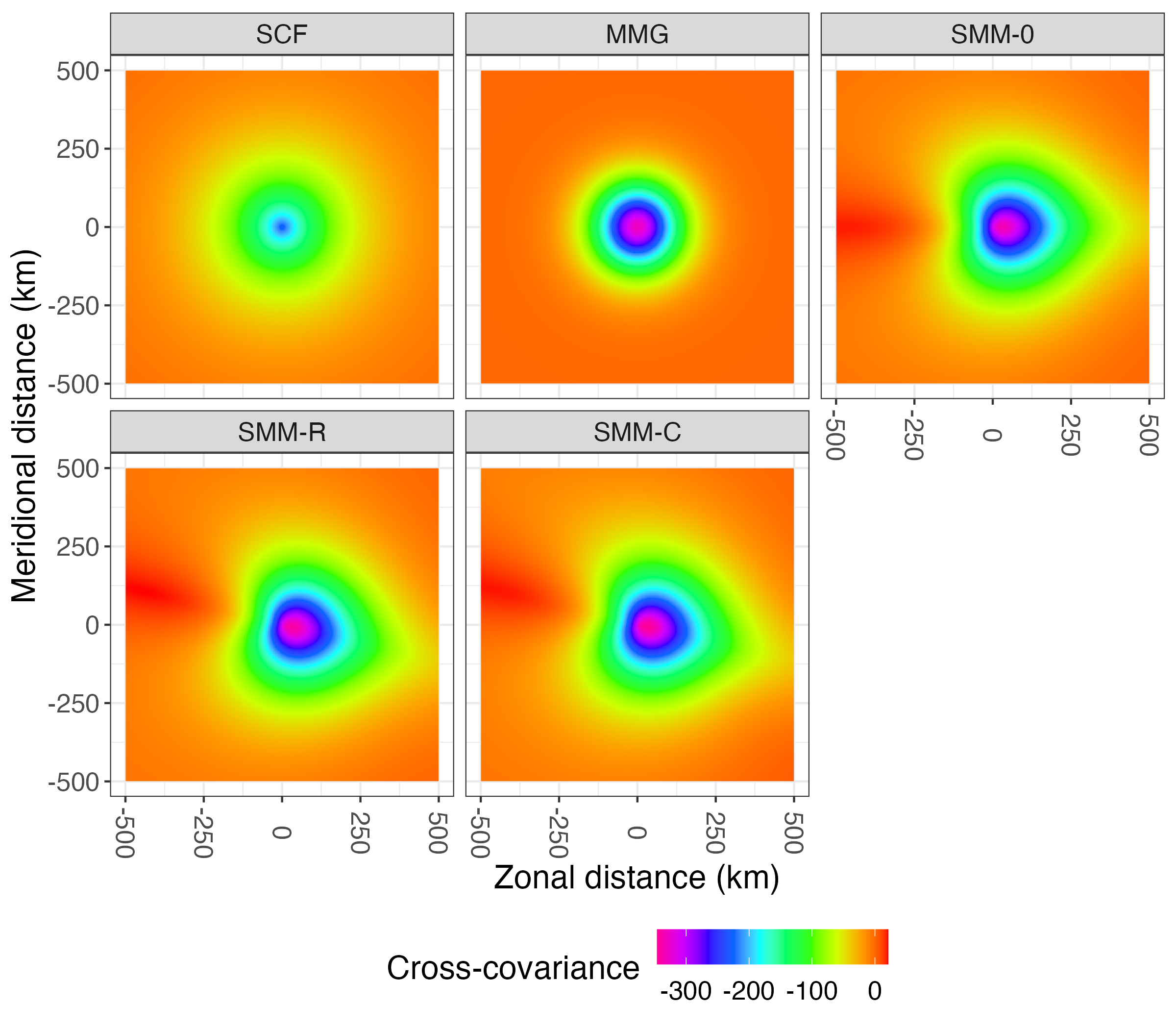}
    \caption{Estimated cross-covariance functions on the Pacific Northwest pressure and temperature data. Zonal distance refers to distance in the east-west (longitude) direction, and meridional distance refers to distance in the north-south (latitude) direction.}
    \label{fig:data_fun_plot}
\end{figure}

To compare predictive performance, we evaluate the estimated models in their prediction using cross-validation. 
Predictions are formed using standard expressions for the conditional expectation. 
We consider 5-fold and $n$-fold cross validation, where a proportion of $1/5$ or $1/n$ of data points are left out of the dataset for one variable, and a prediction is formed using the rest of the data of that variable as well as all data of the other variable. 
We repeat this for both variables comparing root-mean-squared-error averaged across the folds. 
For comparison, we also predict using only the data of the target variable, as well as only data using the other variables. 
Results are presented in Table \ref{tab:cokriging}.
Using both variables improves using only one of the variables at a time, and prediction errors are expectedly lower for $n$-fold compared to 5-fold  cross-validation.
In general, there are not large gaps in prediction performance between the multivariate Mat\'ern of \cite{gneiting_matern_2010} and the SMM models, and for the most part the SMM provide slightly improved point prediction performance. 
Overall, we find that the SMM models introduced in this paper can fit as well or better than the multivariate Mat\'ern of \cite{gneiting_matern_2010} on this standard dataset. 

\begin{table}[ht]
    \centering
    \begin{tabular}{c|c|c|c|c|c|}\hline
    Model &  5f-both & 5f-univariate & $n$f-both  & $n$f-univariate & other \\ \hline
    Prediction of zero & 194.2494 & 194.2494 & 194.2494&194.2494&194.2494\\ 
    SCF    & 125.2853 &  {\bf 132.5836} & 172.9075 &    123.4839 & 172.9075 \\ 
    MMG    & 126.6086 &  134.0422  & 117.2287& {\bf 119.0559} &176.6753 \\ 
    SMM-0    & 122.7652 & 132.8396   & 116.7819 & 121.1790 &  164.7258 \\ 
    SMM-R    & 123.1258 &  132.8966 & 116.5281 &  121.2294 &  {\bf 164.6401}\\ 
    SMM-C    &  {\bf 122.8856} &  132.9320 &  {\bf 116.4606} & 121.2207 &  165.7362 \\ \hline
    \end{tabular}

    \begin{tabular}{c|c|c|c|c|c|}\hline
    Model & 5f-both & 5f-univariate & $n$f-both  & $n$f-univariate & other \\ \hline
    Prediction of zero &2.709560 & 2.709560 &  2.709560 & 2.709560 & 2.709560 \\ 
    SCF    & 1.672005 &  1.693611 &  1.582059 & {\bf 1.620309} & 2.391110 \\ 
    MMG    & 1.631760 &  1.690111  & 1.545588 & 1.624545 & 2.390009 \\
    SMM-0    &  1.631041 &  {\bf 1.689291}  & 1.560799 &1.624241 &  2.228422 \\ 
    SMM-R    &{\bf  1.611251}  & 1.689463  & {\bf 1.540141} & 1.624176 & {\bf 2.225494}\\ 
    SMM-C    & 1.612860 & 1.689452  & 1.542120  & 1.624203 &  2.229628  \\ \hline
    \end{tabular}

    \caption{Cross-validation root-mean-squared error (RMSE) results compared for 5-fold (5f) and $n$-fold ($n$f) cross-validation. Columns with ``both'' use both variables to predict at the left-out points, ``univariate'' uses just the target variable, and ``other'' uses just the other variable (and thus does not depend on the number of folds). These values are compared with the null prediction of $0$ as a baseline. The lowest values in each column are bold. (Top) For pressure. (Bottom) For temperature.  }
    \label{tab:cokriging}
\end{table}

\subsection{Housing market data}

In Figure \ref{fig:housing_cor}, we plot the estimated correlation functions for the housing market data.

\begin{figure}[ht]
    \centering
    \includegraphics[width = .48\textwidth]{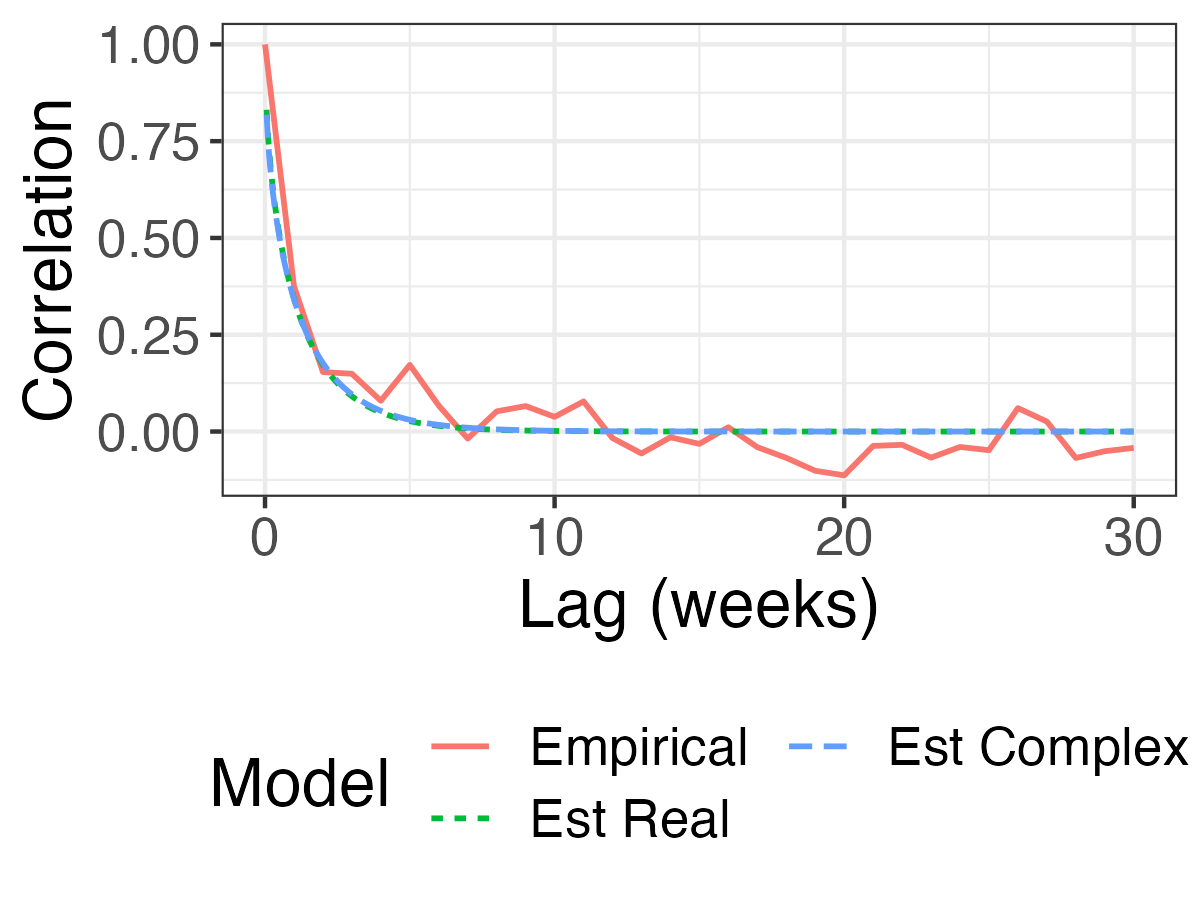}
    \includegraphics[width = .48\textwidth]{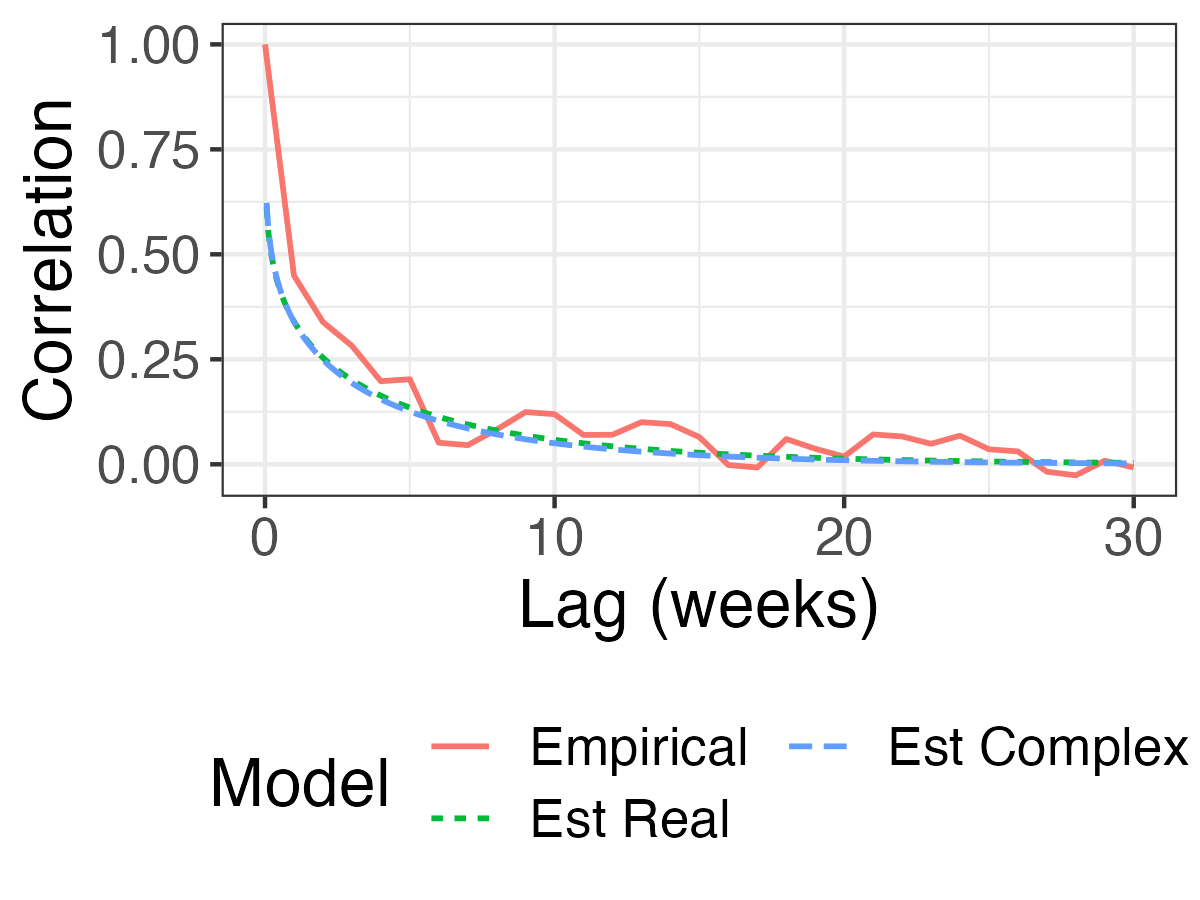}
    \caption{Correlation functions for median sale price (Left) and inventory (Right). 
    ``Empirical'' refers to nonparametric estimators from the \texttt{stats::acf} functions in \texttt{R}. ``Est Real'' and ``Est Complex'' refer to models in Section 3 of the main text, assuming $\Im(\sigma_{12}) = 0$ and both $\Re(\sigma_{12})$ and $\Im(\sigma_{12})$ free, respectively. Note that each correlation function takes values of $1$ at lag $h=0$. }
    \label{fig:housing_cor}
\end{figure}

\subsection{Argo data}

In Tables \ref{tab:argo1}, \ref{tab:argo2}, and \ref{tab:argo3}, we present the parameter and likelihood results for the Argo temperature data. In general, we see larger estimated values of $\nu_1$ and $\nu_2$. 
More notably, there is substantial evidence for an asymmetric cross-covariance between the measurements at different depths. 

\begin{table}[ht]
    \centering
    \begin{tabular}{c|c|c|c|}
       Model & Log-likelihood & $\#$ Parameters &  AIC \\ \hline 
       IM & $22.28663$ &8 &  $-28.57326$ \\ 
       SCF  & $34.53614$ & 7 & $ -55.07228$ \\ 
       MMG & $36.64354$ &11&  $-51.28708$\\ 
       SMM-0 & $38.34603$ & 9 & $ -58.69206$ \\ 
       SMM-R & $37.66356$ & 10 & $-55.32712$ \\ 
       SMM-C & $37.16868$ & 12 & $-50.33736$ \\ \hline
    \end{tabular}
    \caption{Multivariate Mat\'ern maximum likelihoods, number of parameters, and Akaike information criterion (AIC) values for the Argo temperature data.}
    \label{tab:argo1}
\end{table}

\begin{table}[ht]
    \centering
    \begin{tabular}{c|c|c|c|c|c|c|}
       Model & $\nu_{1}$  & $\nu_{2}$ & $\nu_{12}$ & $a_1$ & $a_2$& $a_{12}$ \\ \hline 
       IM &  $2.31$ & $2.09$& - & $2.9 \times 10^{-2}$ &$2.5 \times 10^{-2}$  & -\\ \hline
       SCF  & $2.22$ & - & - & $2.8 \times 10^{-2}$ & - & -   \\ \hline
       MMG & $ 2.10$ & $3.04$ & $3.91$ & $2.7\times 10^{-2}$ & $3.1 \times 10^{-2}$ & $0.034 $  \\ \hline
       SMM-0  & $4.01$ & $1.41$ &- &$4.7 \times 10^{-2}$ &$1.7 \times 10^{-2}$ &- \\ \hline
       SMM-R & $ 3.98$ & $1.42$ &- &$4.7 \times 10^{-2}$ &$1.7 \times 10^{-2}$ &- \\ \hline
       SMM-C  & $4.17$ & $1.57$ &- &$5.1 \times 10^{-2}$ &$2.0 \times 10^{-2}$ &- \\ \hline
    \end{tabular}
    \caption{Multivariate Mat\'ern estimated smoothness and range parameters for the Argo temperature data. A dash indicates that the given parameter does not exist under that model.}
    \label{tab:argo2}
\end{table}

\begin{table}[ht]
    \centering
    \begin{tabular}{c|c|c|c|c|c|c|c|}
       Model &  $\sqrt{\sigma_{11}}$ & $\sqrt{\sigma_{22}}$ & $\frac{\Re(\sigma_{12})}{\sqrt{\sigma_{11}\sigma_{22}}}$ &$\frac{\Im(\sigma_{12})}{\sqrt{\sigma_{11}\sigma_{22}}}$ & $\gamma_{1}$  & $\gamma_{2}\cdot 10^4$  \\  \hline
       IM &   $0.75$ & $0.029$ & - & - & $4.2 \times 10^{-7}$ & $1.2$ \\ \hline
       SCF  & $0.76$ & $0.026$ & $0.522$ & - & $1.0 \times 10^{-9}$ & $2.1$ \\ \hline
       MMG & $ 0.74$ & $0.029$ & $0.614$ & -  & $5.9\times 10^{-6}$ & $1.9 $   \\ \hline
       SMM-0 &  $0.69$ & $0.029$ & $0.679$ & -  & $1.7 \times 10^{-7}$ &  $1.1$  \\ \hline
       SMM-R &  $0.69$ &$0.029$ & $0.678$ & - & $1.7 \times 10^{-7}$ &  $1.1$   \\ \hline
       SMM-C &  $0.64$ &$0.028$ & $0.638$ & $0.275$ & $9.6 \times 10^{-5}$ &  $2.2 $  \\ \hline
    \end{tabular}
    \caption{Estimated variance parameters for the Argo temperature data. A dash indicates that the given parameter does not exist under that model.}
    \label{tab:argo3}
\end{table}

\section{Proofs}

\begin{proof}[Proof of Proposition 2.7] {\em (i):} Suppose $\mb{C}(\mb{s})=\overline{\mb{C}(\mb{s})},\ \mb{s}\in\R^d$. By using Bochner's Theorem, we have 
\begin{align*}
\mb{C}(\mb{s}) = \overline{\mb{C}(\mb{s})} &= \int_{\R^d} e^{-\ii \langle \mb{s},\mb{x}\rangle} \overline{\mb{F}(d\mb{x})}\\
&= 
\int_{\R^d} e^{\ii \langle \mb{s},\mb{x}\rangle} \overline{\mb{F}(-d\mb{x})},
\end{align*}
where the last relation follows from a change of variables.  This shows that $\mb{C}(\mb{s})$ has a representation as in Bochner's theorem with 
$\mb{F}(d\mb{x})$ replaced by $\overline{\mb{F}(-d\mb{x})}$. The uniqueness of the measure $\mb{F}$ in Bochner's theorem entails $\overline{\mb{F}(-d\mb{x})} = \mb{F}(d\mb{x})$.
Conversely, with the same argument $\overline{\mb{F}(-d\mb{x})} = \mb{F}(d\mb{x})$ implies that $\mb{C}(\mb{s})$ is real, completing the proof of {\em (i)}.
The proof of {\em (ii)} follows in the same way as that of {\em (i)} from the uniqueness of the measure $\mb{\xi}$ in the Cram\'er theorem.

{\em (iii):} The implication {\em(ii)}$\Rightarrow${\em (i)} is immediate.  To see that the converse 
is not true, note that the processes $\{z\mb{Y}(\mb{s}),\ \mb{s}\in\R^d\}$ have the same covariance structure for 
all $z\in \C,\ |z|=1$.

Finally, let $\mb{Y}(\mb{s}) = \mb{Y}_1(\mb{s}) + \ii \mb{Y}_2(\mb{s}),$ where $\mb{Y}_i(\mb{s})$ are real. If $\mb{Y}$ has a real auto-covariance $\mb{C}$, then,
\begin{align*}
\mb{C}(\mb{t}-\mb{s}) &= \E\left[ \mb{Y}_1(\mb{t}) \mb{Y}_1(\mb{s})^\top + \mb{Y}_2(\mb{t}) \mb{Y}_2(\mb{s})^\top\right]\\
&~~~~~~~~~+ \ii \E\left[\mb{Y}_2(\mb{t}) \mb{Y}_1(\mb{s})^\top - \mb{Y}_1(\mb{t}) \mb{Y}_2(\mb{s})^\top \right] \\
&=: \mb{C}_1(\mb{t}-\mb{s}) + \ii \mb{C}_2(\mb{t}-\mb{s}),
\end{align*}
for all $\mb{t},\mb{s}\in\R^d.$
Now, since $\mb{C}$ is real, we have $\mb{C}(\cdot)\equiv \mb{C}_1(\cdot)$ and $\mb{C}_2(\cdot)\equiv \mb{0}$. By
possibly considering another probability space, take independent, zero-mean processes $\{\widetilde {\mb{Y}}_i(\mb{s})\},\ i=1,2$ such that 
$\{\mb{Y}_i(\mb{s})\} \stackrel{d}{=}\{\widetilde{\mb{Y}}_i(\mb{s})\}$.  Clearly, since $\E[\widetilde{\mb{Y}}_1(\mb{t})\widetilde{\mb{Y}}_2(\mb{s})^\top] =\mb{0}$, we have that the real process
$\widetilde{\mb{Y}}(\mb{s}):= \widetilde{\mb{Y}}_1(\mb{s}) + \widetilde{\mb{Y}}_2(\mb{s})$ will have auto-covariance $\mb{C}$.
\end{proof}

\section{Additional properties of multivariate Mat\'ern random fields}

\subsection{Existence}

First, we establish that the process $\Y(\ti)$ for general $d$ is well-defined. 

\begin{lemma}\label{lemma:exists}
Assume that $a_k >0$ for $k=1, \dots, p$, $\Nu$ is a real, symmetric, and positive-definite matrix which may be non-diagonal, and that $\mb{c}$ is a real, symmetric, and positive-definite matrix with finite eigenvalues. 
The process $\Y(\ti)$ is well-defined in the sense that \begin{align*}
    \int_0^\infty \int_{\us^{d-1}} e^{\I\langle \ti, r\mb{\theta}\rangle}\mb{c}(\amat + \ang(\mb{\theta})\I r\Imat_p)^{-\Nu - (d/2)\Imat_p} \rmeasure(dr,d\mb{\theta})
\end{align*}exists for all $\ti \in \mathbb{R}^d$. 
Also, $\Y(\ti)$ is a stationary and Gaussian $p$-variate random field. 
\end{lemma}

\begin{proof}[Proof of Proposition \ref{lemma:exists}]
We approach the proof in a similar way to the proof of Proposition 5.16 in \cite{shen_tangent_2022}. 
To prove that the integral is valid, we can examine \begin{align*}
    &\int_0^\infty \int_{\us^{d-1}}\left\lVert e^{\I\langle \ti, r\mb{\theta}\rangle}\mb{c}(\amat +  \ang(\mb{\theta})\I r\Imat_p)^{-\Nu - (d/2)\Imat_p}\right\rVert_{\textrm{op}}^2 r^{d-1}dr \left\lVert\sd\right\rVert_{F}(d\mb{\theta}) \\  
    &=  \int_0^\infty \int_{\us^{d-1}} \left|e^{\I\langle \ti, r\mb{\theta}\rangle} \right|^2 \left\lVert \mb{c}(\amat + \ang(\mb{\theta}) \I r \Imat_p)^{-\Nu - (d/2) \Imat_p}\right\rVert^2_{\textrm{op}}  r^{d-1} dr \left\lVert\sd\right\rVert_{F}(d\mb{\theta}),\end{align*}where the norms $\lVert \cdot \rVert_{\textrm{op}}$ and $\lVert \cdot \rVert_{F}$ on $p\times p$ matrices are the operator norm induced by the Euclidean norm and the Frobenius norm, respectively.
    As $\left|e^{\I\langle \ti, r\mb{\theta}\rangle} \right|^2 = 1$, we turn our attention to \begin{align*}
        \left\lVert \mb{c} (\amat + \I r \Imat_p)^{-\Nu - \frac{d}{2} \Imat_p}\right\rVert^2_{\textrm{op}} = \left\lVert \mb{c}\left(\amat^2 + r^2 \Imat_p\right)^{-\Nu - \frac{d}{2}  \Imat_p}\right\rVert_{\textrm{op}}.
    \end{align*}
    As the exponent has only negative eigenvalues, \begin{align*}
        \left\lVert \mb{c} \left(\amat^2 + r^2 \Imat_p\right)^{-\Nu - \frac{d}{2}  \Imat_p}\right\rVert_{\textrm{op}} &\leq \left\lVert \mb{c}\left(r^2\Imat_p\right)^{-\Nu -\frac{d}{2} \Imat_p}\right\rVert_{\textrm{op}} \\
        &\leq \left\lVert \mb{c}\right\rVert_{\textrm{op}}\left\lVert r^{-2\Nu - d \Imat_p}\right\rVert_{\textrm{op}}.
    \end{align*}
    From here, we see that \begin{align*}
        \left\lVert r^{-2\Nu - d \Imat_p}\right\rVert_{\textrm{op}} &= r^{-2 \min\{\textrm{eig}(\Nu)\} - d}.
    \end{align*}where $\textrm{eig}(\Nu)$ denotes the eigenvalues of $\Nu$. 
    Therefore, we have the bound\begin{align*}
        &\int_0^\infty \int_{\us^{d-1}}\left\lVert e^{\I\langle \ti, r\mb{\theta}\rangle}\mb{c}(\amat + \ang(\mb{\theta}) \I r\Imat_p)^{-\Nu - (d/2)\Imat_p} \right\rVert_{\textrm{op}}^2r^{d-1} \left\lVert\sd\right\rVert_{F}(d\mb{\theta})dr\\
        &~~~~~\leq \left\lVert \mb{c}\right\rVert_{\textrm{op}}\int_0^\infty \int_{\us^{d-1}} r^{-2\min\{\textrm{eig}(\Nu)\} - 1} \left\lVert\sd\right\rVert_{F}(d\mb{\theta})dr,
    \end{align*}which, with respect to $r$, is integrable for large $r$ when $-2\min\{\textrm{eig}(\Nu)\}-1 < -1$, or equivalently, $\min\{\textrm{eig}(\Nu)\} > 0$. 
    If the matrix $\Nu$ is diagonal, the restriction $\min\{\textrm{eig}(\Nu)\} > 0$ is equivalent to having all diagonal entries satisfying $\nu_k > 0$, as expected.
For small $r$, we can make the bound \begin{align*}
     \left\lVert \mb{c} \left(\amat^2 + r^2 \Imat_p\right)^{-\Nu - (d/2) \Imat_p}\right\rVert_{\textrm{op}} &\leq \left\lVert \mb{c}\amat^{-2\Nu - d \Imat_p}\right\rVert_{\textrm{op}},
\end{align*}so that \begin{align*}
     &\int_0^\infty \int_{\us^{d-1}}\left\lVert e^{\I\langle \ti, r\mb{\theta}\rangle}\mb{c}(\amat + \ang(\mb{\theta}) \I r\Imat_p)^{-\Nu - (d/2)\Imat_p}\right\rVert_{\textrm{op}}^2 r^{d-1} \left\lVert\sd\right\rVert_{F}(d\mb{\theta})dr\\
        &~~~~~\leq \left\lVert \mb{c}\right\rVert_{\textrm{op}} \left\lVert \amat^{-2\Nu - d \Imat_p}\right\rVert_{\textrm{op}}\int_0^\infty \int_{\us^{d-1}} r^{d-1} \left\lVert\sd\right\rVert_{F}(d\mb{\theta})dr.
\end{align*}The integral $\int_0^\infty r^{d-1} dx$ is integrable near $0$, and we obtain $\left\lVert \amat^{-2\Nu - d \Imat_p}\right\rVert_{\textrm{op}} < \infty$ if the parameters satisfy $\min \{a_k\} > 0$ and $\min\{\textrm{eig}(\Nu)\} > 0$. 
Thus, the integral across all $r \in (0,\infty)$ is finite. 
Finally, by using the linearity property of the trace, it remains to establish that $\left\lVert\sd(\us^{d-1})\right\rVert_{\textrm{F}} < \infty$, which follows from the assumed properties of $\sd(d\mb{\theta})$. 
Therefore, the stochastic integral exists under the stated conditions.

Stationarity follows from \begin{align*}
    &\Cov(\Y(\ti), \Y(\ti + \hv)) 
     \\
    &~~=\int_0^\infty \int_{\us^{d-1}} e^{-\I\langle \hv, r\mb{\theta}\rangle}\mb{c}(\amat +  \ang(\mb{\theta})\I r\Imat_p)^{-\Nu - (d/2)\Imat_p}\sd(d\mb{\theta})(\amat -  \ang(\mb{\theta})\I r\Imat_p)^{-\Nu -(d/2)\Imat_p}r^{d-1} dr 
\end{align*}which only depends on $\hv$ and not $\ti$. 
Gaussianity follows from the assumption that $\rmeasure(dr, d\mb{\theta})$ is Gaussian. 
\end{proof}

In addition to this result, the covariances and cross-covariances are real when the spectral density is Hermitian. 
Based on the construction and the fact that $\Nu$ is a symmetric, positive-definite, real-valued matrix, we see that \begin{align*}
     \left((\amat +  \ang(\mb{\theta})\I r\Imat_p)^{-\Nu - (d/2)\Imat_p}\right)^*&= (\amat - \ang(\mb{\theta})\I r\Imat_p)^{-\mb{\nu}^* - (d/2)\Imat_p} 
     =(\amat + \ang(-\mb{\theta})\I r\Imat_p)^{-\mb{\nu} - (d/2)\Imat_p}
\end{align*}so that the matrix $(\amat +  \ang(\mb{\theta})\I r\Imat_p)^{-\Nu - (d/2)\Imat_p}$ is self-adjoint. 
The Hermitian property of the entire spectral density then follows directly from the properties of $\rmeasure(dr, d\mb{\theta})$.
Therefore, for all $(r, \mb{\theta})$, \begin{align*}
   &\left\{ \left[(\amat + \ang(\mb{\theta})\I r\Imat_p)^{-\Nu - (d/2)\Imat_p} \rmeasure(dr,d\mb{\theta})\right]^*\right\}_{\ti\in \mathbb{R}^d} \overset{fdd}{=}\\ &~~~~~~~~~~ \left\{(\amat + \ang(-\mb{\theta})\I r\Imat_p)^{-\Nu - (d/2)\Imat_p} \rmeasure(dr,-d\mb{\theta})\right\}_{\ti\in \mathbb{R}^d}.
\end{align*}

\subsection{Covariance-asymmetry}

We next consider a discussion of covariance asymmetry and the cross-covariance models.

\begin{proposition}\label{prop:reversibility}
Suppose that the process $\Y(\ti) = (Y_{j}(\ti), Y_k(\ti))$ is jointly Gaussian with $\ti \in \mathbb{R}^d$, with new multivariate Mat\'ern covariance and cross-covariance for general $d$, and $\Nu = {\normalfont\textrm{diag}}(\nu_j, \nu_k)$. Then the process is covariance-symmetric if and only if $\nu_j = \nu_k$, $a_j = a_k$, and $\sd(d\mb{\theta}) = \Sig d\mb{\theta}$ for some real, positive-definite matrix $\Sig$. 
Furthermore, when $d > 1$, the cross-covariance is isotropic if and only if the same conditions are met. 
\end{proposition}
 
The case for $d=1$ and with real directional measure is mentioned and established in \cite{klar_note_2015} based on properties of the gamma difference distribution.
As seen in the forms of these cross-covariances, if any of the conditions of Proposition \ref{prop:reversibility} are broken, the spectral density will be complex-valued; this leads to an asymmetric cross-covariance; see Section 3.4.1 of \cite{didier_domain_2018} or \cite{shen_tangent_2022} for discussions on isotropy and symmetry. 
Furthermore, if the conditions for a symmetric process for Proposition \ref{prop:reversibility} are met, the cross-covariance takes the same form as the Mat\'ern covariance. 
Therefore, all new cross-covariances proposed in this work are asymmetric. 

\subsection{Separability}

We say that a multivariate covariance function $\mb{C}(\mb{h})$ is separable if \begin{align*}
    \mb{C}(\mb{h}) &= \left[\sigma_{jk} C(\mb{h})\right]_{j,k=1}^p,
\end{align*}for some univariate covariance function $C(\cdot)$ and matrix $\mb{\sigma} = [\sigma_{jk}]_{j,k=1}^p$. 
In such a case, each covariance and cross-covariance is proportional to the same covariance function. 
In the setting of this paper, separability automatically occurs when $\nu_j = \nu$, $a_j = a$ for all $j=1, \dots, p$, and $\mb{\sigma} \in \mathbb{R}^{p\times p}$. 
Otherwise, the model is nonseparable.

\section{Local properties of multivariate Mat\'ern random fields}\label{sec:tangent_process}


  




 The goal of this section is to provide a general outlook on building spatial covariance models.  We do not claim 
 to provide a comprehensive treatment.  Our main motivation is to study the natural question:

 \vskip .5cm
 
 {\em How flexible is a given stationary covariance model in representing the local covariance structure of an 
 underlying stochastic phenomenon?}

 \vskip .5cm

 We will give a partial answer to this question from the perspective of tangent processes and the associated 
 stationary increment processes. A more complete treatment involving higher order tangent processes and (Hilbert space valued) intrinsic random functions can be 
 found in \cite{shen:stoev:hsing:2020_extended}.

\subsection{Univariate tangent processes}


Let $Y = \{Y(\mb{s}) \in \mathbb{R},\ \mb{s}\in\R^d\}$ be a zero-mean stochastic process with continuous paths. Fix an $\mb{s}_0\in\R^d$ and consider 
the local increments of $Y$ around $\mb{s}_0$:
\begin{equation}\label{e:tangent-increments}
 X_\epsilon(\mb{s}) \equiv X_\epsilon(\mb{s}; \mb{s}_0) := \frac{Y(\mb{s}_0 +\epsilon\mb{s}) - Y(\mb{s}_0)}{c(\epsilon)},
\end{equation}
for some normalizing constants $c(\epsilon)>0$. 

\begin{definition}\label{def:tangent} A non-zero stochastic process $X=\{X(\mb{s})\}$ is said to be a tangent process of 
$Y$ at $\mb{s}_0$ if for the increment processes in \eqref{e:tangent-increments} and some 
normalizing constants $c(\epsilon)\downarrow 0,\ \epsilon \downarrow 0$, we have
$$
\{X_\epsilon(\mb{s})\} \stackrel{d}{\longrightarrow} \{X(\mb{s})\},\ \ \mbox{ as $\epsilon \downarrow 0$,}
$$
where $\stackrel{d}{\to}$ denotes convergence in finite-dimensional distributions. 
\end{definition} 

We are interested in the possible  covariance structures of the tangent processes.  Therefore, we shall assume that 
$Y$ has finite variance and for concreteness and without loss of generality, we will suppose that it is Gaussian.


The seminal works of \cite{falconer:2002,falconer:2003} have established many fundamental properties of the tangent processes.  Notably, 
the tangent processes $X(\cdot;\mb{s}_0)$ are necessarily self-similar and under broad regularity conditions they also 
have stationary increments, for almost all $\mb{s}_0$.  Recall that $X$ is self-similar if there is an exponent $H>0$ such that
\begin{equation}\label{e:H-ss}
\{X(c \mb{s}),\ \mb{s}\in\R^d\}\stackrel{d}{=}\{ c^H X(\mb{s}),\ \mb{s}\in\R^d\},
\end{equation}
for all $c>0$.  Also, $X$ is said to have stationary increments if for all $\mb{h}\in\R^d$,
 $$
 \{X(\mb{s}+\mb{h}) - X(\mb{h}),\ \mb{s}\in\R^d\} \stackrel{d}{=} \{ X(\mb{s}) - X(\mb{0}),\ \mb{s}\in\R^d\},
 $$
 where $\overset{d}{=}$ denotes equality in finite-dimensional distributions. 
 
 Self-similar processes $X$ with an exponent $H$ and stationary increments
 will be referred to as $H$-sssi.  Notice that an $H$-sssi process $X$ is its own tangent process.
 Since $X(\mb{0})=0$, if $X$ has finite variance, using the stationary increments property, one can show that
 the covariance structure of $X$ is completely determined by its variogram (sometimes called semi-variogram) function:
 $$
\gamma_X(\mb{s}) := \frac{1}{2} {\rm Var}(X(\mb{s})) = \frac{1}{2} {\rm Var}( X(\mb{s}+\mb{h}) - X(\mb{h}) ).
 $$
 Namely, for the covariance $C_X(\mb{s},\mb{t}) := {\rm Cov}(X(\mb{s}),X(\mb{t}))$ and $\mb{s},\mb{t}\in\R^d$, we have:
 \begin{equation}\label{e:C_X-via-gamma}
 C_X(\mb{s},\mb{t}) = \gamma_X(\mb{s}) + \gamma_X(\mb{t}) - \gamma_X(\mb{s} - \mb{t}).
 \end{equation}
 On the other hand, the self-similarity property in \eqref{e:H-ss} entails that $\gamma_X(\cdot)$ is a $2H$-homogeneous function, i.e., 
 $\gamma_X(c \mb{s}) = c^{2H} \gamma_X(\mb{s}),\ \mb{s}\in\R^d,$ for all $c>0$. In fact, $H$ is necessarily in the range $(0,1]$
 and the semi-variogram $\gamma_X$ has the following representation.

 \begin{proposition} Let $X$ be an $H$-sssi finite variance process.  Then, $X$ is $L^2$-continuous, 
 $$
 0 < H\le 1
$$
and for $ \mb{s}\in\R^d$, the semi-variogram of $X$ is:
\begin{equation}\label{e:gamma-rep}
\gamma_X(\mb{s})= \left\{ \begin{array}{ll}
 \int_{\us^{d-1}} | \mb{s}^\top \mb{\theta}|^{2H} \sigma(d\mb{\theta}) &,\ \mbox{ if }0<H<1\\
 \mb{s}^\top \Sigma \mb{s} &,\ \mbox{ if }H=1,
\end{array}\right.,
\end{equation}
where $\sigma(d\mb{\theta})$ is a finite symmetric measure on the unit sphere $\us^{d-1}$ and $\Sigma$ is
a positive semidefinite definite (psd) matrix.  

The measure $\sigma$ (matrix $\Sigma$, respectively) in \eqref{e:gamma-rep} is uniquely determined by $\gamma_X$. 
Conversely, for every choice of such a measure $\sigma$ (psd matrix $\Sigma$, respectively) Relation \eqref{e:gamma-rep} yields a
valid semi-variogram of an $H$-sssi process.
\end{proposition}

This result is a consequence of the general representation of the class of continuous negative definite functions dating back to 
\cite{schoenberg:1938} and \cite{von_Neumann:schoenberg:1941}, Theorem 1 therein.  It is also closely related to the L\'evy-Khintchine 
representation of infinitely divisible distributions \citep[see e.g., page 69 of][]{chiles2012geostatistics}.

In view of \cite{falconer:2002,falconer:2003}, the above result characterizes essentially all tangent processes to a Gaussian 
process $Y$.  Namely, these are the spatial counterparts to the celebrated {\em fractional Brownian fields}.

\begin{definition} An $H$-sssi zero mean Gaussian process $X=\{X(\mb{s}),\ \mb{s}\in\R^d\}$
is referred to as a fractional Brownian field (fBf, in short). The self-similarity exponent
$H$ is referred to as the Hurst parameter of the fBf $X$. When $d=1$, fBf are referred to as fractional Brownian motions.
\end{definition}

 If $d=1$, for each value of the Hurst self-similarity parameter $H\in (0,1]$, there is essentially one fBf 
 process up to a scale constant. This is the {\em fractional Brownian motion (fBm)} process, which in view of 
 \eqref{e:C_X-via-gamma}, has covariance structure:
 $$
  C_X(t,s) = \frac{v}{2} \Big( |s|^{2H} + |t|^{2H} - |s-t|^{2H}\Big),\ \ s,t\in\R,
 $$
 where $v = {\rm Var}(X(1))$.
 
 In the spatial setting, when $d\ge 2$, however, the family of fBf processes (with $0<H<1$) is very rich since their covariance structure 
 is parameterized by the infinite-dimensional {\em directional measure} $\sigma$ arising in \eqref{e:gamma-rep}. 
 
 It is instructive to provide the general spectral representation of fBf processes. 
 We do so in polar coordinates, for $\mb{x}\in\R^d$ recalling the notation of its radial and angular components $r:= \|\mb{x}\|$ and $\mb{\theta}:= \mb{x}/\|\mb{x}\|$ so that 
 $\R^d\setminus\{\mb{0}\}$ is homeomorphic to $(0,\infty)\times \us^{d-1}$.  
 Consider an arbitrary finite symmetric measure $\sigma$ on $\us^{d-1}$ and define the measure 
 $F_\sigma (d\mb{x})$ on $\R^d$ such that in polar coordinates, 
 we have 
 \begin{equation}\label{e:F_sigma}
  F_\sigma (d \mb{x}) \equiv F_\sigma(d r, d\mb{\theta}) := r^{d-1} d r \times \sigma( d \mb{\theta}).
 \end{equation}
 
 \begin{proposition} \label{p:fBf-representation} Let $0<H<1$ and $\sigma$ be a finite symmetric measure on $\us^{d-1}$ and consider the 
 complex Hermitian Gaussian measure $W$ on $\R^d$ with control measure $F_\sigma$, i.e.,  such that in polar coordinates, we have
 $$
\E |W(dr,d\mb{\theta})|^2 =F_\sigma(d r, d\mb{\theta}).
 $$

 Then, the stochastic process for $\mb{s}\in\R^d$ of
 \begin{align}\label{e:fBf-representation}\begin{split}
  B(\mb{s}) :&=c_H^{-\frac{1}{2}} \int_0^\infty \int_{\us^{d-1}} (e^{\I \langle \mb{s}, r\mb{\theta}\rangle} -1)  r^{-H-\frac{d}{2}} W(d r,d \mb{\theta}),
  \end{split}
 \end{align}
 is a fractional Brownian field with semi-variogram as in \eqref{e:gamma-rep}, where the constant $c_H$ equals 
 \begin{align}\label{e:c_H}
  c_H &= \int_0^\infty \frac{1-\cos(r)}{r^{2H+1}} dr 
  =\frac{1}{2H} \int_0^\infty \frac{\sin (r)}{r^{2H}} dr \nonumber \\
  &= \left\{  \begin{array}{ll}
   \frac{\Gamma(2-2H) \cos(\pi H) }{2H(1-2H)} &, \mbox{ if }H\not = \frac{1}{2}\\
   \pi/4H &,\mbox{ if }H =\frac{1}{2} \end{array}
   \right..
 \end{align}
 \end{proposition}

  \begin{proof}[Proof of Proposition \ref{p:fBf-representation}]
 In view of \eqref{e:F_sigma} the integrand in \eqref{e:fBf-representation} is square-integrable and the process $B = \{B(\mb{s})\}$ is well-defined.
 Since this process $B$ is real-valued, Gaussian with zero-mean, stationary increments and such that $B(\mb{0}) = 0$, it follows that 
 its dependence structure is completely determined by its variogram.  We have
 \begin{align*}
  2c_H \gamma_B(\mb{s}) = c_H {\rm Var}(B(\mb{s}))&= 
  c_H \int_0^\infty \int_{\us^{d-1}} |1- e^{\I r \mb{s}^\top\mb{\theta}}|^2 r^{-2H - d} F_\sigma(d r, d \mb{\theta}) \\
  & = \int_{\us^{d-1}}  \Big\{ \int_0^\infty 2(1-\cos(r \mb{s}^\top\mb{\theta})) r^{-2H-1} dr \Big\} \sigma(d \mb{\theta})\\
  & = 2 \int_{\us^{d-1}} |\mb{s}^\top\mb{\theta}|^{2H} \sigma(\D \mb{\theta}) \int_0^\infty (1-\cos(u)) u^{-2H-1} d u,
 \end{align*}
 where the last relation follows from the fact that $\cos(r \mb{s}^\top\mb{\theta}) = \cos( r |\mb{s}^\top\mb{\theta}|)$, and the
 change of variables $u:= r |\mb{s}^\top\mb{\theta}|$.  In view of \eqref{e:c_H}, we obtain that the semi-variogram 
 $\gamma_B(\mb{s})$ equals the expression in \eqref{e:gamma-rep}.  This completes the proof.  The expression in \eqref{e:c_H} follows 
 from integration by parts and Relation (1.2.9) on page 17 in \cite{samorodnitsky:taqqu:1994book}, for example. 
 
  \end{proof}
 We now give a result on the tangent process of a Mat\'ern-like process. 

  \begin{proposition}
      Consider a univariate Mat\'ern-like process $Y(\mb{s})$ that has mean zero and is second-order stationary with covariance \begin{align*}
    &{\rm Cov}(Y(\ti), Y(\ti + \mb{h})) = \frac{a^{2\nu}\Gamma(\nu + \frac{d}{2})}{\pi^{\frac{d}{2}}\Gamma(\nu)} \int_0^\infty \int_{\us^{d-1}} e^{\I\langle \mb{h}, r\mb{\theta}\rangle}\left(a^2 + r^2\right)^{-\nu - \frac{d}{2}} r^{d-1} \sigma(d\mb\theta) dr.
      \end{align*}Also assume that $Y(\mb{s})$ is Gaussian and $0< \nu < 1$. Then, defining \begin{align*}
          X_\epsilon(\mb{s}) = \left(\frac{\pi^{\frac{d}{2}}\Gamma(\nu)}{a^{2\nu}\Gamma(\nu + \frac{d}{2})}\right)^{1/2}\cdot
          \frac{ (Y(\mb{s}_0 + \epsilon\mb{s}) - Y(\mb{s}_0))}{\epsilon^\nu},
      \end{align*}the tangent process is a fractional Brownian field, i.e., \begin{align*}
          &\{X_\epsilon(\mb{s})\}_{\ti \in \mathbb{R}^d} \overset{d}{\to}\left\{\int_0^\infty \int_{\us^{d-1}} \left(e^{\I\langle \ti, r\mb{\theta}\rangle} - 1\right)r^{-\nu - \frac{d}{2}  }W(dr,d\mb{\theta})\right\}_{\ti \in \mathbb{R}^d},
      \end{align*}as $\epsilon \downarrow 0$, where $W(dr, d\mb{\theta})$ is defined as in Proposition \ref{p:fBf-representation}.
  \end{proposition}

  This result will follow from the more general development in the rest of the section. 
    If one takes $\sigma(d\mb{\theta}) \propto d\mb{\theta}$, one obtains the classical univariate Mat\'ern process for $Y(\mb{s})$. 
  When taking a more general $\sigma(d\mb{\theta})$, we obtain a more flexible class of anisotropic processes. An example is given in Figure \ref{fig:anisotropic_cov}. 
  In principle, the freedom to choose $\sigma(d\mb{\theta})$ enables the Mat\'ern tangent processes to achieve a wide class of structures implied by the general from of fractional Brownian fields. 

  \begin{figure}
      \centering
      \includegraphics[width = .22\textwidth]{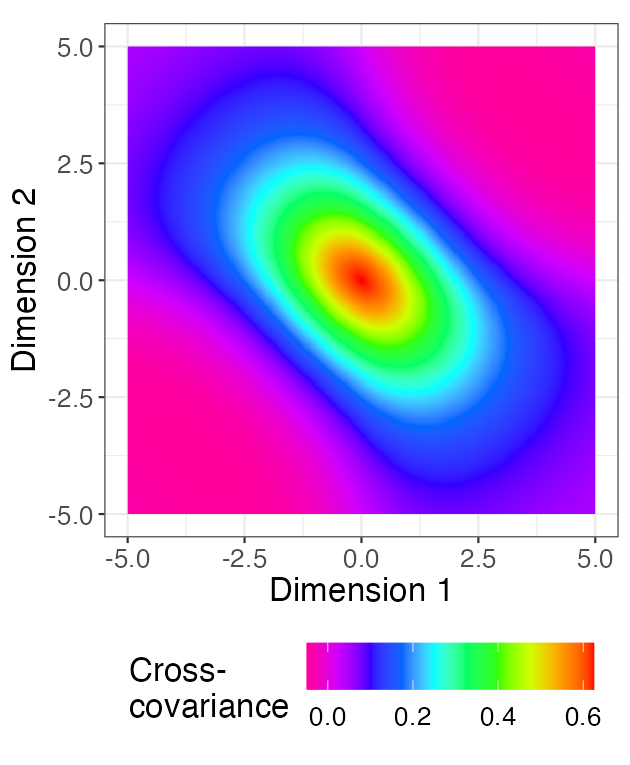}
       \includegraphics[width = .22\textwidth]{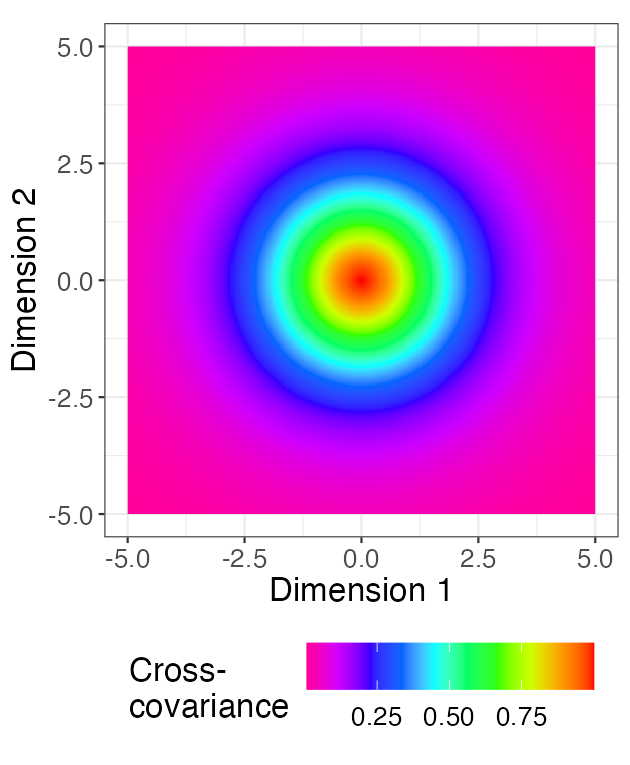}
      \caption{Mat\'ern covariance functions with $d=2$, $\nu = 1.5$,  and $a = 1$. (Left) With $\sigma(d\mb\theta) =\mathbb{I}\left(\textrm{sign}(\theta_1) =  \textrm{sign}(\theta_2)\right)+0.25 \mathbb{I}\left(\textrm{sign}(\theta_1) \neq  \textrm{sign}(\theta_2)\right) d\mb{\theta}$. (Right) With $\sigma(d\mb\theta) = d\mb{\theta}$, resulting in the classical Mat\'ern with unit variance.}
      \label{fig:anisotropic_cov}
  \end{figure}

We next turn to the multivariate case. 
In particular, we will demonstrate that the tangent processes of the multivariate Mat\'ern processes $\Y(\ti)$ are {\it operator} fractional Brownian fields. 
Due to a one-to-one correspondence between the tangent processes of multivariate Mat\'ern models and operator fractional Brownian fields (OFBFs), our proposed family of multivariate Mat\'ern models succinctly are able to realize all possible local behavior of stationary vector-valued processes. 
We will revert to thinking of $\Nu$ as a real, symmetric, and positive-definite matrix that commutes with $\amat$, rather than a diagonal matrix. 
As in the main text, we interpret $\mb{A}^{\mb{B}}$ for two matrices as $\textrm{exp}(\mb{B}\log(\mb{A}))$ with $
\log(\mb{A})$ defined for example through its Gregory series expansion reviewed in the main text, for a diagonal matrix $\mb{A}$ with positive and bounded real parts of its eigenvalues; see again \cite{cardoso_conditioning_2018} and \cite{barradas_iterated_1994}. 

\subsection{Operator fractional Brownian fields}\label{sec:OFBM}

We next describe operator fractional Brownian fields, following \cite{didier_domain_2018}. 
Let $\mb{H}$ be a real, positive-definite $p\times p$ matrix, $\mb{E}$ be a real $d \times d$ matrix, and for all $c>0$ interpret $c^{\mb{H}}$ as $\exp\{ \log(c) \mb{H}\}$, where the matrix exponent is defined as $\exp\{\mb{H}\} = \sum_{\ell=0}^\infty\mb{H}^\ell/\ell!$.
Let $\Imat_d$ be the identity matrix of dimension $d \times d$.
As in \cite{didier_domain_2018}, we consider matrices $\mb{H}$ and $\mb{E}$ with $0 < \min \Re(\textrm{eig}(\mb{H})) \leq \max \Re(\textrm{eig}(\mb{H})) < \min \Re(\textrm{eig}(\mb{E}^*)) = 1$ without loss of generality, where $\textrm{eig}(\mb{A})$ denotes the eigenvalues of the matrix $\mb{A}$ and $\mb{E}^*$ is the conjugate transpose of $\mb{E}$.
Switching the integrals to the polar coordinates of $r = \left\lVert \mb{x}\right\rVert$ and $\mb{\theta} = \mb{x}/\left\lVert \mb{x}\right\rVert \in \us^{d-1}$ where $\us^{d-1}$ is the unit sphere in $d-1$ dimensions,
Remark 3.2 of \cite{didier_domain_2018} represents an operator fractional Brownian field $\OFBF(\ti)$ as 
\begin{align}\begin{split}
\{
    \OFBF(\ti)\}_{\ti \in \mathbb{R}^d} &\overset{fdd}{=}\Bigg\{\int_0^\infty \int_{\us^{d-1}} \left(e^{\I\langle \ti, r^{\mb{E}^*}\mb{\theta}\rangle} - 1\right) r^{-\mb{H} - \frac{d}{2}\Imat_p}\rmeasure(dr,d\mb{\theta})\Bigg\}_{\ti \in \mathbb{R}^d},\label{e:didier}\end{split}
\end{align}
where $\overset{fdd}{=}$ denotes equality in finite-dimensional-distribution and $\rmeasure(dr, d\mb{\theta})$ is a 
Hermitian Gaussian random measure with \begin{align}
    \mathbb{E}\left[\rmeasure(dr, d\mb{\theta})\rmeasure(dr, d\mb{\theta})^*\right]&= r^{d-1} dr\measure(d\mb{\theta}).\label{eq:random_measure2}
\end{align}
(The slight difference in the form of the integrand in \eqref{e:didier} above and that in Remark 3.2 of \cite{didier_domain_2018}
is due to the different control measure of $\rmeasure(dr,d\mb{\theta})$ adopted in \eqref{eq:random_measure2}.)
Here, $\sd(d\mb{\theta})$ is a Hermitian measure on $\us^{d-1}$, i.e.,
such that $\sd(d\mb{\theta}) = \overline{\sd(-d\mb{\theta})}$, which follows from the fact that the process $\OFBF$ is
real.  Furthermore, we take $\sd(d\mb{\theta})$ to be finite so that $\left\lVert \sd(\us^{d-1})\right\rVert_{F} < \infty$ for the 
Frobenius norm $\lVert \cdot\rVert_{F}$. This characterizes the set of operator fractional Brownian fields; that is, this set 
contains all zero-mean, Gaussian, stationary-increment processes, which have the following $(\mb{H},\mb{E})$ operator self-similarity property.  Namely, such that for all $c>0$
\begin{equation}\label{e:OFBM-ss}
\left\{c^{\mb{H}}\OFBF(\ti)\right\}_{\ti \in \mathbb{R}^d} 
  \overset{fdd}{=} \left\{\OFBF\left(c^{\mb{E}}\ti\right)\right\}_{\ti \in \mathbb{R}^d}.
\end{equation}

\begin{remark}
In the sequel, for simplicity, we will take $\mb{E} = \Imat_d$. Further flexibility could be achieved if general operator
rescaling of the domain is adopted in the increments \eqref{e:tangent-increments} appearing in the definition of 
tangent processes. 
\end{remark}

\subsection{Multivariate tangent processes}

Suppose now $\Y = \{ \Y(\mb{s}),\ \mb{s}\in\R^d\}$ is a multivariate stochastic process taking values in $\R^p$ and consider the
following natural generalization of Definition \ref{def:tangent}. Namely, fix $\ti_0$ and consider the increment processes
\begin{align}\label{e:multivariate-increments}
    \mb{X}_\tpl(\ti) = \mb{A}_\tpl(\Y(\ti_0 + \tpl\ti) - \Y(\ti_0))
\end{align}
for $\tpl>0$ and for a $p\times p$ matrix-valued function $\mb{A}_{\epsilon}$. A non-zero stochastic process $\mb{Z}=\{\mb{Z}(\ti)\}$ is said to be a tangent process to  $\Y=\{\Y(\ti),\ \ti\in\R^d\}$ at $\ti_0$ if
\begin{align*}
    \left\{\mb{X}_\tpl(\ti)\right\}_{\ti \in \mathbb{R}^d} \overset{d}{\longrightarrow} \{\mb{Z}(\ti)\}_{\ti \in \mathbb{R}^d}
\end{align*}
as $\tpl \downarrow 0$ for some choice of $\mb{A}_\tpl$.  As in the scalar case, the limiting process $\mb{Z}$ 
will have stationary increments and it will be self-similar.  The self-similarity, however, is with respect to
linear operator rescaling in the sense of \eqref{e:OFBM-ss} with some $p\times p$ matrix-valued exponent 
$\mb{H}$ and $\mb{E} = \mb{I}_d$.

We shall establish below the tangent process of the Mat\'ern models for $\Re(\textrm{eig}(\Nu)) \subset (0,1)$ when 
\begin{equation}\label{e:A-matern}
 \mb{A}_\tpl := \tpl^{-\Nu}\mb{c}^{-1},
\end{equation}
with the non-singular normalizing constant matrix $\mb{c}$. 
To begin to characterize the resulting tangent process, we introduce Lemma \ref{lemma:tangent_process_characterization}. 


\begin{lemma}\label{lemma:tangent_process_characterization}  Suppose that the matrices $\mb{a}$ and $\mb{\nu}$ commute.  Then,
the increments process \eqref{e:multivariate-increments} with $\mb{A}_\tpl$ as in \eqref{e:A-matern},
can be characterized as follows
\begin{align*}
    &\mb{X}_{\tpl}=\{\mb{X}_{\tpl}(\ti)\}_{\ti \in \mathbb{R}^d} \\&\overset{fdd}{=}  \Bigg\{\int_0^\infty \int_{\us^{d-1}} \left(e^{\I\langle \ti, r\mb{\theta}\rangle} - 1\right)e^{\I\langle \ti_0, \tpl^{-1}r\mb{\theta}\rangle}  \left(\amat \tpl + \ang(\mb{\theta})\I r\Imat_p\right)^{-\Nu - \frac{d}{2}\Imat_p}\rmeasure(dr,d\mb{\theta})\Bigg\}_{\ti \in \mathbb{R}^d}.
\end{align*}
\end{lemma}

\begin{proof}[Proof of Lemma \ref{lemma:tangent_process_characterization}]
Throughout, we will use the fact that for a square matrix $\mb{A}$ of dimension $p$ and scalars $c$ and $b$, we have $b^c \mb{A} 
= \mb{A}b^c$.  We write the tangent process as
\begin{align*}
   \mb{X}_{\tpl}(\ti)&=  \tpl^{-\Nu}\int_0^\infty \int_{\us^{d-1}} \left(e^{\I\langle \tpl\ti, r\mb{\theta}\rangle} - 1\right)e^{\I\langle \ti_0, r\mb{\theta}\rangle} (\amat +  \ang(\mb{\theta})\I r\Imat_p)^{-\Nu - (d/2)\Imat_p}\rmeasure(dr,d\mb{\theta}).
\end{align*}
Then, in view of the construction of the random measure $\rmeasure(dr, d\mb{\theta})$, making the change of variables $u = \tpl r$, we obtain 
$$
 \Big\{ \rmeasure(dr, d\mb{\theta})\Big\}  \overset{fdd}{=} \Big\{ \epsilon^{-\frac{d-1}{2}}\tpl ^{-\frac12}\rmeasure(du, d\mb{\theta})\Big\},
 $$ 
 and hence
 \begin{align*}
    &\{\mb{X}_{\tpl}(\mb{s})\} \overset{fdd}{=}  \Big\{\tpl^{-\Nu}\int_0^\infty \int_{\us^{d-1}} \left(e^{\I\langle \ti, u\mb{\theta}\rangle} - 1\right)e^{\I\langle \ti_0, \tpl^{-1}u\mb{\theta}\rangle}\\
    &~~~~~~~~~~~~~~~~~~~~~~~~~~\times \left(\amat +  \ang(\mb{\theta})\I u\tpl^{-1}\Imat_p\right)^{-\Nu - (d/2)\Imat_p} \tpl^{-\frac{d}{2}}\rmeasure(du,d\mb{\theta})\Big\}.
\end{align*}
Next, we see that 
\begin{align*}r\left(\amat +  \ang(\mb{\theta})\I u\tpl^{-1}\Imat_p\right)^{-\Nu - (d/2)\Imat_p} &=  \left(\tpl^{-1}\left(\amat \tpl+  \ang(\mb{\theta})\I u\Imat_p\right)\right)^{-\Nu - (d/2)\Imat_p} \\
&= \tpl^{\Nu + (d/2)\Imat_p}\left(\amat \tpl +  \ang(\mb{\theta})\I u\Imat_p\right)^{-\Nu - (d/2)\Imat_p},
\end{align*}
where we used the fact that since $\mb{\nu}\mb{a} = \mb{a}\mb{\nu}$ the matrices 
$\amat \tpl+  \ang(\mb{\theta})\I u\Imat_p $ and $\Nu + (d/2)\Imat_p$ commute, so that the term 
$\epsilon^{\Nu + (d/2)\Imat_p}$ can be factored out.  Here, we used the important fact that if $\mb{A}$ and $\mb{B}$
are commuting matrices and $\epsilon>0$, we have that $(\epsilon \mb{A})^{\mb{B}} = \epsilon^{\mb{B}} \mb{A}^{\mb{B}}$, whenever the 
matrix powers are well-defined.

Simplifying terms involving $\tpl$, we see \begin{align*}
     \{\mb{X}_{\tpl}(\ti)\} &\overset{fdd}{=}  \int_0^\infty \int_{\us^{d-1}} \left(e^{\I\langle \ti, u\mb{\theta}\rangle} - 1\right) e^{\I\langle \ti_0, \tpl^{-1} u\mb{\theta}\rangle} \left(\amat \tpl +  \ang(\mb{\theta})\I u\Imat_p\right)^{-\Nu - (d/2)\Imat_p}\rmeasure(du,d\mb{\theta}),
\end{align*}the desired result.  
\end{proof}



We next characterize the limit of $\mb{X}_{\tpl}$, as $\epsilon \downarrow 0$. As anticipated from the
scalar-valued case, the resulting tangent process is an operator fractional Brownian field. 
We focus on the important special case of $\Re(\textrm{eig}(\Nu)) \subset (0,1)$.

\begin{theorem}\label{thm:tpl} Let $\Re({{\normalfont\textrm{eig}}}(\Nu)) \subset (0,1)$, $\amat$ and $\Nu$ commute,  and $\{X_\epsilon(\mb{s})\}$ 
be the increments of the multivariate Mat\'ern model introduced in \eqref{e:multivariate-increments} with $\mb{A}_\tpl$ as in
\eqref{e:A-matern}.  Then, as $\epsilon\downarrow 0$, we have
\begin{align*}
& \{\mb{X}_
{\tpl}(\ti)\}_{\ti \in \mathbb{R}^d}\overset{fdd}{\longrightarrow} \left\{\int_0^\infty \int_{\us^{d-1}} (e^{\I\langle \ti, r\mb{\theta}\rangle} - 1)r^{-\Nu - \frac{d}{2}\Imat_p}\mb{B}_{\tilde{\mb{\sd}}}(dr,d\mb{\theta})\right\}_{\ti \in \mathbb{R}^d},
\end{align*}
where $\mb{B}_{\tilde{\sd}}(dr,d\mb{\theta})$ is a complex Hermitian Gaussian random measure with orthogonal increments and
control measure $\mathbb{E}\left[\mb{B}_{\tilde{\sd}}(dr,d\mb{\theta})\mb{B}_{\tilde{\sd}}(dr,d\mb{\theta})^*\right]= r^{d-1}dr\tilde{\sd} (d\mb{\theta})$, with
\begin{align}\begin{split}\label{e:mu-tilde}
 &\tilde \mu(d\mb{\theta}):= e^{ - \I \pi \varphi(\mb{\theta})(\mb{\nu} + \frac{d}{2}\mb{I}_p)}\mu(d\mb{\theta})e^{  \I \pi \varphi(\mb{\theta})(\mb{\nu} + \frac{d}{2}\mb{I}_p)}.\end{split}
\end{align}
\end{theorem}
This form matches the characterization of OFBF in \cite{didier_domain_2018}. Therefore, the local properties of multivariate 
Mat\'ern models presented here realize a large class of OFBFs that includes covariance-asymmetric processes.

\begin{remark} Theorem \eqref{thm:tpl} covers the important case when $\Re({{\normalfont\textrm{eig}}}(\Nu)) \subset (0,1)$.
It can be shown that when $\Re(\textrm{eig}(\Nu)) \subset [k,k+1)$ for $k \in \mathbb{N}$, the increments in 
\eqref{e:multivariate-increments} have trivial limits.  Indeed, for $k\ge 1$, the process is mean-square differentiable 
and the tangent fields are simply random linear processes. In these cases, by considering higher order increments, one can obtain
non-trivial higher-order tangent fields, which turn out to be intrinsically stationary in the sense of 
\cite{matheron:1973} \citep[see, e.g.,][]{chiles2012geostatistics,shen_tangent_2022}.  This analysis can be elegantly carried out using 
Matheron's framework and the resulting $\lceil k \rceil$-th order tangent processes would be characterized 
by an operator-self-similar $\R^p$-valued intrinsic random function of order $k$. For more details, see \cite{shen_tangent_2022}.
\end{remark}

\begin{remark}
In order to achieve all possible tangent fields, one would choose $\amat = a \Imat_p$ for $a > 0$. 
In this setting, $\amat$ and $\Nu$ commute, so that $\Nu$ may be a general real matrix such that $\Re({\rm eig}(\Nu)) \subset(0,1)$.
This shows that an arbitrary operator fractional Brownian motion with general matrix-valued Hurst exponent 
$\mb{H} = \Nu$ can arise as the tangent fields of a Mat\'ern model.  Notice that $\amat$ does not appear in the process's tangent 
field as it more directly describes the decay of the covariance rather than the local properties. 
However, in applications, for modeling flexibility, having different $a_j$ values for each process is important for modeling 
different ranges of dependence in the different coordinates of multivariate spatial models. 
\end{remark}

We now turn to the proof of Theorem \ref{thm:tpl}. 

\begin{proof}[Proof of Theorem \ref{thm:tpl}]
Assuming Gaussianity, it is enough to show the convergence of the covariance function. 
From Equation (3.11) of \cite{didier_domain_2018}, we see that $\OFBF(\ti)$ with such parameters has covariance 
\begin{align}\begin{split}\label{e:Cov-OFBM}
    &\Cov(\OFBF(\ti_1), \OFBF(\ti_2))= \int_0^\infty \int_{\us^{d-1}} (e^{\I\langle \ti_1, u\mb{\theta}\rangle} - 1)(e^{-\I\langle \ti_2, u\mb{\theta}\rangle} - 1) u^{-\Nu} \tilde{\sd}(d\mb{\theta})u^{-\Nu} u^{-1} du.\end{split}
\end{align}
Returning to our expression of $\{\mb{X}_{\tpl}(\ti)\}$ in Lemma \ref{lemma:tangent_process_characterization}, we see that 
\begin{align}\label{e:Cov-tangent-thm}
\begin{split}
    &\Cov(\mb{X}_{\tpl}(\ti_1), \mb{X}_{\tpl}(\ti_2))= \int_0^\infty \int_{\us^{d-1}}(e^{\I\langle \ti_1, u\mb{\theta}\rangle} - 1)
        (e^{-\I\langle \ti_2, u\mb{\theta}\rangle} - 1)\\
        &~~~~~~~~~\times(\amat \tpl+  \ang(\mb{\theta})\I u \Imat_p)^{-\Nu - \frac{d}{2} \Imat_p}  \mathbb{E}\left[\rmeasure(du, d\mb{\theta})\rmeasure(du, d\mb{\theta})^*\right](\amat \tpl -  \ang(\mb{\theta})\I u \Imat_p)^{-\Nu - \frac{d}{2}\Imat_p}.\end{split}
\end{align}
First, we can straightforwardly evaluate the limit inside the integral: \begin{align*}
        \lim_{\tpl \downarrow 0 }&(\amat  \tpl+  \ang(\mb{\theta})\I u \Imat_p)^{-\Nu - \frac{d}{2}\Imat_p}\mathbb{E}\left[\rmeasure(du, d\mb{\theta})\rmeasure(du, d\mb{\theta})^*\right](\amat \tpl -  \ang(\mb{\theta})\I u \Imat_p)^{-\Nu - \frac{d}{2}\Imat_p} \\
    &= ( \ang(\mb{\theta})\I u \Imat_p)^{-\Nu - \frac{d}{2}\Imat_p}\mathbb{E}\left[\rmeasure(du, d\mb{\theta})\rmeasure(du, d\mb{\theta})^*\right]( -\ang(\mb{\theta})\I u\Imat_p)^{-\Nu - \frac{d}{2}\Imat_p} \\
    &= u^{-\Nu -\frac{d}{2}\Imat_p}\mathbb{E}\Big[ e^{\I\pi \ang(\mb{\theta})(-\Nu - \frac{d}{2}\Imat_p)/2}\rmeasure(du, d\mb{\theta})\rmeasure(du, d\mb{\theta})^* e^{-\I\pi \ang(\mb{\theta})(-\Nu - \frac{d}{2}\Imat_p)/2}\Big] u^{-\Nu - \frac{d}{2}\Imat_p}. 
\end{align*}
 Now, using the definition of $\tilde{\sd}$ in \eqref{e:mu-tilde},  we see that \begin{align*}
    &\mathbb{E}\Big[ e^{\I\pi\ang(\mb{\theta}) (-\nu - \frac{d}{2}\Imat_p)/2}\rmeasure(du, d\mb{\theta}) \rmeasure(du, d\mb{\theta})^* e^{-\I\pi\ang(\mb{\theta}) (-\Nu - \frac{d}{2}\Imat_p)/2}\Big] = u^{d-1} du \tilde{\sd}(d\mb{\theta}).
\end{align*}
Therefore, provided we can justify the exchange of limits and integration, we have 
\begin{align*}
    &\lim_{\tpl \downarrow 0} \Cov(\mb{X}_\tpl(\ti_1), \mb{X}_\tpl(\ti_2)) =\int_0^\infty \int_{\us^{d-1}} (e^{\I\langle\ti_1, u\mb{\theta}\rangle} - 1)
         (e^{-\I\langle\ti_2, u\mb{\theta}\rangle} - 1)u^{-\Nu}\tilde{\sd}(d\mb{\theta}) u^{-\Nu} u^{-1} du,
\end{align*}
matching the covariance of the OFBF in \eqref{e:Cov-OFBM}, which is well-defined since 
$\Re (\textrm{eig}(\Nu)) \subset(0,1)$.  

To complete the proof, we will next justify the exchange of the limit and integration by using the dominated 
convergence theorem.  Note that all integrals involving matrix-valued functions and measures are in fact component-wise integrals 
with respect to signed measures.  However, conceptually, it will be convenient to view integration with respect to a
measure $\sd(\cdot)$ taking values in the space of positive semidefinite matrices \citep[see also][]{shen_tangent_2022}.  In this
case a ($p\times p$-matrix-valued) function $(u,\mb{\theta})\mapsto G(u,\mb{\theta})$ is integrable with respect to $\sd$, provided
$$
\int_0^\infty \int_{\us^{d-1}} \|G(u,\mb{\theta})\| du \|\sd\|(d\mb{\theta}) <\infty,
$$
where $\|\mb{A}\|:={\rm trace}((\mb{A}^*\mb{A})^{1/2})$ is the trace matrix norm.  Here, 
$\|\sd\|(d\mb{\theta}) := \|\sd(d\mb{\theta})\|$ is a finite (scalar) measure on $\us^{d-1}$, defined 
by letting $\|\mu\|(B):= \|\mu(B)\|$, for all measurable $B\subset\us^{d-1}.$  To see that $\|\sd\|$ is a valid measure,
observe that for disjoint measurable sets $B,C\subset \us^{d-1},$ we have that the matrices $\sd(B)$ and $\sd(C)$
are positive semidefinite.  Hence, for $\sd(B\cup C) = \sd(B) + \sd(C)$, we have
\begin{align*}
\|\sd(B\cup C)\| &= {\rm trace}(\sd(B) + \sd(C)) \\
&= {\rm trace}(\sd(B)) + {\rm trace}(\sd(C)) \\
&= \|\sd(B)\| + \|\sd(C)\|,
\end{align*}
proving the additivity of $\|\sd(\cdot)\|$.  The $\sigma$-additivity of $\|\sd(\cdot)\|$ follows similarly from the 
$\sigma$-additivity of $\sd(\cdot)$ \citep[][]{shen_tangent_2022}. 

Now, turning to the integrand in \eqref{e:Cov-tangent-thm}, we have that 
\begin{align*}
    &\left\lVert (\amat \tpl +\ang(\mb{\theta})\I u)^{-\Nu - d/2\Imat_p}(\amat \tpl -\ang(\mb{\theta})\I u\Imat_p)^{-\Nu - d/2\Imat_p}\right\rVert_{\textrm{}}= 
    \left\lVert \left(\amat^2\tpl^2 +u^2\Imat_p\right)^{-\Nu - d/2\Imat_p}\right\rVert_{\textrm{}},
\end{align*}
which we will now show is integrable uniformly in $\epsilon>0$, using a similar approach as in the proof of 
Lemma \ref{lemma:exists}. Specifically, we want to examine 
\begin{align*}
    &\int_0^\infty \int_{\us^{d-1}}\left|\left(e^{\I\langle \ti_1, u\mb{\theta}\rangle}-1\right)\left(e^{-\I\langle \mb{s}_2, u\mb{\theta}\rangle}-1\right)\right|\left\lVert(\amat^2 \tpl^2+ u^2\Imat_p)^{-\Nu - (d/2)\Imat_p}\right\rVert_{\textrm{}} u^{d-1}du\left\lVert\sd\right\rVert_{\textrm{}}(d\mb{\theta}),  
\end{align*}
and show that the integrand is bounded (uniformly in $\epsilon>0$) by a measurable non-negative 
function $g(u,\mb{\theta})$ such that 
$$
\int_0^\infty \int_{\us^{d-1}} g(u,\mb{\theta}) du\left\lVert\sd\right\rVert_{\textrm{}}(d\mb{\theta}) <\infty.
$$
Observe that, for all $x\in \R$, we have
$$
| e^{\I x} - 1 |^2 = 2(1-\cos(x)) \le \min\{4,\,  x^2\}.
$$
Indeed, the last inequality follows by observing that the function $x\mapsto x^2 + 2(\cos(x) - 1)$ is non-negative
for all $x\in\R$.  This implies that
 \begin{align*}
    &\left|\left(e^{\I\langle \ti_1, u\mb{\theta}\rangle}-1\right)\left(e^{-\I\langle \mb{s}_2, u\mb{\theta}\rangle}-1\right)\right| \leq \min\{ 4, (\|\mb{s}_1\|^2\vee\|\mb{s}_2\|^2)\cdot u^2 \} \le 4 \wedge (C_{\mb{s}_1,\mb{s}_2} \cdot u^2),
 \end{align*}
 where $C_{\mb{s}_1,\mb{s}_2} =\|\mb{s}_1\|^2\vee\|\mb{s}_2\|^2$, $a \vee b=\textrm{max}\{a,b\}$, $a \wedge b = \textrm{min}\{a,b\}$, and we used the fact that $|\langle \mb{s}_i, u\mb{\theta}\rangle | \le \| \mb{s}_i\|\cdot |u|$ for all $\mb{\theta}\in \us^{d-1}$.
 We turn our attention to 
 \begin{align*}
        &\left\lVert \left(\amat^2 \tpl^2 + u^2 \Imat_p\right)^{-\Nu - (d/2) \Imat_p}\right\rVert_{\textrm{}} \le p\times  \left\lVert \left(\amat^2 \tpl^2 + u^2 \Imat_p\right)^{-\Nu - (d/2) \Imat_p}\right\rVert_{\textrm{op}},
    \end{align*}
    where $\|\mb{A}\|_{{\rm op}}:=\sup_{\mb{x}\not=\mb{0}} \|\mb{A}\mb{x}\|/\|\mb{x}\|$ stands for the operator matrix norm 
    induced by the
    Euclidean norm in $\R^p$. As the real parts of the eigenvalues in the exponent are non-negative, 
    \begin{align*}
        \left\lVert \left(\amat^2 \tpl^2 + u^2 \Imat_p\right)^{-\Nu - (d/2) \Imat_p}\right\rVert_{\textrm{op}} &\leq \left\lVert \left(u^2\Imat_p\right)^{-\Nu - (d/2) \Imat_p}\right\rVert_{\textrm{op}} \\
        &\leq \left\lVert u^{-2\Nu - d \Imat_p}\right\rVert_{\textrm{op}}
    \end{align*}
    \citep[see also Relation (5.20) in][]{shen:stoev:hsing:2020_extended}. 
    This step requires that the matrices $\amat$ and $\Nu$ are positive-definite and commute, so that they may be simultaneously diagonalized with the inequalities may straightforwardly take effect in this diagonalized space.  
    From here, we see that \begin{align*}
        \left\lVert u^{-2\Nu - d \Imat_p}\right\rVert_{\textrm{op}} &= u^{-2 \min\{\textrm{eig}(\Nu)\} - d}.
    \end{align*} 
    Therefore, we have the bound\begin{align*}
        &\int_0^\infty \int_{\us^{d-1}}\left|\left(e^{\I\langle \ti, u\mb{\theta}\rangle}-1\right)\left(e^{-\I\langle \mb{s}_2, u\mb{\theta}\rangle}-1\right)\right| \left\lVert(\amat^2\tpl^2 + u^2\Imat_p)^{-\Nu - (d/2)\Imat_p} \right\rVert_{\textrm{op}}u^{d-1} \left\lVert\sd\right\rVert_{F}(d\mb{\theta})du\\
        &~~~~~\leq \int_0^\infty\int_{\us^{d-1}}  4  u^{-2\min\{\textrm{eig}(\Nu)\} - 1} \left\lVert\sd\right\rVert_{F}(d\mb{\theta})du,
    \end{align*}which, with respect to $u$, is integrable for large $u$ when $-2\min\{\textrm{eig}(\Nu)\}-1 < -1$, or equivalently, $\min\{\textrm{eig}(\Nu)\} > 0$. 
    If the matrix $\Nu$ is diagonal, the restriction $\min\{\textrm{eig}(\Nu)\} > 0$ is equivalent to having all diagonal entries satisfying $\nu_k > 0$, as expected.
For small $u$, we can make the bound \begin{align*}
     &\int_0^\infty \int_{\us^{d-1}} \left|\left(e^{\I\langle \ti, u\mb{\theta}\rangle}-1\right)\left(e^{-\I\langle \mb{s}_2, u\mb{\theta}\rangle}-1\right)\right|\left\lVert (\amat^2\tpl^2 + u^2 \Imat_p)^{-\Nu - (d/2)\Imat_p}\right\rVert_{\textrm{op}} u^{d-1} \left\lVert\sd\right\rVert_{F}(d\mb{\theta})du\\
        &~~~~~\leq C_{\mb{s}_1,\mb{s}_2}\int_0^\infty \int_{\us^{d-1}} u^{1-2 \min\{\textrm{eig}(\Nu)\}} \left\lVert\sd\right\rVert_{}(d\mb{\theta})du.
\end{align*}
The is integrable with respect to $u$ for small $u$ when $1 - 2 \min\{\textrm{eig}(\Nu)\} > -1$, requiring $\min{\textrm{eig}(\Nu)} <1$, which is done by assumption. 
Thus, the integral across all $u \in (0,\infty)$ is integrable. 
Finally, by using the linearity property of the trace, it remains to establish that $\left\lVert\sd(\us^{d-1})\right\rVert_{} < \infty$, which follows from the assumed properties of $\sd(d\mb{\theta})$. 
Therefore, the integrand is dominated by an integrable function, and we may exchange the limit $\epsilon \to 0$ and integration. 
\end{proof}

\section{Schoenberg's representation of the classical Mat\'ern covariance}\label{sec:why_matern}

The advantages and features of the classical Mat\'ern covariance functions have been
widely recognized in the statistics community and beyond.  See, for example the monograph \cite{stein_interpolation_2013} 
and the recent comprehensive review \cite{porcu_matern_2023}. Here, we shall
present yet another perspective to the classical Mat\'ern models and also provide an intuitive derivation of 
of the Mat\'ern spectral density from the Mat\'ern covariance. See also \cite{cho:kim:park:yun:2017}.

The utility of the Mat\'ern covariance models
$$
{\cal M}(h;a,\nu,\sigma^2) = \sigma^2 \frac{2^{1-\nu}}{\Gamma(\nu)} (a h)^{\nu}{\cal K}_\nu(a h)
$$
has long been recognized in spatial statistics. 
The inverse range parameter $a$ and the regularity exponent $\nu$ allow for a great flexibility in modeling local dependence 
as well as the path roughness property of the underlying stochastic process $Y=\{Y(\mb{t}),\ \mb{t}\in\R^d\}$. 
Another, notable feature of the classical Mat\'ern covariance
$$
 \E[ Y(\mb{t}+\mb{h}) Y(\mb{t})] = C_Y(\|\mb{h}\|) = {\cal M}(\|\mb{h}\|; a,\nu,\sigma^2)
$$for $ \mb{h},\mb{t}\in\R^d$,
is that it is {\em isotropic} and the function $C_Y :\R \to \R$ is such that $\mb{h} \mapsto C_Y(\|\mb{h}\|)$ is 
positive semidefinite, in {\em all dimensions} $d\ge 1$, where $\|\cdot\|$ stands for the Euclidean norm in $\R^d$. This 
dimension-free positivity property is rather special and convenient since one does not need to specify (or even know!) 
the dimension $d$ of the underlying space.  
Indeed, suppose that we have a collection of points $x_i,\ i=1,\cdots,n$ equipped with a metric $\rho(x_i,x_j)$. So long as the
metric is embeddable in an Euclidean space, i.e., for some $d\ge 1$, there exist $\mb{t}_i \in \R^d$, such that 
$\rho(x_i,x_j) = \| \mb{t}_i - \mb{t}_j\|,\ i,j=1,\cdots,d$, then the matrix $C := (C_Y(\rho(x_i,x_j)))_{n\times n}$ is a valid covariance matrix.  Hence, if one has a process $\{Y(x),\ x\in {\cal X}\}$ 
over an arbitrary index set, so long as the underlying distance structure on ${\cal X}$ is embeddable in an Euclidean space, one can
model the covariance structure of $Y(x_i),\ i=1,\cdots,n$ with Mat\'ern.  This is particularly appealing in spatial statistics 
applications, where one can do kriging for example in the same way in all dimensions.

Fundamental problems on the embeddability of metric spaces in Hilbert spaces and related questions have been studies by Schoenberg.  Specifically, a result due to \cite{schoenberg:1938} characterized the class of continuous 
functions $\varphi:[0,\infty)\to \R$ such that $\mb{h}\mapsto \varphi(\|\mb{h}\|),\ \mb{h}\in\R^d$ is positive semidefinite for all 
dimensions $d\ge 1$, where $\|\cdot\|$ stands for the Euclidean norm. Namely, they are precisely variance mixtures of Gaussian 
kernels:
$$
\varphi(x) =\int_0^\infty e^{-x^2 u}\mu(du),
$$
for some finite measure $\mu$ on $[0,\infty)$ \citep[Theorem 2]{schoenberg:1938}.

This characterization calls for the natural question as to what is the measure $\mu$ in the 
Schoenberg representation of the Mat\'ern covariance.   Using the integral representation (10.32.10) of ${\cal K}_\nu$ 
in \cite{NIST:DLMF}, with the change of variables $u:= z^2 u$ therein, or equivalently Theorem 1.1 in \cite{cho:kim:park:yun:2017}
one obtains:
\begin{align}\begin{split}\label{eq:Schoenberg}
{\cal M}(z;1,\nu,1) &\equiv \frac{2^{1-\nu}}{\Gamma(\nu)} z^\nu {\cal K}_\nu(z)\\
&= \int_0^\infty e^{-z^2 u} g_\nu(u)du,\end{split}
\end{align}
where $g_\nu(u) = u^{-\nu-1} e^{-1/(4u)}/(2^{2\nu}\Gamma(\nu)),\ u>0$ is the density of an inverse Gamma probability distribution.  
Namely, $g_\nu$ is the probability density of the random variable $1/X$, where $X\sim {\rm Gamma}(\nu,1/4)$.

Let $f_{d,(1,\nu,1)}(\mb{x}) =:f_{d,\nu}(\mb{x})$ be the spectral density corresponding to the Mat\'ern 
autocovariance ${\cal M}(\|\mb{h}\| ;1,\nu,1)$, where for simplicity and without loss of generality, we will consider the {\em standard Mat\'ern} model
with $a=\sigma^2 =1$. That is,
$$
 {\cal M}(\|\mb{h}\|; 1,\nu,1) = \int_{\R^d} e^{\I \langle \mb{h},\mb{x}\rangle} f_{d,\nu}(\mb{x}) dx. 
$$ 
The following result establishes the form of the spectral density $f_{d,\nu}$, from which the proof  
that the spectral density implies the covariance readily follows with a change of variables.
\begin{proposition}\label{p:matern-spectral-density-proof}  For $\mb{x}\in\R^d$, we have that
\begin{equation}\label{eq:matern-spectral-density-proof}
f_{d,(1,\nu,1)}(\mb{x})  = \frac{\Gamma(\nu+d/2)}{\pi^{d/2} \Gamma(\nu)} \frac{1}{(1+\|\mb{x}\|^2)^{\nu+d/2}}.
\end{equation}
\end{proposition}

\begin{remark} Notice that the marginal variance of the standard Mat\'ern equals $1$.  This implies that its
spectral density integrates to $1$ and hence it can be interpreted as a probability density. In particular,
the standard Mat\'ern spectral density in \eqref{eq:matern-spectral-density-proof} is the density of the multivariate 
t-distribution in $\R^d$ with parameters $\Sigma = (2\nu)^{-1}\mathbb{I}_d$ and degrees of freedom parameter $2\nu$.
\end{remark}
\begin{proof}[Proof of Proposition \ref{p:matern-spectral-density-proof}] 
By the Fourier inversion theorem and the Schoenberg representation \eqref{eq:Schoenberg}, we have that 
\begin{align*}
    &f_{d,(1,\nu,1)}(\mb{x})  =\frac{1}{(2\pi)^d}\int_{\R^d} e^{-\I\langle \mb{x},\mb{h}\rangle } {\cal M}(\mb{x};1,\nu,1) d\mb{h} \\
    &~~~~ =\int_0^\infty \left( \frac{1}{(2\pi)^d}  \int_{\R^d}  e^{-\|\mb{x}\|^2 u} e^{-\I\langle \mb{x},\mb{h}\rangle} d\mb{h} \right) g_\nu(u) du,
\end{align*}
where in the last relation we applied Fubini's theorem.  Observe that the inner integral, with the multiplicative constant, 
is precisely the inverse Fourier transform of the characteristic function of a zero-mean Gaussian distribution in $\R^d$ with 
covariance matrix $\Sigma = 2u\cdot \mathbb{I}_d$. Therefore,
\begin{align}
    f_{d,(1,\nu,1)}(\mb{x})  &=  \frac{1}{2^{2\nu} \Gamma(\nu)}
    \int_0^\infty \frac{1}{(4\pi u)^{d/2}} e^{-\|\mb{x}\|^2/(4u)} \nonumber u^{-\nu -1} e^{ -1/(4u) } du\nonumber\\
    & = \frac{1}{\pi^{d/2} 2^{2\nu+d} \Gamma(\nu)} \int_0^\infty e^{-(\|\mb{x}\|^2+1)/(4u)}  \nonumber u^{-(d/2+\nu+1)} du \nonumber\\
    & = \frac{(1+\|\mb{x}\|^2)^{-\nu-\frac{d}{2}}}{\pi^{d/2} 2^{2\nu} \Gamma(\nu)} \int_0^\infty e^{-\frac{1}{4t}} t^{-(\nu+\frac{d}{2}+1)} dt
     \label{eq:u-to-t} \\
    & = \frac{\Gamma(\nu+d/2) 2^{2\nu+d}}{\pi^{d/2} 2^{2\nu+d} \Gamma(\nu)}\cdot \frac{1}{(1+\|\mb{x}\|^2)^{\nu+\frac{d}{2}}},
    \label{eq:inverse-Gamma-consideration}
\end{align}
where Relation \eqref{eq:u-to-t} follows by making the change of variables $t:= u/(1+\|\mb{x}\|)$ and \eqref{eq:inverse-Gamma-consideration},
from the observation that the integrand in \eqref{eq:u-to-t} is up to a constant, the density of an inverse Gamma law with 
parameters $\nu+d/2$ and $1/4$.  Relation \eqref{eq:inverse-Gamma-consideration} equals \eqref{eq:matern-spectral-density-proof}.
\end{proof}

One may wonder if there is a similar representation for the odd cross-covariances presented in Section 3.2. 
We address that question with the following.

\begin{lemma}
Suppose that $d=1$ and let \begin{align*}
    C_{jk}(h) = -\mathcal{H}\left[\frac{2^{1-\nu}}{\Gamma(\nu)} (\cdot)^{\nu} \mathcal{K}_\nu(\cdot)\right](h),
\end{align*}where ${\cal H}$ stands for the Hilbert transform. 
Then, we can represent \begin{align*}
    C_{jk}(h) &= \frac{2}{\pi^{1/2}}\int_0^\infty F\left(\sqrt{u}h\right)g_\nu(u) du,
\end{align*}where $F(\cdot)$ is Dawson's integral, and again $g_\nu(u) = u^{-\nu-1} e^{-1/(4u)}/(2^{2\nu}\Gamma(\nu)),\ u>0$ is the density of an inverse Gamma probability distribution.  
\end{lemma}

\begin{proof}

Through definition of the Hilbert transform and use of \eqref{eq:Schoenberg}, \begin{align*}
    C_{jk}(h) &= -\mathcal{H}\left[\frac{2^{1-\nu}}{\Gamma(\nu)} (\cdot)^{\nu} \mathcal{K}_\nu(\cdot)\right](h) \\
    &=- \frac{1}{\pi}\int_{-\infty}^\infty \frac{1}{h - z} \frac{2^{1-\nu}}{\Gamma(\nu)} (z)^{\nu} \mathcal{K}_\nu(z) dz \\
    &= - \frac{1}{\pi}\int_{-\infty}^\infty \frac{1}{h - z}\int_{0}^\infty e^{-z^2 u}\frac{u^{-\nu - 1}e^{-1/(4u)}}{2^{2\nu}\Gamma(\nu)} du dz.
\end{align*}
For $h \neq 0$, we can apply Fubini's theorem to give \begin{align*}
      C_{jk}(h) &=  \frac{1}{\pi}\int_{0}^\infty \int_{-\infty}^\infty \frac{1}{h - z}e^{-z^2 u}dz g_\nu(u) du.
\end{align*}
Making the change of variables $w = \sqrt{u}z$ and $dw = \sqrt{u} dz$, one obtains \begin{align*}
     \frac{1}{\pi}\int_{-\infty}^\infty \frac{1}{h - z}e^{-z^2 u}dz&=   \frac{1}{\pi}\int_{-\infty}^\infty \frac{1}{\sqrt{u}h - w}e^{-w^2}dw \\
     &= \frac{2}{\sqrt{\pi}}F\left(\sqrt{u} h \right),
\end{align*}where the final line follows from the Hilbert transform of $e^{-w^2}$, which is $2F(w)/\sqrt{\pi}$ \citep{king2009hilbert}.

This gives \begin{align*}
    C_{jk}(h) &=  \frac{2}{\pi^{1/2}}\int_{0}^\infty F\left(\sqrt{u} h \right)g_\nu(u) du,
\end{align*}the final result.

For $h=0$, we obtain $C_{jk}(h)= \pi^{-1} \int_{-\infty}^\infty 2^{1-\nu}\Gamma(\nu)^{-1} z^{\nu-1} \mathcal{K}_{\nu}(z) dz$. The integrand is now of the form $z^{-1} \cdot h(z)$, where $h(z)= \pi^{-1} 2^{1-\nu}\Gamma(\nu)^{-1} z^\nu \mathcal{K}_{\nu}(z)$ is an integrable even function.
Therefore, the integrand is odd and thus $C_{jk}(h) = 0$ as expected. 

\end{proof}

\section{Numerical comparison}

We consider a numerical comparison for computation of cross covariance functions. 
We take $d=1$, $\nu_1 = \nu_2 = 0.5$, $a_1 = a_2 = 1$, and $\sigma_{12} = 1$, so that $C_{12}(h) = e^{-|h|}$.
We evaluate the cross-covariance at $100$ evenly-spaced points between $-3$ and $3$, and repeat the computation $500$ times.
We compare the following in \texttt{R}: 
\begin{itemize}
    \item \texttt{base::exp}: Computation of $e^{-|h|}$ directly; our main point of comparison.
    \item \texttt{base::besselK}: Computation of Mat\'ern covariance.
    \item \texttt{fftwtools::fftw\_c2c}, ($2^{10}$): Computation of fast Fourier transform based on a grid of size $2^{10}$ between $-10$ and $10$, then interpolated using \texttt{stats::approx}. This uses \cite{frigo_fast_1999}.
    \item \texttt{fftwtools::fftw\_c2c}, ($2^{12}$): Computation of fast Fourier transform based on a grid of size $2^{12}$ between $-10$ and $10$, then interpolated using \texttt{stats::approx}. This uses \cite{frigo_fast_1999}.
    \item \texttt{fftwtools::fftw\_c2c}, ($2^{14}$): Computation of fast Fourier transform based on a grid of size $2^{14}$ between $-10$ and $10$, then interpolated using \texttt{stats::approx}. This uses \cite{frigo_fast_1999}.
    \item \texttt{fAsianOptions::whittakerW}: Computation using the confluent hypergeometric function denoted $\mathcal{U}(\cdot, \cdot, \cdot)$ using \cite{fAsianOptions}. 
    For this, we perturb $\nu_k = 0.5 + 10^{-9}$ to avoid numerical stability issues when $\nu_j = \nu_k$.
    \item \texttt{stats::integrate}: Integrating the spectral density with a lower bound of $10^{-2}$ and an upper bound of $10^2$.
\end{itemize}

We compare the time of computation and the mean-squared-error (compared with \texttt{base::exp}), presented in Table \ref{tab:numerical}.
The fast Fourier approach can perform similarly in time to computations involving the Mat\'ern covariance with a loose approximation grid $2^{10} = 1024$, while it can perform as well as and faster than \texttt{fAsianOptions::whittakerW} when using a better approximation $2^{14}$. 

\begin{table}[ht]
    \centering
    \begin{tabular}{|c|c|c|c|}\hline
        Main function & Elapsed & Exponential comp & MSE diff, \texttt{base::exp} \\ \hline\hline
        \texttt{base::exp} & 0.020  & 1 & 0\\ \hline
        \texttt{base::besselK} & 0.067 & 3.35 & 6e-33\\ \hline
        \texttt{fftwtools:::fftw\_c2c}, ($2^{10}$) & 0.074    & 3.895 & 7e-10\\  \hline
        \texttt{fftwtools:::fftw\_c2c}, ($2^{12}$) & 0.265 & 13.25 & 7e-13\\  \hline
        \texttt{fftwtools:::fftw\_c2c}, ($2^{14}$)&  0.826 & 41.30 & 3e-15 \\  \hline
        \texttt{fAsianOptions::whittakerW}  & 2.747 & 137.35 & 2e-14 \\ \hline
        \texttt{stats::integrate} & 14.800 & 752.25 & 3e-08\\  \hline
    \end{tabular}
    \caption{Comparison of numerical computation. The columns show the approach for computation, the elapsed time to compute at $100$ points $500$ times, the ratio of the time the method took with the time \texttt{base::exp} took, and the mean-squared-difference between \texttt{base::exp} and the method's computation at the points.}
    \label{tab:numerical}
\end{table}

\bibliographystyle{apalike}

\bibliography{main}